\documentclass[a4paper, 4pt]{elsarticle}
\usepackage{fullpage,graphicx,subfigure,mathpazo,color}
\usepackage{amsmath,amscd,tikz,mathrsfs}
\usepackage[normalem]{ulem}
\usepackage{amsmath}
\usepackage{epstopdf}
\usepackage{tikz}
\usepackage{setspace}
\usepackage{lipsum}
\usepackage{amssymb}
\usepackage{graphicx}
\usepackage{amsfonts}
\usepackage{geometry}
\usepackage{ulem}
\usepackage{epsfig}
\usepackage{CJK}
\usepackage{amsthm}
\usepackage{setspace}
\usepackage{graphicx}
\usepackage{epstopdf}
\usepackage{subeqnarray}
\usepackage{mathrsfs}
\usepackage{bbding}
\usepackage{cases}
\usepackage{subfigure}
\usepackage{mathrsfs}
\usepackage{float}
\usepackage{tikz}
\usepackage{geometry}
\usepackage{lineno,hyperref}
\usepackage{fullpage,graphicx,subfigure,mathpazo,color}
\usepackage{amsmath,amscd,tikz,mathrsfs}
\usepackage[normalem]{ulem}
\usepackage{amsmath}
\usepackage{epstopdf}
\usepackage{tikz}
\usepackage{setspace}
\modulolinenumbers[5]
\journal{Physica D}

\newcommand{\ii}{\mathrm{i}}
\newcommand{\ee}{\mathrm{e}}
\newcommand{\dd}{\mathrm{d}}

\newcommand{\M}{\mathbf{M}}
\newcommand{\V}{\mathbf{V}}
\newcommand{\oo}{\mathcal{O}}

\newtheorem{rhp}{Riemann-Hilbert Problem}
\newtheorem{theorem}{Theorem}
\newtheorem{lemma}{Lemma}
\newtheorem{prop}{Proposition}

\newtheorem{remark}{Remark}
\newtheorem{assumption}{Assumption}[section]
\usepackage{graphicx}      
\usepackage{titlesec}

\titleformat{\section}{\centering\LARGE\bfseries}{\thesection}{1em}{}
\titleformat{\subsection}{\Large\bfseries}{\thesubsection}{1em}{}
\begin{document}

\title{Long-time asymptotics of the coupled nonlinear Schr\"odinger equation 
	in a weighted Sobolev space}
\author[]{Yubin Huang$^{1}$}
\author[]{Liming Ling$^{1}$}
\author[]{Xiaoen Zhang$^{2}$\corref{mycorrespondingauthor}}
\cortext[mycorrespondingauthor]{Corresponding author}
\ead{zhangxiaoen@sdust.edu.cn}
\address{1 School of Mathematics, South China University of Technology, Guangzhou, China, 510641\\
2 College of Mathematics and Systems Science, Shandong University of Science and Technology, Qingdao, China, 266590}

\begin{abstract}
	We study the Cauchy problem for the focusing coupled nonlinear Schr$\ddot{\rm o}$dinger (CNLS) equation with initial data $\mathbf{q}_0$ 
	lying in the weighted Sobolev space and the scattering data having $n$ simple zeros. 
	Based on the corresponding $3\times3$ matrix spectral problem, we deduce the Riemann-Hilbert problem (RHP)
	for CNLS equation through inverse scattering transform. We remove discrete spectrum of initial RHP using Darboux transformations.
	By applying the nonlinear steepest-descent method for RHP introduced by Deift and Zhou, 
	we compute the long-time asymptotic 
	expansion of the solution $\mathbf{q}(x,t)$ to an (optimal) residual error of order $\mathcal{O}\left(t^{-3 / 4+1/(2p)}\right)$ where 
	$2\le p<\infty$. The leading order term in this expansion is a multi-soliton whose parameters are modulated by soliton-soliton and 
	soliton-radiation interactions. Our work strengthens and extends the earlier work regarding long-time 
	asymptotics for solutions of the nonlinear Schr\"odinger equation with a delta potential and even initial data by Deift and Park.
\end{abstract}

\maketitle
\section{Introduction}
As a cornerstone of nonlinear mathematical physics, the nonlinear Schr$\ddot{\rm o}$dinger (NLS) 
equation plays a crucial role in studying nonlinear optical fibers and Bose-Einstein condensation (BEC). 
Its relevance extends to several other fields, such as fluid mechanics, plasma physics, 
and even financial models \cite{Chiao1964, Zakharov1968, YAN2011}. When interactions involve more 
than two components, the scalar NLS equation becomes insufficient to describe dynamics of the system. 
Consequently, coupled systems have garnered significant attention. Compared to the scalar equation, 
coupled models exhibit more diverse dynamics, including the inelastic collision of bright-bright 
solitons for one component, bright-dark solitons, vector rogue waves, breathers, and non-degenerate 
or multi-hump fundamental solitons \cite{Guo2011, Mu2015}. Therefore, studying multi-component systems 
and exploring their richer dynamics is important. In this paper, we investigate the long-time 
asymptotic behavior of the focusing coupled nonlinear Schr$\ddot{\rm o}$dinger (CNLS) equation with initial data 
lying in the weighted Sobolev space.

The long time behavior of the defocusing NLS equation \eqref{NLS-eq}, 
with $\lambda=-1$, has been thoroughly studied \cite{zakharov1976asymptotic,deift1994long,Deift1993_long}.
\begin{equation}\label{NLS-eq}
	\ii q_t+\frac12 q_{xx}+\lambda|q|^2q=0,\quad\lambda=\pm1,\quad q(x,0)=q_0(x).
\end{equation}
In the defocusing case, it is founded that as $t\to+\infty$, 
\begin{equation}
	q(x,t)=t^{-\frac12}\alpha(z_0)\ee^{\ii x^2/(2t)-\ii v(z_0)\log(4t)}+\varepsilon(x,t),
\end{equation}
where 
\begin{equation}
	v(z)=-\frac{1}{2\pi}\log(1-|r(z)|^2),\quad |\alpha(z)|^2=v(z)^2,
\end{equation}
and 
\begin{equation}
	\arg \left(\alpha(z)\right)=\frac{1}{\pi}\int_{-\infty}^{z}\log(z-s)\dd(\log(1-|r(z)|^2))+
	\frac{\pi}{4}+\arg\left(\Gamma(\ii v(z))\right)-\arg \left(r(z)\right).
\end{equation}
Here $z_0=-\frac{x}{2t}$ is the stationary point of the phase function for NLS equation, 
$\Gamma$ is the Gamma function, and $r$ is the reflection coefficient. 
The error term $\varepsilon(x,t)$ depends on smoothness and decay 
assumptions of the potential function $q_0$. The leading term was first obtained without estimates of 
the error term in \cite{zakharov1976asymptotic}. In 1993, Deift and Zhou introduced
the nonlinear steepest descent method (NSD) for oscillatory Riemann-Hilbert problems (RHPs) \cite{deift1993steepest}.
Later, in 1994, Deift and Zhou studied the long-time behavior of Cauchy problem of NLS equation by 
generalizing NSD and demonstrated that if $q_0$ lies in Schwartz space which implies 
a high degree of smoothness and decay, then $\varepsilon(x,t)=\oo(t^{-1}\log t)$ 
in \cite{deift1994long,deift1994casestudy}. 
The error term $\varepsilon(x,t)$ was subsequently improved to $\oo(t^{-(1/2+\kappa)})$ 
by improved nonlinear steepest descent method (INSD) \cite{deift2002long}, 
for any $0<\kappa<1/4$, under the significantly weaker assumption that $q_0$
belongs to the weighted Sobolev space $H^{1,1}(\mathbb{R})$, where 
\begin{equation}
	H^{i,j}(\mathbb{R})=\{f(x)\in L^2(\mathbb{R}):f^{(i)}(x),x^jf(x)\in L^2(\mathbb{R})\}.
\end{equation}
In the focusing case, Deift and Park \cite{deift2011long} obtained the asymptotic behavior of the 
stationary 1-soliton solution of the focusing nonlinear Schr$\ddot{\rm o}$dinger (fNLS) equation, 
establishing the same estimate of the error term, $\varepsilon(x,t)=\oo(t^{-(1/2+\kappa)})$ for any $0<\kappa<1/4$.

Recently, McLaughlin and Miller \cite{McLaughlin2006,McLaughlin2008} developed the 
$\bar{\partial}$ steepest descent method for asymptotic 
analysis of RHPs based on $\bar{\partial}$ problems, 
rather than relying on the asymptotic analysis of singular integrals along contours. 
The original RHP is transformed into a $\bar{\partial}$-RHP, which is factorized into the product 
of solutions for a solvable RHP and a pure $\bar{\partial}$ problem. 
In 2008, McLaughlin et al. \cite{dieng2008long} improved the error estimates for 
the long-time asymptotic behavior of solutions to the defocusing NLS equation to $\oo(t^{-3/4})$, 
without imposing additional restrictions on the initial data. 
In 2017, McLaughlin et al. \cite{McLaughlin_2018} demonstrated that the error remains $\oo(t^{-3/4})$ 
for the long-time asymptotic expansion of the solution to the focusing NLS equation 
in any fixed space-time cone.

Besides the classical NLS equation, 
NSD have been widely applied to study the long-time asymptotic behavior of other integrable nonlinear evolution equations, 
such as CNLS equation based on $3\times3$ Lax pairs \cite{Geng_Liu_2017}, 
CNLS equation based on $4\times4$ Lax pairs \cite{Geng2022CNLS}, the derivative nonlinear Schr\"odinger (DNLS) equation 
\cite{GengLiu2024DNLS}, the short pulse equation 
\cite{Geng2024short-pulse,Geng2024short-pulse2}, the Boussinesq equation with initial data belonging to the Schwartz class\cite{charlier2023good,charlier2024boussinesq,charlier2024boussinesq-1},  the Boussinesq equation involving
soliton resolution conjecture\cite{charlier2024soliton} and others. 
The $\bar{\partial}$ steepest descent method have been also widely applied to 
study the long-time asymptotic behavior of other integrable nonlinear evolution equations, such as 
the modified Korteweg-de Vries (mKdV) equation \cite{Fan2023mKdV,Fan2023nonload-mKdV,chen2021soliton}, 
the Novikov equation \cite{Fan2023Novikov}, the Hunter-Saxton equation \cite{Fan2024Hunter-Saxton}, 
DNLS equation \cite{Fan2022DNLS}, the defocusing Ablowitz-Ladik system \cite{Fanengui2024Ablowitz-Ladik}, 
the Sasa-Satsuma equation \cite{Fan2022Sasa-Satsuma}, the Camassa-Holm equation \cite{Fan2024Camassa-Holm}, and others.

However, to the best of our knowledge, the long-time asymptotic behavior of the solution to the focusing CNLS equation, with the potential function 
$\mathbf{q}_0$ lying in the weighted Sobolev space, has not yet been addressed in the existing literature. 
It is worth exploring whether the error estimate can be maintained at $\oo(t^{-(1/2+\kappa)}),\,0<\kappa<1/4$,
or even at $\oo(t^{-3/4})$.

In this work, we apply INSD to the inverse scattering transform (IST) for the 
focusing CNLS equation \eqref{CNLS} to obtain the corresponding long-time asymptotic behavior of solutions. We assume 
that the initial data $\mathbf{q}_0$ for 
equation \eqref{CNLS} lies in the weighted Sobolev space 
$H^{1,2}(\mathbb{R})$ and 
that the ZS-AKNS operator associated with equation 
\eqref{CNLS} possesses finite simple discrete spectrum anywhere in 
$\mathbb{C}\setminus\mathbb{R}$. Consequently, the long-time behavior of solutions to the focusing CNLS equation 
is necessarily more intricate than in the defocusing case due to the presence of solitons corresponding to 
the discrete spectrum. 
In the focusing case, the real axis represents the continuous spectrum of the ZS-AKNS operator, 
along which we define a reflection coefficient $\mathbf{R}:\mathbb{R}\to\mathbb{C}^2$.
This is the classical scattering map $S:\mathbf{q}_0\to D$ for the CNLS equation, where $D$
consists of scattering data. 
In 2019, it was shown \cite{Liu_2019} that 
the scattering map $S$ of the IST for the focusing CNLS equation is a bijective 
(in fact, bi-Lipschitz) map from $H^{i,j}(\mathbb{R})$ to 
$H^{j,i}(\mathbb{R})$. This provides an important foundation 
for our subsequent proofs. 
The key idea of NSD is to deform the jump matrix $\V(\lambda)$ off the real axis.
By employing the Taylor expansion of reflection coefficient 
$\mathbf{R}(\lambda)$ in Schwartz space, we can decompose 
$\mathbf{R}(\lambda)$ into an analytic part and a non-analytic remainder, 
both of which admit suitable estimates on the corresponding contours.
In this paper, since the reflection coefficient 
$\mathbf{R}(\lambda)\in H^{2,1}(\mathbb{R})$ lacks a high degree of smoothness 
and decay, we utilize $[\mathbf{R}]$ as defined in \cite[pp.1041]{deift2011long} to perform RHP transformation 
instead of decomposing reflection coefficient into an analytic part and 
a non-analytic remainder.

In the focusing case, we use Darboux transformation to construct 
RHP \ref{initial-rhp-without-role} without residual conditions. It is necessary to estimate 
$\M(\lambda;x,t)$ at $\lambda=\lambda_i$ additionally in recovering the solution of equation \eqref{CNLS}. 
A function $\delta$ is typically introduced in the context of classical NSD 
for the NLS equation and can be explicitly solved using the Plemelj formula. 
But the function $\pmb{\delta}$ related to the CNLS equation satisfies a $2 \times 2$ matrix RHP without an explicit expression, 
which presents challenges in estimating $\pmb{\delta}$. 
To address the issues related to the CNLS equation with initial data 
$\mathbf{q}(x,t)\in H^{1,2}(\mathbb{R})$, 
we construct a non-homogeneous matrix RHP to estimate $\pmb{\delta}$ 
through the properties of the available function $\det \pmb{\delta}$, 
as presented in Lemma \ref{estimate-R-delta}. 
This approach allows us to demonstrate 
that the solvable model RHP \ref{rhp-M4}, which is crucial for proving Theorem \ref{main-result}, 
serves as a sufficiently accurate approximation to 
RHP \ref{rhp-M3}.

The framework presented in this paper is applicable to other coupled integrable 
systems, such as vector DNLS equations and vector mKdV equations. 
A key point to emphasize is that for systems with multiple stationary points, 
the constructed approximant $[\mathbf{R}](\lambda)$ must match reflection coefficient 
$\mathbf{R}(\lambda)$ at all stationary points. Furthermore, for non-generic scattering data 
containing higher-order zeros, the initial Riemann-Hilbert problem with higher-order 
singularities can be also reduced to a pole-free one via the construction of 
an appropriate generalized Darboux matrix 
\cite{Bilman-2020} or the limit technique, 
but which requires the further analysis to get 
their long-time asymptotics.

Our main result can be stated as follows:
\begin{theorem}\label{main-result}
	Let $\mathbf{q}(x,t)$, $x\in\mathbb{R}$, $t\ge0$, solve CNLS equation \eqref{CNLS} 
	with $\mathbf{q}_0=\mathbf{q}(x,0)$ satisfying Assumption \ref{assumption-q0}. Then for $2\le p<\infty$, 
	as $t\rightarrow\infty$,
	\begin{equation}\label{expansion-q}
		\mathbf{q}(x,t)
		=\frac{\ee^{\frac{\pi\kappa_{\xi}}{2}-
				\frac{\pi \ii}{4}}\eta^2\kappa_{\xi}
				\Gamma(\ii\kappa_{\xi})\pmb{A}(\xi)}{\sqrt{2\pi t}}
		-2\mathbf{X}_1\mathbf{G}^{-1}
		(\mathbf{X}_2)^\dagger+\oo(t^{-\frac{3}{4}+\frac{1}{2p}}),
	\end{equation}
	where $\xi=-\frac{x}{2t}$, $\Gamma(\cdot)$ is the Gamma function, the vector-valued 
	function $\mathbf{R}(\lambda)$ is defined by equation 
	\eqref{def-R-no-soliton}, the constant $\pmb{A}(\xi)$ 
	is given by equation \eqref{def-A-xi} and 
	\begin{equation*}
		\begin{aligned}
			&\eta=(2\sqrt{t})^{-\ii\kappa_{\xi}}\ee^{\ii t\xi^2}\ee^{-\beta(\xi,\xi)-\ii\kappa_{\xi}},\quad
			\mathbf{X}=\begin{bmatrix}
				\pmb{\varphi}_1,...,\pmb{\varphi}_N
			\end{bmatrix},\\
			&\kappa_\xi
			=-\frac{1}{2\pi}\ln(1+|\mathbf{R}(\xi)|^2),\quad
			\pmb{\varphi}_{i}=\left(\mathbb{I}_3+
			t^{-1/2}\tilde{\mathbf{P}}_i\right)\pmb{\Delta}(\lambda_i)\ee^{-\mathrm{i} 
				\lambda_{i}\left(x+ \lambda_{i} t\right) \pmb{\Lambda}_{3}} \mathbf{v}_{i},\\
			&\mathbf{v}_{i}=\begin{bmatrix}
				\prod_{j=i+1}^{N}\frac{\lambda_i-\lambda_j^*}{\lambda_i-\lambda_j},
				\frac{\mathbf{c}_i^{\top}}{\lambda_i-\lambda_i^*}
				\prod_{j=1}^{i-1}\frac{\lambda_i-\lambda_j}{\lambda_i-\lambda_j^*}
			\end{bmatrix}^{\top},\quad
			\mathbf{G}=\left(\frac{\pmb{\varphi}_j^\dagger \pmb{\varphi}_i}
			{\lambda_i-\lambda_j^*}\right)_{1\le j,i\le n},\\
			&\beta(\xi,\xi)
			=\int_{-\infty}^{\xi-1}\frac{\ln(1+|\mathbf{R}(s)|^2)}{s-\xi}\frac{ds}{2\pi \ii}\\
			&+\int_{\xi-1}^{\xi}
			\frac{\ln(1+|\mathbf{R}(s)|^2)-\ln(1+|\mathbf{R}(\xi)|^2)
				(s-\xi+1)}{s-\xi}\frac{ds}{2\pi i},
		\end{aligned}
	\end{equation*}
	and $\pmb{\Delta}(\cdot)$ is defined by equation \eqref{def-Delta}, $\tilde{\mathbf{P}}_i$ is 
	defined by equation \eqref{def-tilde-P}. Here $\mathbf{X}_1$ and $\mathbf{X}_2$ represent the first row and 
	the last two rows of $\mathbf{X}$, respectively. 
\end{theorem}  

\begin{remark}
    The parameter $p$ in the error estimate 
	originates from the $L^p$-estimates of the Cauchy projection operator 
	used in the Riemann-Hilbert problem within the nonlinear steepest-descent method. 
	In particular, Lemma \ref{bdd-M1-M2}, which controls the difference 
	between $\mathbf{M}^{(1)}(\lambda)$ and $\mathbf{M}^{(2)}(\lambda)$, 
	is valid only for $1 \le p < \infty$. As a consequence, the error term 
	in the asymptotic expansion \eqref{expansion-q} takes the 
	form $\mathcal{O}\bigl(t^{-\frac{3}{4}+\frac{1}{2p}}\bigr)$ with $1 \le p < \infty$. 
	This error estimate is consistent with the result for the scalar NLS equation 
	in \cite{deift2002long}, though it does not reach the sharper order 
	$\mathcal{O}\bigl(t^{-\frac{3}{4}}\bigr)$ that is available 
	via the $\bar{\partial}$-method, as shown in \cite{McLaughlin_2018}.
\end{remark}

If $N=1$, then the set $\mathcal{Z}$ consists of only a single point $\lambda_1=\xi_1+\ii\eta_1$. 
Hence, as $t\to+\infty$, the asymptotic expansion \eqref{expansion-q} can be reduced to
\begin{equation}
	\mathbf{q}(x,t)
	=2\eta_1\operatorname{sech}(2\eta_1(x-x_1+2t\xi_1))\ee^{-2\ii(\xi_1x-t(\eta_1^2-\xi_1^2))}\ee^{\ii\phi_1}\mathbf{c}_1^\dagger\pmb{\delta}(\lambda_1^*)+\oo(t^{-\frac{1}{2}}),
\end{equation}
where 
\begin{equation*}
	\begin{aligned}
		&x_1=\frac{1}{2}\int_{-\infty}^{-x/2t}\frac{\kappa(s)}{(s-\xi_1)^2-\eta_1^2}\dd s+
		\frac{1}{2\eta_1}\ln|\pmb{\delta}^{-1}(\lambda_1)|,\\
		&\phi_1=-\ii\ln(2\eta_1|\pmb{\delta}^{-1}(\lambda_1)|)
		+\ii\eta_1\int_{-\infty}^{-x/2t}\frac{\kappa(s)}{(s-\xi_1)^2-\eta_1^2}\dd s
		-\int_{-\infty}^{-x/2t}\frac{\kappa(s)(s-\xi_1+\ii\eta_1)}{(s-\xi_1)^2-\eta_1^2}\dd s.
	\end{aligned}
\end{equation*}

This article is organized as follows. In Section 2, we revisit the direct scattering transform for 
the CNLS equation \eqref{CNLS}, discussing the Lax pair, analyticity, symmetries, and asymptotics of Jost solutions 
and the reflection coefficient. We then perform the IST to derive the reconstruction 
formula by constructing the corresponding RHP. In Section 3, we initially remove all poles from the initial RHP 
\ref{initial-rhp-with-pole} using the Darboux transform. Subsequently, we utilize the corresponding RHP, which 
can be deformed into a solvable RHP model \ref{solvable-RHP-model}, to approximate the initial RHP 
\ref{remove-pole-of-initial-rhp} without poles and analyze the long-time asymptotic behavior of solutions for 
the CNLS equation \eqref{CNLS}.

\begin{remark}
	The superscript $^{*}$ and $^{\dagger}$ in a matrix denote the element-wise complex conjugate and the Hermitian conjugate respectively.
	The superscript $^{\top}$ denotes matrix tranpose.
	$\mathbb{I}_{n}$ indicates $n\times n$ identity matrix. Let a $3\times 3$ matrix A 
	be rewritten as a block form
	\begin{equation*}
		A=\begin{bmatrix}A_{11}&&A_{12}\\A_{21}&&A_{22}\end{bmatrix},
	\end{equation*}
	where $A_{11}$ is scalar.
\end{remark}

\section{Results of scattering theory for focusing CNLS}
We briefly review the IST \cite{manakov1974theory,novikov1984theory} for the CNLS equation.
The focusing CNLS equation reads in the vector form:
\begin{equation}\label{CNLS}
	\ii\mathbf{q}_{t}+\frac{1}{2}\mathbf{q}_{xx}+\mathbf{q}\mathbf{q}^{\dagger}\mathbf{q}=0,
	\quad\mathbf{q}(x,t)=\begin{bmatrix}
		q_{1}(x,t),q_{2}(x,t)
	\end{bmatrix},
\end{equation}
\noindent
where $q_{1}(x,t)$ and $q_{2}(x,t)$ are the wave envelopes, which were first derived by Manakov in 1974 to describe electrical propagation 
\cite{manakov1974theory}, and are therefore known as the Manakov model. 
Subsequently, the focusing CNLS equation \eqref{CNLS} has garnered significant attention 
across various fields, ranging from Bose-Einstein condensates (BEC) \cite{busch2001dark} 
to optical fibers \cite{agrawal2000nonlinear} and biophysics \cite{scott1984launching}. 
From a mathematical perspective, the CNLS equation \eqref{CNLS} is an integrable equation 
that admits the following Lax pair:
\begin{equation}\label{lax-pair-Phi}
	\begin{array}{ll}
		\left(\partial_{x}-\mathcal{L}\right) \mathbf{\mathbf{\mathbf{\Phi}}}=0, 
		& \mathcal{L}=\ii\left(-\lambda\mathbf{\Lambda}_3+\mathbf{Q}\right), \\
		\left(\partial_{t}-\mathcal{B}\right) \mathbf{\mathbf{\Phi}}=0, & \mathcal{B}=\ii\lambda\left(-\lambda\mathbf{\Lambda}_3+\mathbf{Q}\right)
		+\frac{1}{2}\left(\mathbf{Q}_{x}+\ii\mathbf{Q}^{2}\right)\mathbf{\Lambda}_3,
	\end{array}
\end{equation}
\noindent 
where $\lambda \in \mathbb{C}$ is the spectral parameter and
\begin{equation*}
	\mathbf{\Lambda}_3=\mathrm{diag}\left(1,-1,-1\right)\quad,\quad 
	\mathbf{Q}=\begin{pmatrix}0&&\mathbf{q}\\\mathbf{q}^\dagger&&0\end{pmatrix}.
\end{equation*}
The boundary conditions are given by
\begin{equation*}
	\mathbf{q}(x,t)\to0,\quad\mathbf{q}_{x}(x,t)\to0,\quad\mathrm{as}\quad x\to\pm\infty.
\end{equation*}
\noindent
We consider the Jost solutions of spectral problem,
\begin{equation*}
	\mathbf{\mathbf{\Phi}}^{\pm}(\lambda;x,t)\to\mathrm{e}^{-\mathrm{i}
		\lambda(x+\lambda t)
		\mathbf{\Lambda}_{3}},
	\quad x\to\pm\infty.
\end{equation*}
\noindent
If the function $\mathbf{q}(x,t)$ satisfies the focusing CNLS equation, which is equivalent
to the compatibility condition $\mathbf{\Phi}_{xt}(\lambda;x,t)=\mathbf{\Phi}_{tx}(\lambda;x,t)$,
then the matrix function $\mathbf{\Phi}(\lambda;x,t)$ exists and can be determined by the curve integral.
Inserting the ansatz
\begin{equation*}
	\mathbf{\Phi}^{\pm}(\lambda;x,t)=\pmb{\mu}^{\pm}(\lambda;x,t)\ee^{-\ii\lambda(x+\lambda t)\mathbf{\Lambda}_{3}}
\end{equation*}
into the spectral problem, we will obtain the Jost solution
\begin{equation}\label{lax-pair-mu}
	\begin{aligned}
		&\pmb{\mu}^{\pm}_{x}=-\ii\lambda\left[\mathbf{\Lambda}_{3},\pmb{\mu}^{\pm}
		\right]+\ii\mathbf{Q}\pmb{\mu}^{\pm},\,\begin{bmatrix}
			{\bf A,B}
		\end{bmatrix}={\bf AB-BA},\\
		&\pmb{\mu}^{\pm}_{t}=-\ii\lambda^{2}\left[\mathbf{\Lambda}_{3},\pmb{\mu}^{\pm}\right]+
		\left[\ii\lambda \mathbf{Q}+\frac{1}{2}(\mathbf{Q}_{x}+\ii\mathbf{Q}^2)\mathbf{\Lambda}_{3}\right]
		\pmb{\mu}^{\pm},
	\end{aligned}
\end{equation}

\noindent
with the normalization condition
\begin{equation*}
	\pmb{\mu}^{\pm}(\lambda;x,t)\to\mathbb{I}_{3},\quad x\to\pm\infty.
\end{equation*}
These solutions can be expressed as the Volterra type integrals
\begin{multline*}
	\pmb{\mu}^{\pm}(\lambda;x,t)
	=\mathbb{I}+\mathrm{i}\int_{(\pm\infty,t)}^{(x,t)}\mathrm{e}^{-\mathrm{i}\lambda(x-y)\mathbf{\Lambda}_{3}}
	\mathbf{Q}(y,s){\pmb{\mu}}^{\pm}(\lambda;y,s)\ee^{\mathrm{i}\lambda(x-y)\mathbf{\Lambda}_{3}}\mathrm{d}y  \\
	+\int_{(\pm\infty,0)}^{(\pm\infty,t)}\ee^{-\ii\lambda^{2}(t-s)\mathbf{\Lambda}_{3}}
	\left[\mathrm{i}\lambda \mathbf{Q}(y,s)+\frac{1}{2}\left(\mathbf{Q}_{y}(y,s)+
	\mathrm{i}\mathbf{Q}^{2}(y,s)\right)\mathbf{\Lambda}_{3}\right]{\pmb{\mu}}^{\pm}(\lambda;y,s)
	\ee^{\mathrm{i}\lambda^{2}(t-s)\mathbf{\Lambda}_{3}}\mathrm{d}s.
\end{multline*}
\noindent
From the conditions $\mathbf{q}(x,t)\to 0$ and $\mathbf{q}_{x}(x,t)\to 0$ as $x\to\pm\infty$
for arbitrary time $t$ ,the above curve integral can be reduced to
\begin{equation*}
	\pmb{\mu}^{\pm}(\lambda;x,t)
	=\mathbb{I}_{3}+\mathrm{i}\int_{\pm\infty}^{x}\mathrm{e}^{-\mathrm{i}\lambda(x-y)\mathbf{\Lambda}_{3}}
	\mathbf{Q}(y,t){\pmb{\mu}}^{\pm}(\lambda;y,t)\mathbf{e}^{\mathrm{i}\lambda(x-y)\mathbf{\Lambda}_{3}}\mathrm{d}y.
\end{equation*}
Define the following matrices
\begin{equation*}
	\pmb{\mu}_{+}(\lambda;x,t)=(\pmb{\mu}^{-}_{1},\pmb{\mu}^{+}_{2}),
	\quad\pmb{\mu}_{-}(\lambda;x,t)=(\pmb{\mu}^{+}_{1},\pmb{\mu}^{-}_{2}),
\end{equation*}
where $\pmb{\mu}_{1}$ denotes the first column of $\pmb{\mu}$ 
and $\pmb{\mu}_{2}$ denotes the last two columns of $\pmb{\mu}$.
\noindent
Through the analytic properties of Volterra integral equations, we obtain that the matrix 
$\pmb{\mu}_{+}(\lambda;x,t)$ is analytic in the upper half $\lambda$-plane and the matrix 
$\pmb{\mu}_{-}(\lambda;x,t)$ is analytic in the lower half $\lambda$-plane under the condition that 
$\{q_{1}(x,t),q_{2}(x,t)\}\in H^{1,2}(\mathbb{R})$ for the fixed time $t$.
Also, if $\mathbf{\Phi}(\lambda;x,t)$ is the solution of equation \eqref{lax-pair-Phi}, then 
$[\mathbf{\Phi}^\dagger(\lambda^*;x,t)]^{-1}$ can also solve equation \eqref{lax-pair-Phi}. Hence we have the following 
symmetry of the Jost solutions
\begin{equation}\label{symmetry-Phi-CNLS}
	\mathbf{\Phi}(\lambda;x,t)\mathbf{\Phi}^\dagger(\lambda^*;x,t)=\mathbb{I}_{3}.
\end{equation}
Define the scattering matrix
\begin{equation}\label{def-scattering-matrix}
	\mathbf{\Phi}^-(\lambda;x,t)=\mathbf{\Phi}^+(\lambda;x,t)
	\mathbf{S}(\lambda;t),\quad\lambda\in\mathbb{R}
\end{equation}
\noindent
where
\begin{equation*}
	\mathbf{S}(\lambda;t)=
	\begin{pmatrix}\bar{a}(\lambda)&&\mathbf{b}(\lambda)\\
		\overline{\mathbf{b}}(\lambda)&&\mathbf{a}(\lambda)\end{pmatrix},\quad \\ 
	\det(\mathbf{S}(\lambda;t))=1.
\end{equation*}
\noindent
Then we can also obain the symmetric relation of the scattering matrix
\begin{equation*}
	\mathbf{S}(\lambda;t)\mathbf{S}^\dagger(\lambda^*;t)=\mathbb{I}_{3},
\end{equation*}
\noindent
which implies that
\begin{equation*}
	\bar{a}(\lambda)=\det(\mathbf{a}^\dagger(\lambda^*)),\quad 
	\overline{\mathbf{b}}(\lambda)=-{\rm adj} (\mathbf{a}^\dagger(\lambda^*))\mathbf{b}^\dagger(\lambda^*),
\end{equation*}
\noindent
where ${\rm adj}(\mathbf{A})$ denotes the adjoint matrix of $\mathbf{A}$ in the context of linear algebra.

The evaluation of equation \eqref{def-scattering-matrix} at $t=0$ show that
\begin{equation*}
	\mathbf{\mathbf{S}}(\lambda;t)=\lim_{x\to+\infty}\ee^{\ii\lambda 
		x\hat{\mathbf{\Lambda}}_{3}}\pmb{\mu}_{-}(\lambda;x,t),\quad 
	\ee^{\hat{\bf B}}{\bf A}=\ee^{\bf B}{\bf A}\ee^{\bf -B},
\end{equation*}
\noindent
then we have
\begin{equation*}
	\mathbf{a}(\lambda)
	=\mathbb{I}_{2}+\mathrm{i}\int_{-\infty}^{+\infty}
	q^{\dagger}(y,0){\pmb{\mu}}^{-}_{12}(y,0)\mathrm{d}y.
\end{equation*}
Therefore, we can know that $\mathbf{a}(\lambda)$ is analytic in the lower half plane $\mathbb{C}_{-}$.
Moreover, we have $\det(\pmb{\mu}_{-})=\det(\mathbf{a}(\lambda))$. Assume the form 
$\det(\mathbf{a}(\lambda))=\prod_{i=1}^{N}\frac{\lambda-\lambda_{i}^{*}}{\lambda-\lambda_{i}}a_{0}(\lambda)$, 
where $a_{0}(\lambda)$ has no zeros and $\lambda_{i}\in\mathbb{C}_{+}$. 
So define  
$\mathcal{Z}=\left\{\lambda_i|\bar{a}(\lambda_{i})=\det(\mathbf{a}^\dagger(\lambda_i^*))=0\right\}_{i=1}^{N}$. 
Specifically, a norming constant matrix 
$\mathbf{h}_{i}=\overline{\mathbf{b}}(\lambda_i)$ exists such that:
\begin{equation*}
	\mathbf{\Phi}^-_1(\lambda_i;x,t)=\mathbf{\Phi}^+_1(\lambda_i;x,t) 
	\bar{a}(\lambda_i)+\mathbf{\Phi}^+_2(\lambda_i;x,t)
	\overline{\mathbf{b}}(\lambda_i)=\mathbf{\Phi}^+_2(\lambda_i;x,t)\mathbf{h}_i,
\end{equation*}
where $\mathbf{\Phi}^\pm_1$ denotes the first column of $\mathbf{\Phi}^\pm$ 
and $\mathbf{\Phi}^\pm_{2}$ denotes the last two columns of $\mathbf{\Phi}^\pm$.
The symmetry of scattering matrix $\mathbf{S}(\lambda)$ implies that 
\begin{equation*}
	\begin{aligned}
		\mathbf{\Phi}^-_2(\lambda_i^*;x,t){\rm adj}(\mathbf{a}(\lambda_i^*))
		&=\left(\mathbf{\Phi}^+_1(\lambda_i^*;x,t)\mathbf{b}(\lambda_i^*)+
		\mathbf{\Phi}^+_2(\lambda_i^*;x,t)\mathbf{a}(\lambda_i^*)\right)
		{\rm adj}(\mathbf{a}(\lambda_i^*))\\
		&=\mathbf{\Phi}^+_1(\lambda_i^*;x,t)\mathbf{b}(\lambda_i^*)
		{\rm adj}(\mathbf{a}(\lambda_i^*))+
		\mathbf{\Phi}^+_2(\lambda_i^*;x,t){\rm det}(\mathbf{a}(\lambda_i^*))\\
		&=-\mathbf{\Phi}^+_1(\lambda_i^*;x,t)\overline{\mathbf{b}}(\lambda_i)^\dagger
		=-\mathbf{\Phi}^+_1(\lambda_i^*;x,t)\mathbf{h}_i^\dagger.
	\end{aligned}
\end{equation*}


\begin{assumption}\label{assumption-q0}
	The initial data $\mathbf{q}_0(x)=\left(q_{1}(x,0),q_{2}(x,0)\right)$ of the Cauchy problem \eqref{CNLS} for CNLS equation 
	generates generic 
	scattering data in the sense that:
	\begin{itemize}
		\item[(a)] There are no spectral singularities, i.e., there exists a constant $c>0$ such that 
		$|\det(\mathbf{a}^\dagger(\lambda^*))|>c$ for any $\lambda\in\mathbb{R}$.
		\item[(b)] The discrete spectrum is simple, i.e., every zero of $\det(\mathbf{a}^\dagger(\lambda^*))$ in 
		$\mathbb{C}_+$ is simple.
		\item[(c)] $\{q_{1}(x,0),q_{2}(x,0)\}
		\in H^{1,2}(\mathbb{R})$.
	\end{itemize}
\end{assumption}

\noindent
Under Assumption \ref{assumption-q0}, it follows from Proposition 3.1 in \cite{Liu_2019} 
that the reflection coefficient $\hat{\mathbf{R}}(\lambda)=\mathbf{b}(\lambda)\mathbf{a}^{-1}(\lambda)$ associated with  
CNLS equation \eqref{CNLS} lies in 
$H^{2,1}(\mathbb{R})$.

The inverse scattering theory seeks to recover the solution of equation \eqref{CNLS} through the RHP. This is done as follows: form the Jost function 
$\mathbf{\Phi}^\pm(\lambda;x,t)=\pmb{\mu}^\pm(\lambda;x,t)\ee^{-\ii\lambda(x+\lambda t)\mathbf{\Lambda}_3}$ one constructs 
the function
\begin{equation}\label{def-M}
	\hat{\M}(\lambda)=\hat{\M}(\lambda;x,t):=
	\begin{cases}
		\pmb{\mu}_+\begin{pmatrix}
			\frac{1}{\mathrm{det}(\mathbf{a}^\dagger(\lambda^*))}&0\\0&\mathbb{I}_2
		\end{pmatrix}
		,&\lambda\in\mathbb{C}_+,\\
		\pmb{\mu}_-\begin{pmatrix}
			1&0\\0&\mathbf{a}^{-1}(\lambda)
		\end{pmatrix}
		,&\lambda\in\mathbb{C}_-.
	\end{cases}
\end{equation}

For data $\mathbf{q}_0$ satisfying Assumption 1, the matrix $\hat{\M}$ defined by equation \eqref{def-M} is the solution 
of the following RHP.

\begin{rhp}\label{initial-rhp-with-pole}
	Find an analytic function $\hat{\M}$: 
	$\mathbb{C}\setminus(\mathbb{R}\cup\mathcal{Z}\cup\mathcal{Z}^{*})\rightarrow SL_{2}(\mathbb{C})$ 
	with the following properties
	\begin{itemize}
		\item[1.] $\hat{\M}(\lambda)=\mathbb{I}_{3}+\mathcal{O}\left(\lambda^{-1}\right)$ as $\lambda\rightarrow\infty.$
		\item[2.] For each $\lambda\in\mathbb{R}$ (with $\mathbb{R}$ oriented left-to-right),
		$\hat{\M}$ takes continuous boundary values $\hat{\M}_{\pm}(\lambda)$ which satisfy 
		the jump relation: 
		$\hat{\M}_+(\lambda)=\hat{\M}_-(\lambda)\hat{\V}(\lambda)$, where
		\begin{equation}\label{initial-Jump}
			\hat{\V}(\lambda)=
			\begin{pmatrix}
				1+\hat{\mathbf{R}}(\lambda)\hat{\mathbf{R}}^\dagger(\lambda)&-\hat{\mathbf{R}}(\lambda)\ee^{-2\ii t\theta(\lambda)}\\
				-\hat{\mathbf{R}}^\dagger(\lambda)\ee^{2\ii t\theta(\lambda)}&\mathbb{I}_2
			\end{pmatrix},
		\end{equation}
		and
		\begin{equation}\label{def-R}
			\theta=\theta(\lambda;x,t)=\lambda^{2}-2\xi 
			\lambda=(\lambda-\xi)^{2}-\xi^{2},\quad\xi=-x/(2t), 
			\quad\hat{\mathbf{R}}(\lambda)=\mathbf{b}(\lambda)\mathbf{a}(\lambda)^{-1}.
		\end{equation}
		\item[3.] $\hat{\M}(\lambda)$ has simple poles at each $\lambda_i\in\mathcal{Z}$ and $\lambda_i^*\in\mathcal{Z}^*$ 
		at which
		\begin{equation*}
			\begin{array}{ll}
				& \mathop{\rm Res}\limits_{\lambda=\lambda_i}\hat{\M}=
				\mathop{\lim}\limits_{\lambda\to\lambda_i}\hat{\M}
				\begin{pmatrix}0&0\\-\mathbf{c}_i\ee^{2\ii t\theta}&0\end{pmatrix},\\
				& \mathop{\rm Res}\limits_{\lambda=\lambda_i^*}\hat{\M}=
				\mathop{\lim}\limits_{\lambda\to\lambda_i^*}\hat{\M}
				\begin{pmatrix}0&\mathbf{c}^\dagger_i\ee^{-2\ii t\theta}\\0&0\end{pmatrix},
			\end{array}
		\end{equation*}
		where $\mathbf{c}_i=-\frac{\mathbf{h}_i}{k_i}$ and 
		$k_i=\left.\frac{\dd\left({\rm det}(\mathbf{a}^\dagger(\lambda^*))\right)}{\dd \lambda}\right|_{\lambda=\lambda_i}$.
	\end{itemize}	
\end{rhp}
It is a simple consequence of Liouville theorem that if a solution exists, 
it is unique. The existence of solutions of RHP \ref{initial-rhp-with-pole} for all $(x,t)\in\mathbb{R}^2$ 
follows by means of vanishing lemma argument \cite{zhouscatteringtransforms1989} after replacing the poles by 
jumps along small circular contours in a standard way. Expanding this solution 
as $\lambda\rightarrow\infty$, $\hat{\M}=\mathbb{I}_3+\lambda^{-1}\hat{\M}_1+\oo(\lambda^{-2})$ and 
inserting this into equation \eqref{lax-pair-mu} one finds that 
\begin{equation}\label{recover-CNLS}
	\mathbf{q}(x,t)=\lim_{\lambda\rightarrow\infty} (2\lambda \hat{\M}(\lambda;x,t))_{12}.
\end{equation}

\section{The Long-Time Asymptotics}

\noindent
In this section, we make the following basic notations:
\begin{itemize}
	\item[1.] For any matrix $\M$ define $|\M|=(tr\M^\dagger \M)^{\frac{1}{2}}$.  For any matrix 
	function $\mathbf{A}(\cdot)$, the $L^p$ norm and $H^{i,j}$ norm are defined respectively as
	$\|\mathbf{A}(\cdot)\|_{L^p(\mathbb{R})}=\| \left|\mathbf{A}(\cdot) \right|\|_{L^p(\mathbb{R})}$, and 
	$\|\mathbf{A}(\cdot)\|_{H^{i,j}(\mathbb{R})}
	=\|\left|\mathbf{A}(\cdot)\right|\|_{H^{i,j}(\mathbb{R})}$.
	\item[2.] For two quantities A and B define $A\lesssim B$ if there exists a constant $C>0$ such that 
	$|A|\leq CB$.
	\item[3.] For any oriented contour $\sum$, we say that the left side is $``+"$ and the right side is $``-"$.
\end{itemize}

\subsection{Decomposition of continuous spectrum and discrete spectrum}
\noindent
Rather than directly analyzing the asymptotic behavior of the RHP \ref{initial-rhp-with-pole}, it is more convenient 
to examine a RHP without poles. Once the long-time behavior of the RHP without poles 
is established, we can readily derive the asymptotic behavior of the RHP \ref{initial-rhp-with-pole}.
In the following, we employ the Darboux transformation to decompose the sectional analytic matrix into its 
singular and analytic components. The singular part is represented by the Darboux matrix, while the 
left analytic component will satisfy a new RHP without the residue conditions. 
We will provide details for decomposing one discrete spectrum from the RHP \ref{initial-rhp-with-pole} with $N$ poles.

To begin with, we show that the solution $\hat{\M}$ of the initial RHP \ref{initial-rhp-with-pole} has 
the following symmetric property.
\begin{prop}\label{symmetry-M}
	The meromorphic matrix $\hat{\M}(\lambda;x,t)$ satisfies the follow symmetry: 
	\begin{equation}
		\hat{\M}(\lambda;x,t)\hat{\M}^\dagger(\lambda^*;x,t)=\mathbb{I}_{3}.
	\end{equation}
\end{prop}
\begin{proof}
	Set $\mathbf{H}(\lambda):=\mathbf{H}(\lambda;x,t)=\hat{\M}(\lambda;x,t)\hat{\M}^\dagger(\lambda^*;x,t)$ and 
	it follows from RHP \ref{initial-rhp-with-pole} that $\mathbf{H}(\lambda)$ approaches identity as $\lambda$ 
	tends to infinity. By the fact that 
	\begin{equation*}
		\det\hat{\V}(\lambda)=1\quad\text{and}\quad
		\hat{\V}(\lambda)=\hat{\V}^\dagger(\lambda), \quad \forall\lambda\in\mathbb{R},
	\end{equation*} 
	we have 
	\begin{equation*}
		\begin{aligned}
			\mathbf{H}_+(\lambda)&=\hat{\M}_+(\lambda)\hat{\M}^\dagger_-(\lambda)=
			\hat{\M}_-(\lambda)\hat{\V}(\lambda)\hat{\V}^\dagger
			(\lambda)^{-1}\hat{\M}^\dagger_+(\lambda)\\
			&=\hat{\M}_-(\lambda)\hat{\M}^\dagger_+(\lambda)=\mathbf{H}_-(\lambda),
		\end{aligned}
	\end{equation*}
	for $\lambda\in\mathbb{R}$. Therefore $\mathbf{H}(\lambda)$ has no jump on the real line. 
	By the residue relation in RHP \ref{initial-rhp-with-pole}, we have the Laurent expansions 
	\begin{equation*}
		\begin{aligned}
			\hat{\M}(\lambda)&=\mathbf{C}_i\left[\frac{\mathbf{N}_i}{\lambda-\lambda_i}+
			\mathbb{I}_3\right]+\oo(\lambda-\lambda_i),\quad 
			\mathbf{N}_i=\begin{pmatrix}0&0\\-\mathbf{c}_i\ee^{2\ii t\theta(\lambda_i)}&0\end{pmatrix}\\
			\hat{\M}(\lambda)&=\mathbf{C}_i^\prime\left[\frac{-\mathbf{N}^\dagger_i}
			{\lambda-\lambda_i^*}+\mathbb{I}_3\right]+\oo(\lambda-\lambda_i^*),
		\end{aligned}
	\end{equation*}
	where $\mathbf{C}_i$ and $\mathbf{C}_i^\prime$ are the constant terms in the Laurent expansions of $\hat{\M}(\lambda)$ 
	at $\lambda_i$ and $\lambda_i^*$ respectively. 
	This implies that as $\lambda\to\lambda_i$, we have
	\begin{equation*}
		\begin{aligned}
			\mathbf{H}(\lambda)&=\hat{\M}(\lambda)\hat{\M}^\dagger(\lambda^*)\\
			&=\left(\mathbf{C}_i\left[\frac{\mathbf{N}_i}{\lambda-\lambda_i}+
			\mathbb{I}_3\right]+\oo(\lambda-\lambda_i)\right)
			\left(\left[\frac{-\mathbf{N}_i}{\lambda-\lambda_i}+
			\mathbb{I}_3\right]\mathbf{C}_i^{\prime\dagger}+\oo(\lambda-\lambda_i)\right)\\
			&=\oo(1),
		\end{aligned}
	\end{equation*}
	since $\mathbf{N}_i$ is a nilpotent matrix, and 
	hence $\mathbf{H}(\lambda)$ is analytic in the whole complex plane.
	Then we conclude 
	that $\mathbf{H}(\lambda)\equiv\mathbb{I}_3$ 
	by the Liouville theorem and the asymptotics $\mathbf{H}(\lambda)\to\mathbb{I}_3$ as $\lambda\to\infty$.
\end{proof}
Now we proceed to decompose a single discrete spectrum from the matrix $\hat{\M}$ with the aid of Darboux transformation. 
\begin{lemma}\label{remove-a-pole}
	For the solution $\hat{\M}(\lambda;x,t)$ of 
	RHP \ref{initial-rhp-with-pole}, there exist Darboux matrix $\mathbf{T}_1^{(\pm)}(\lambda; x, t)$ and $\hat{\M}^{[1]}(\lambda;x,t)$ such that 
	\begin{equation}\label{relation-M-M1}
		\begin{cases}
			\hat{\M}(\lambda ; x, t)=\mathbf{T}_1^{(+)}(\lambda ; x, t) \hat{\M}^{[1]}(\lambda ; x, t) 
			\mathrm{diag}\left(\frac{\lambda-\lambda_{1}^{*}}{\lambda-\lambda_{1}}, 1,1\right),
			\quad&{\rm Im} \lambda>0,\\
			\hat{\M}(\lambda ; x, t)=\mathbf{T}_1^{(-)}(\lambda ; x, t) \hat{\M}^{[1]}(\lambda ; x, t) 
			\mathrm{diag}\left(1, \frac{\lambda-\lambda_{1}}{\lambda-\lambda_{1}^{*}},\frac{\lambda-\lambda_{1}}{\lambda-\lambda_{1}^{*}}\right),
			&{\rm Im} \lambda<0,
		\end{cases}
	\end{equation}
	where $\hat{\M}^{[1]}$ satisfies the following RHP:
	\begin{rhp}\label{remove-pole-of-initial-rhp}
		Find an analytic function $\hat{\M}^{[1]}$: 
		$\mathbb{C}\setminus(\mathbb{R}\cup\mathcal{Z}_1\cup\mathcal{Z}_1^{*})\rightarrow SL_{2}(\mathbb{C})$, where $\mathcal{Z}_n=\mathcal{Z}\setminus\{\lambda_i\}_{i=1}^n$
		with the following properties
		\begin{itemize}
			\item[1.] $\hat{\M}^{[1]}(\lambda)=\mathbb{I}_{3}+\mathcal{O}\left(\lambda^{-1}\right)$ as $\lambda\rightarrow\infty.$
			\item[2.] For each $\lambda\in\mathbb{R}$ (with $\mathbb{R}$ oriented left-to-right),
			$\hat{\M}^{[1]}$ takes continuous boundary values $\hat{\M}^{[1]}_{\pm}(\lambda)$ which satisfy 
			the jump relation: $\hat{\M}^{[1]}_+(\lambda)=\hat{\M}^{[1]}_-(\lambda)\hat{\V}^{[1]}(\lambda)$, where
			\begin{equation*}
				\hat{\V}^{[1]}(\lambda)=
				\begin{pmatrix}
					1+\hat{\mathbf{R}}_1(\lambda)\hat{\mathbf{R}}_1^\dagger(\lambda)&-\hat{\mathbf{R}}_1(\lambda)\ee^{-2\ii t\theta(\lambda)}\\
					-\hat{\mathbf{R}}_1^\dagger(\lambda)\ee^{2\ii t\theta(\lambda)}&\mathbb{I}_2
				\end{pmatrix}
			\end{equation*}
			and $\hat{\mathbf{R}}_1(\lambda)=\frac{\lambda-\lambda_{1}^{*}}{\lambda-\lambda_{1}}\mathbf{\hat{R}}(\lambda)$.
			\item[3.] $\hat{\M}^{[1]}(\lambda)$ has simple poles at each $\lambda_i\in\mathcal{Z}_1$ and $\lambda_i^*\in\mathcal{Z}_1^*$, 
			for $i=2\ldots N$,  
			whose residue condition can be represented as
			\begin{equation*}
				\begin{array}{ll}
					& \mathop{\rm Res}\limits_{\lambda=\lambda_i}\hat{\M}^{[1]}=
					\mathop{\lim}\limits_{\lambda\to\lambda_i}\hat{\M}^{[1]}
					\begin{pmatrix}
						0&0\\
						-\mathbf{c}_{1i}e^{2\ii t\theta}&0
					\end{pmatrix},\\
					& \mathop{\rm Res}\limits_{\lambda=\lambda_i^*}\hat{\M}^{[1]}=
					\mathop{\lim}\limits_{\lambda\to\lambda_i^*}\hat{\M}^{[1]}
					\begin{pmatrix}
						0&\mathbf{c}_{1i}^{\dagger}e^{-2\ii t\theta}\\
						0&0
					\end{pmatrix},
				\end{array}
			\end{equation*}
			where $\mathbf{c}_{1i}=\mathbf{c}_i\frac{\lambda_i-\lambda_1}{\lambda_i-\lambda_1^*}$.
		\end{itemize}	
	\end{rhp}
\end{lemma}

\begin{proof}
	The elementary Darboux transformation $\mathbf{T}_1^{(\pm)}(\lambda; x, t)$ can be constructed 
	as the following from \cite{Backlundtransformations20001,lingzhang2024}:
	\begin{equation}
		\mathbf{T}_1^{(+)}(\lambda ; x, t)=\mathbb{I}_3-\frac{\lambda_{1}-\lambda_{1}^*}{\lambda-\lambda_{1}^*} 
		\mathbf{P}_{1}(x, t), \quad \mathbf{P}_{1}(x, t)=\frac{\mathbf{\Phi}_{1} \mathbf{\Phi}_{1}^{\dagger}}{\mathbf{\Phi}_{1}^{\dagger} \mathbf{\Phi}_{1}},
	\end{equation}
	and
	\begin{equation}\label{relation-T+-T-}
		\mathbf{T}_1^{(-)}(\lambda ; x, t)=
		\frac{\lambda-\lambda_{1}^*}
		{\lambda-\lambda_{1}} \mathbf{T}_1^{(+)}(\lambda ; x, t)=\mathbb{I}_3-\frac{\lambda_{1}^{*}-\lambda_{1}}
		{\lambda-\lambda_{1}}\left(\mathbb{I}_3-\mathbf{P}_{1}(x, t)\right),
	\end{equation}
	where
	\begin{equation*}
		\mathbf{\Phi}_{1}=\hat{\M}_+^{[1]}(\lambda_{1} ; x, t) \ee^{-\mathrm{i} 
			\lambda_{1}\left(x+\lambda_{1} t\right) \pmb{\Lambda}_{3}} 
		\mathbf{v}_{1},\quad
		\mathbf{v}_{1}=[1,b_{11},b_{12}]^\dagger,
	\end{equation*}
	and $\mathbf{v}_{1}$ is determined below.
	The Darboux matrix $\mathbf{T}_1^{(\pm)}(\lambda; x, t)$ is analytic in the upper (lower) half plane. 
	In virtue of Darboux transformation, we can establish the relation between $\hat{\M}(\lambda ; x, t)$ and $\hat{\M}^{[1]}(\lambda ; x, t)$
	\begin{equation}\label{relation-Mhat-M+}
		\hat{\M}_+(\lambda ; x, t)=\mathbf{T}_1^{(+)}(\lambda ; x, t) 
		\hat{\M}_+^{[1]}(\lambda ; x, t) 
		\mathrm{diag}\left(\frac{\lambda-\lambda_{1}^{*}}{\lambda-\lambda_{1}}, 1,1\right),
		\quad {\rm Im}\lambda>0,
	\end{equation}
	and 
	\begin{equation}\label{relation-Mhat-M-}
		\hat{\M}_-(\lambda ; x, t)=\mathbf{T}_1^{(-)}(\lambda ; x, t) 
		\hat{\M}_-^{[1]}(\lambda ; x, t) 
		\mathrm{diag}\left(1, \frac{\lambda-\lambda_{1}}{\lambda-\lambda_{1}^{*}},
		\frac{\lambda-\lambda_{1}}{\lambda-\lambda_{1}^{*}}\right),
		\quad {\rm Im}\lambda<0.
	\end{equation}
	Note the fact that $\mathbf{P}_{1}=\mathbf{P}_{1}^\dagger$ and $\mathbf{P}_{1}^2=\mathbf{P}_{1}$, and then 
	we have 
	\begin{equation}\label{symmetry-T+}
		\mathbf{T}_1^{(+)}(\lambda ; x, t)\mathbf{T}_1^{(+)}(\lambda^* ; x, t)^\dagger=\mathbb{I}_3.
	\end{equation}
	Together with Proposition \ref{symmetry-M} and equation \eqref{relation-T+-T-}, we have 
	\begin{equation}
		\begin{aligned}
			\hat{\M}_+^{[1]}(\lambda ; x, t)\hat{\M}_-^{[1]}(\lambda^* ; x, t)^\dagger
			&=\frac{\lambda-\lambda_1}{\lambda-\lambda_1^*}\mathbf{T}_1^{(+)}(\lambda ; x, t)^{-1}
			\hat{\M}_+(\lambda;x,t)\hat{\M}_-^\dagger(\lambda^*;x,t)
			\left(\mathbf{T}_1^{(-)}(\lambda^* ; x, t)^\dagger\right)^{-1}\\
			&=\frac{\lambda-\lambda_1}{\lambda-\lambda_1^*}
			\left(\frac{\lambda-\lambda_1}{\lambda-\lambda_1^*}
			\mathbf{T}_1^{(+)}(\lambda ; x, t)\mathbf{T}_1^{(+)}
			(\lambda^* ; x, t)^\dagger\right)^{-1}=\mathbb{I}_3.
		\end{aligned}
	\end{equation}
	Now we consider the kernel of above
	elementary Darboux matrix $\mathbf{T}_1^{(\pm)}(\lambda; x, t)$:
	\begin{equation}\label{kernel-T+}
		\left.\mathbf{T}_1^{(+)}(\lambda ; x, t) \hat{\M}_+^{[1]}(\lambda ; x, t) 
		\mathrm{e}^{-\mathrm{i} \lambda(x+\lambda t) \pmb{\Lambda}_{3}}
		\right|_{\lambda=\lambda_{1}} \mathbf{v}_{1}=
		\left(\mathbb{I}_3-\mathbf{P}_{1}(x, t)\right)\mathbf{P}_{1}(x, t)=0,
	\end{equation}
	and
	\begin{equation}\label{kernel-T-}
		\begin{aligned}
			&\quad\left.\mathbf{T}_1^{(-)}(\lambda ; x, t) \hat{\M}_-^{[1]}(\lambda ; x, t) \mathrm{e}^{-\mathrm{i} 
				\lambda(x+\lambda t) \boldsymbol{\pmb{\Lambda}}_{3}}
			\right|_{\lambda=\lambda_{1}^*}\begin{bmatrix}
				\mathbf{w}_{11}, \mathbf{w}_{12}
			\end{bmatrix}\\
			&=\mathbf{P}_{1}(x, t)\left(\hat{\M}^{[1]}_+(\lambda_1 ; x, t)^\dagger\right)^{-1}
			\mathrm{e}^{-\mathrm{i} 
				\lambda_1^*(x+\lambda_1^* t) \boldsymbol{\pmb{\Lambda}}_{3}}\begin{bmatrix}
				\mathbf{w}_{11}, \mathbf{w}_{12}
			\end{bmatrix}\\
			&=\frac{\mathbf{\Phi}_{1} }{\mathbf{\Phi}_{1}^{\dagger} \mathbf{\Phi}_{1}}
			\mathbf{v}_{1}^\dagger
			\mathrm{e}^{\mathrm{i} 
				\lambda_1^*(x+\lambda_1^* t) \boldsymbol{\pmb{\Lambda}}_{3}}
			\hat{\M}^{[1]}_+(\lambda_1 ; x, t)^\dagger\left(\hat{\M}^{[1]}_+(\lambda_1 ; x, t)^\dagger\right)^{-1}
			\mathrm{e}^{-\mathrm{i} 
				\lambda_1^*(x+\lambda_1^* t) \boldsymbol{\pmb{\Lambda}}_{3}}\begin{bmatrix}
				\mathbf{w}_{11}, \mathbf{w}_{12}
			\end{bmatrix}\\
			&=\frac{\mathbf{\Phi}_{1} }{\mathbf{\Phi}_{1}^{\dagger} \mathbf{\Phi}_{1}}
			\mathbf{v}_{1}^\dagger\begin{bmatrix}
				\mathbf{w}_{11}, \mathbf{w}_{12}
			\end{bmatrix}=0,
		\end{aligned}
	\end{equation}
	where $\mathbf{w}_{11}=[-b_{11},1,0]^{\top}$ and 
	$\mathbf{w}_{12}=[-b_{12},0,1]^{\top}$.
	
	Then we can obtain the jump condition 
	of $\hat{\M}(\lambda;x,t)$ on 
	$\lambda\in\mathbb{R}$ as follows:
	\begin{equation}
		\begin{aligned}
			\hat{\M}_+(\lambda ; x, t)&=\mathbf{T}_1^{(+)}(\lambda ; x, t) \hat{\M}^{[1]}_{-}(\lambda ; x, t)\hat{\V}^{[1]}(\lambda)
			\mathrm{diag}\left(\frac{\lambda-\lambda_{1}^{*}}{\lambda-\lambda_{1}}, 1,1\right)\\
			&=\hat{\M}_-(\lambda ; x, t)\mathrm{diag}\left(\frac{\lambda-\lambda_{1}}{\lambda-\lambda_{1}^{*}}, 1,1\right)\hat{\V}^{[1]}(\lambda)
			\mathrm{diag}\left(\frac{\lambda-\lambda_{1}^{*}}{\lambda-\lambda_{1}}, 1,1\right)\\
			&=\hat{\M}_-(\lambda ; x, t)
			\begin{pmatrix}
				1+\hat{\mathbf{R}}_1(\lambda)\hat{\mathbf{R}}_1^\dagger(\lambda)&-\frac{\lambda-\lambda_{1}}{\lambda-\lambda_{1}^{*}}\hat{\mathbf{R}}_1(\lambda)\ee^{-2\ii t\theta(\lambda)}\\
				-\frac{\lambda-\lambda_{1}^{*}}{\lambda-\lambda_{1}}\hat{\mathbf{R}}_1^\dagger(\lambda)\ee^{2\ii t\theta(\lambda)}&\mathbb{I}_2
			\end{pmatrix}.
		\end{aligned}
	\end{equation}
	Therefore, it follows from equation \eqref{initial-Jump} necessarily that 
	$\mathbf{\hat{R}}(\lambda)=\frac{\lambda-\lambda_{1}}{\lambda-\lambda_{1}^{*}}\hat{\mathbf{R}}_1(\lambda)$.
	Moreover, by equations \eqref{relation-Mhat-M+}, \eqref{kernel-T+} and \eqref{kernel-T-}, we 
	obtain the residue conditions for $\hat{\M}$:
	\begin{equation*}
		\begin{aligned}
			\mathop{\rm Res}\limits_{\lambda=\lambda_1}\hat{\M}&=
			\begin{pmatrix}
				(\lambda_1-\lambda_1^*)[\mathbf{T}_1^{(+)}(\lambda_1;x,t)\hat{\M}^{[1]}(\lambda_1;x,t)]_1&0
			\end{pmatrix}\\
			&=\begin{pmatrix}
				-(\lambda_1-\lambda_1^*)\ee^{2\ii t\theta(\lambda_1)}[\mathbf{T}_1^{(+)}(\lambda_1;x,t)\hat{\M}^{[1]}(\lambda_1;x,t)]_2
				\begin{pmatrix}
					b_{11}^{*}\\b_{12}^{*}
				\end{pmatrix}&0
			\end{pmatrix}\\
			&=\mathop{\lim}\limits_{\lambda\to\lambda_1}
			\hat{\M}\begin{pmatrix}
				0&0\\
				-(\lambda_1-\lambda_1^*)\begin{pmatrix}
					b_{11}^*\\b_{12}^*
				\end{pmatrix}\ee^{2\ii t\theta(\lambda_1)}&0
			\end{pmatrix},
		\end{aligned}    
	\end{equation*}
	and 
	\begin{equation*}
		\begin{aligned}
			\mathop{\rm Res}\limits_{\lambda=\lambda_1^*}\hat{\M}&=
			\begin{pmatrix}
				0&(\lambda_1^*-\lambda_1)[\mathbf{T}_1^{(-)}(\lambda_1;x,t)\hat{\M}^{[1]}_{-}(\lambda_1;x,t)]_2
			\end{pmatrix}\\
			&=\begin{pmatrix}
				0&(\lambda_1^*-\lambda_1)\ee^{-2\ii t\theta(\lambda_1^*)}[\mathbf{T}_1^{(-)}(\lambda_1;x,t)\hat{\M}^{[1]}_{-}(\lambda_1;x,t)]_1
				\begin{pmatrix}
					b_{11}&b_{12}
				\end{pmatrix}
			\end{pmatrix}\\
			&=\mathop{\lim}\limits_{\lambda\to\lambda_1^*}
			\hat{\M}\begin{pmatrix}
				0&(\lambda_1^*-\lambda_1)\ee^{-2\ii t\theta(\lambda_1^*)}
				\begin{pmatrix}
					b_{11}&b_{12}
				\end{pmatrix}\\
				0&0
			\end{pmatrix}
		\end{aligned}
	\end{equation*}
	where $[\M]_1$ denotes the first column of matrix $\M$ 
and $[\M]_2$ denotes the last two columns of $\M$.
	By comparing the RHP  \ref{initial-rhp-with-pole} and RHP \ref{remove-pole-of-initial-rhp}, a straightforward calculation shows that 
	\begin{equation}
		\mathbf{c}_1=(\lambda_1-\lambda_1^*)\begin{pmatrix}
			b_{11}^*\\b_{12}^*
		\end{pmatrix}.
	\end{equation}
	
	Similarly, for $\lambda_j$, $j=2,\dots,N$, we have the residue conditions
	\begin{equation*}
		\begin{aligned}
			\mathop{\rm Res}\limits_{\lambda=\lambda_j}\hat{\M}&=
			\mathbf{T}_1^{(+)}(\lambda_j;x,t)
			\mathop{\lim}\limits_{\lambda\to\lambda_j}
			\hat{\M}^{[1]}\begin{pmatrix}
				0&0\\
				-\mathbf{c}_{1j}\ee^{2\ii t\theta(\lambda_j)}&0
			\end{pmatrix}
			\mathrm{diag}\left(\frac{\lambda_j-\lambda_1^*}{\lambda_j-\lambda_1},1,1\right)\\
			&=\mathop{\lim}\limits_{\lambda\to\lambda_j}
			\hat{\M}\begin{pmatrix}
				0&0\\
				-\mathbf{c}_{1j}\frac{\lambda_j-\lambda_1^*}{\lambda_j-\lambda_1}\ee^{2\ii t\theta(\lambda_j)}&0
			\end{pmatrix},
		\end{aligned}    
	\end{equation*}
	and 
	\begin{equation*}
		\begin{aligned}
			\mathop{\rm Res}\limits_{\lambda=\lambda_j^*}\hat{\M}&=
			\mathbf{T}_1^{(-)}(\lambda_j;x,t)
			\mathop{\lim}\limits_{\lambda\to\lambda_j^*}
			\hat{\M}^{[1]}\begin{pmatrix}
				0&\mathbf{c}^{\dagger}_{1j}\ee^{-2\ii t\theta(\lambda_j^*)}\\
				0&0
			\end{pmatrix}
			\mathrm{diag}\left(1,\frac{\lambda_j^*-\lambda_1}{\lambda_j^*-\lambda_1^*},
			\frac{\lambda_j^*-\lambda_1}{\lambda_j^*-\lambda_1^*}\right)\\
			&=\mathop{\lim}\limits_{\lambda\to\lambda_j^*}
			\hat{\M}\begin{pmatrix}
				0&\mathbf{c}^{\dagger}_{1j}\frac{\lambda_j^*-\lambda_1}{\lambda_j^*-\lambda_1^*}\ee^{-2\ii t\theta(\lambda_j^*)}\\
				0&0
			\end{pmatrix}.
		\end{aligned}    
	\end{equation*}
	The above residue conditions imply that 
	\begin{equation*}
		\mathbf{c}_j=\mathbf{c}_{1j}\frac{\lambda_j-\lambda_1^*}{\lambda_j-\lambda_1}.
	\end{equation*}
	Then we complete the proof.
\end{proof}

From Lemma \ref{remove-a-pole}, the solution $\hat{\M}(\lambda;x,t)$ of RHP \ref{initial-rhp-with-pole} 
can be decomposed to three parts as equation \eqref{relation-M-M1} and hence 
the recovery formula \eqref{recover-CNLS} can be rewritten as 
\begin{equation*}
	\begin{aligned}
		\mathbf{q}(x,t)&=2\left(\left(\mathbf{T}_1^{(+)}\right)_1+\hat{\M}^{[1]}_1\right)_{12}\\
		&=\lim_{\lambda\rightarrow\infty} 2\lambda \left(\hat{\M}^{[1]}(\lambda;x,t)\right)_{12}-
		2(\lambda_1-\lambda_1^*)\left(\mathbf{P}_1(x,t)\right)_{12},
	\end{aligned}
\end{equation*}
where $\left(\mathbf{T}_1^{(+)}\right)_1$ and $\hat{\M}^{[1]}_1$ are the residues of 
$\mathbf{T}^{(+)}(\lambda;x,t)$ and $\hat{\M}^{[1]}(\lambda;x,t)$ at $\lambda=\infty$ respectively.
In order to recover $\mathbf{q}(x,t)$ by the solution of RHP without any pole, 
we should decompose $N$ single discrete spectrums from the matrix $\hat{\M}$. In fact, we only need to replace 
the procedure in Lemma \ref{remove-a-pole} $N$ times and then obtain the following Theorem \ref{remove-N-poles}.
Before that, we will introduce some necessary notations:
\begin{equation}\label{def-Ti}
	\begin{aligned}
		&\mathbf{T}_i^{(+)}(\lambda ; x, t)=\mathbb{I}_3-\frac{\lambda_{i}-\lambda_{i}^*}
		{\lambda-\lambda_{i}^*} 
		\mathbf{P}_{i}(x, t), \quad \mathbf{P}_{i}(x, t)=\frac{\mathbf{\Phi}_{i} 
			\mathbf{\Phi}_{i}^{\dagger}}{\mathbf{\Phi}_{i}^{\dagger} \mathbf{\Phi}_{i}},\\
		&\mathbf{T}_i^{(-)}(\lambda ; x, t)=
		\frac{\lambda-\lambda_{i}^*}
		{\lambda-\lambda_{i}} \mathbf{T}_i^{(+)}(\lambda ; x, t),\quad
		\mathbf{\Phi}_{i}=\hat{\M}_+^{[i]}(\lambda_{i} ; x, t) \ee^{-\mathrm{i} 
			\lambda_{i}\left(x+\lambda_{i} t\right) \pmb{\Lambda}_{3}} 
		\mathbf{v}_{i},\\
		&\mathbf{v}_{i}=\begin{bmatrix}
			1,b_{i1},b_{i2}
		\end{bmatrix}^\dagger, \quad
		\begin{bmatrix}
			b_{i1},b_{i2}
		\end{bmatrix}^\dagger=\frac{\mathbf{c}_i}{\lambda_i-\lambda_i^*}
		\prod_{j=1}^{i-1}\frac{\lambda_i-\lambda_j}{\lambda_i-\lambda_j^*}.
	\end{aligned}
\end{equation}
where $\hat{\M}^{[i]}$, for $1\le i\le N$, is the solution of the following RHP 
\begin{rhp}
	Find an analytic function $\hat{\M}^{[i]}$: 
	$\mathbb{C}\setminus(\mathbb{R}\cup\mathcal{Z}_i\cup\mathcal{Z}_i^{*})\rightarrow SL_{2}(\mathbb{C})$, 
	where $1\le i\le N-1$ and $\mathcal{Z}_n=\mathcal{Z}\setminus\{\lambda_j\}_{j=1}^n$, 
	with the following properties
	\begin{itemize}
		\item[1.] $\hat{\M}^{[i]}(\lambda)=\mathbb{I}_{3}+\mathcal{O}\left(\lambda^{-1}\right)$ as $\lambda\rightarrow\infty.$
		\item[2.] For each $\lambda\in\mathbb{R}$ (with $\mathbb{R}$ oriented left-to-right),
		$\hat{\M}^{[i]}$ takes continuous boundary values $\hat{\M}^{[i]}_{\pm}(\lambda)$ which satisfy 
		the jump relation: $\hat{\M}^{[i]}_+(\lambda)=\hat{\M}^{[i]}_-(\lambda)\hat{\V}^{[i]}(\lambda)$, where
		\begin{equation*}
			\hat{\V}^{[i]}(\lambda)=
			\begin{pmatrix}
				1+\hat{\mathbf{R}}_i(\lambda)\hat{\mathbf{R}}_i^\dagger(\lambda)&-\hat{\mathbf{R}}_i(\lambda)\ee^{-2\ii t\theta(\lambda)}\\
				-\hat{\mathbf{R}}_i^\dagger(\lambda)\ee^{2\ii t\theta(\lambda)}&\mathbb{I}_2
			\end{pmatrix}
		\end{equation*}
		and $\mathbf{R}_i(\lambda)=\prod_{k=1}^{i}\frac{\lambda-\lambda_k^*}{\lambda-\lambda_k}\hat{\mathbf{R}}(\lambda)$.
		\item[3.] $\hat{\M}^{[i]}(\lambda)$ has simple poles at each $\lambda_j\in\mathcal{Z}_i$ 
		and $\lambda_j^*\in\mathcal{Z}_i^*$
		at which
		\begin{equation*}
			\begin{array}{ll}
				& \mathop{\rm Res}\limits_{\lambda=\lambda_j}\hat{\M}^{[i]}=
				\mathop{\lim}\limits_{\lambda\to\lambda_j}\hat{\M}^{[i]}
				\begin{pmatrix}
					0&0\\
					-\mathbf{c}_{ij}e^{2\ii t\theta}&0
				\end{pmatrix},\\
				& \mathop{\rm Res}\limits_{\lambda=\lambda_j^*}\hat{\M}^{[i]}=
				\mathop{\lim}\limits_{\lambda\to\lambda_j^*}\hat{\M}^{[i]}
				\begin{pmatrix}
					0&\mathbf{c}_{ij}^{\dagger}e^{-2\ii t\theta}\\
					0&0
				\end{pmatrix},
			\end{array}
		\end{equation*}
		where $\mathbf{c}_{ij}=\mathbf{c}_j\prod_{k=1}^{i}\frac{\lambda_j-\lambda_k}{\lambda_j-\lambda_k^*}$.
	\end{itemize}	
\end{rhp}
Let $\M\equiv\hat{\M}^{[N]}$ be the solution 
of a pole-free RHP.
\begin{rhp}\label{initial-rhp-without-role}
	Find an analytic function $\M$: 
	$\mathbb{C}\setminus\mathbb{R}\rightarrow SL_{2}(\mathbb{C})$ 
	with the following properties
	\begin{itemize}
		\item[1.] $\M(\lambda)=\mathbb{I}_{3}+\mathcal{O}\left(\lambda^{-1}\right)$ as $\lambda\rightarrow\infty.$
		\item[2.] For each $\lambda\in\mathbb{R}$ (with $\mathbb{R}$ oriented left-to-right),
		$\M$ takes continuous boundary values $\M_{\pm}(\lambda)$ which satisfy 
		the jump relation: $\M_+(\lambda)=\M_-(\lambda)\V(\lambda)$ where
		\begin{equation}\label{def-R-no-soliton}
			\V(\lambda)=
			\begin{pmatrix}
				1+\mathbf{R}(\lambda)\mathbf{R}^\dagger(\lambda)&-\mathbf{R}(\lambda)\ee^{-2\ii t\theta(\lambda)}\\
				-\mathbf{R}^\dagger(\lambda)\ee^{2\ii t\theta(\lambda)}&\mathbb{I}_2
			\end{pmatrix},
		\end{equation}
		where $\mathbf{R}(\lambda)=\prod_{i=1}^{N}\frac{\lambda-\lambda_i^*}{\lambda-\lambda_i}\hat{\mathbf{R}}(\lambda)$.
	\end{itemize}
\end{rhp}

\begin{theorem}\label{remove-N-poles}
	For $\hat{\M}$ which is the solution of RHP \ref{initial-rhp-with-pole} with N poles, 
	there exist Darboux matrix $\mathbf{T}^{(\pm)}=\prod_{i=1}^{N}\mathbf{T}^{(\pm)}_i$ and $\M$ such that 
	\begin{equation}\label{relation-hatM-M}
		\begin{cases}
			\hat{\M}(\lambda ; x, t)=\mathbf{T}^{(+)}(\lambda ; x, t) \M(\lambda ; x, t) 
			\mathrm{diag}\left(\prod_{i=1}^{N}\frac{\lambda-\lambda_{i}^{*}}{\lambda-\lambda_{i}}, 1,1\right),
			\quad&{\rm Im}\lambda>0,\\[4pt]
			\hat{\M}(\lambda ; x, t)=\mathbf{T}^{(-)}(\lambda ; x, t) \M(\lambda ; x, t) 
			\mathrm{diag}\left(1, \prod_{i=1}^{N}\frac{\lambda-\lambda_{i}}{\lambda-\lambda_{i}^{*}},
			\prod_{i=1}^{N}\frac{\lambda-\lambda_{i}}{\lambda-\lambda_{i}^{*}}\right),
			&{\rm Im}\lambda<0.
		\end{cases}
	\end{equation}
	where $\M$ is the solution of RHP \ref{initial-rhp-without-role} without any poles.
	
	In particular, the solution $\mathbf{q}(x,t)$ of CNLS equation \eqref{CNLS} can be recovered by a 
	better formula
	\begin{equation}\label{matrix-form-T1}
		\mathbf{q}(x,t)=\lim_{\lambda\rightarrow\infty} \left(2\lambda \M(\lambda;x,t)\right)_{12}-
		2\mathbf{X}_1\mathbf{G}^{-1}(\mathbf{X}_2)^\dagger,
	\end{equation}
	where 
	\begin{equation*}
		\begin{aligned}
			&\mathbf{X}=\begin{bmatrix}
				|x_1\rangle,\cdots,|x_N\rangle
			\end{bmatrix},\quad \mathbf{G}=\left(\mathbf{G}_{ji}\right)_{1\le j,i\le N},\quad
			\mathbf{G}_{ji}=\frac{\langle x_j|x_i\rangle}{\lambda_i-\lambda_j^*},\\
			&|x_i\rangle=\M(\lambda_i; x, t)\mathrm{diag}\left(\prod_{j=i+1}^{N}\frac{\lambda_i-\lambda_{j}^{*}}{\lambda_i-\lambda_{j}}, 1,1\right)
			\ee^{-\ii\lambda_i(x+\lambda_i t)\pmb{\Lambda}_3}\mathbf{v}_i,
			\quad\langle x_j|=|x_j\rangle^\dagger,
		\end{aligned}
	\end{equation*}
	and $\mathbf{X}_1$ represents the first row of 
	$\mathbf{X}$ and $\mathbf{X}_2$ represents the second and third row of $\mathbf{X}$.
\end{theorem}
\begin{proof}
	Repeated the procedure in Lemma \ref{remove-a-pole} $N$ times, we can obtain the 
	relation \eqref{relation-hatM-M} between $\hat{\M}(\lambda)$ and $\M(\lambda)$. 
	Moreover, $\M(\lambda)$ would satisfy RHP \ref{initial-rhp-without-role} without 
	residue conditions.
	
	Then we only need to rewrite 
	$\mathbf{T}^{(\pm)}=\prod_{i=1}^{N}\mathbf{T}^{(\pm)}_i$ as another 
	matrix form \eqref{matrix-form-T1}.
	Now we consider the kernel of $\mathbf{T}_i^{(+)}(\lambda ; x, t)$ at $\lambda=\lambda_i$
	\begin{equation*}
		\mathbf{T}_i^{(+)}(\lambda_i; x, t)\mathbf{\Phi}_{i}=(\mathbb{I}_3-\mathbf{P}_i)\mathbf{\Phi}_{i}=0,\qquad i=1,2,\cdots, N,
	\end{equation*}
	and hence we deduce the kernel of $\mathbf{T}^{(+)}(\lambda ; x, t)$ at $\lambda=\lambda_N$
	\begin{equation*}
		\mathbf{T}^{(+)}(\lambda; x, t)\mathbf{\Phi}_{N}|_{\lambda=\lambda_N}=0.
	\end{equation*}
	Then we will determine the kernel of $\mathbf{T}^{(+)}(\lambda; x, t)$ 
	at $\lambda=\lambda_i$, for $i=1,2,\cdots, N-1$,  i.e. $\mathbf{T}^{(+)}(\lambda; x, t)\pmb{\varphi}_{i}|_{\lambda=\lambda_i}=0$. And the kernel $\pmb{\varphi}_{i}$ can be determined by 
	\begin{multline*}
		\prod_{j=i+1}^{N}\mathbf{T}_j^{(+)}(\lambda ; x, t)\pmb{\varphi}_i\big|_{\lambda=\lambda_i}=\mathbf{\Phi}_i
		=\hat{\M}^{[i]}(\lambda_i)\ee^{-\ii\lambda_i(x+\lambda_i t)\pmb{\Lambda}_3}\mathbf{v}_i\\
		=\prod_{j=i+1}^{N}\mathbf{T}_j^{(+)}(\lambda ; x, t)\M(\lambda ; x, t)\big|_{\lambda=\lambda_i}
		\mathrm{diag}\left(\prod_{j=i+1}^{N}\frac{\lambda_i-\lambda_{j}^{*}}{\lambda_i-\lambda_{j}}, 1,1\right)
		\ee^{-\ii\lambda_i(x+\lambda_i t)\pmb{\Lambda}_3}\mathbf{v}_i,
	\end{multline*}
	which implies that
	\begin{equation*}
		\pmb{\varphi}_{i}=\M(\lambda_i; x, t)\mathrm{diag}\left(\prod_{j=i+1}^{N}\frac{\lambda_i-\lambda_{j}^{*}}{\lambda_i-\lambda_{j}}, 1,1\right)
		\ee^{-\ii\lambda_i(x+\lambda_i t)\pmb{\Lambda}_3}\mathbf{v}_i.
	\end{equation*}
	
	By equation \eqref{def-Ti}, $\mathbf{T}^{(+)}$ can be written in a fractional form
	\begin{equation}\label{eq:ndt}
		\mathbf{T}^{(+)}(\lambda;x,t)=\mathbb{I}_3-
		\sum_{i=1}^{N}\frac{\,|y_i\rangle \langle x_i|\,}{\lambda-\lambda_i^*},
	\end{equation}
	where $|y_i\rangle$ and $\langle x_i|$ have been not determined yet and denote $\langle y_i|=|y_i\rangle^\dagger$ and 
	$|x_i\rangle=\langle x_i|^\dagger$.
	By the symmetry of $\mathbf{T}_i^{(+)}$, we have $\mathbf{T}^{(+)}(\lambda)(\mathbf{T}^{(+)}(\lambda^*))^\dagger=\mathbb{I}_3$ and hence 
	\begin{equation*}
		0=\mathop{\rm Res}\limits_{\lambda=\lambda_i}\mathbf{T}^{(+)}(\lambda)(\mathbf{T}^{(+)}(\lambda^*))^\dagger
		=-\mathbf{T}^{(+)}(\lambda_i)|x_i\rangle \langle y_i|.
	\end{equation*}
	Since $\langle y_i|\neq0$, the matrix $\mathbf{T}^{(+)}(\lambda_i)|x_i\rangle=0$. 
	Then we can choose $|x_i\rangle$ as $\pmb{\varphi}_i$ and further we have
	\begin{equation*}
		|x_i\rangle=\sum_{j=1}^{N}\frac{\,|y_j\rangle \langle x_j|\,}
		{\lambda_i-\lambda_j^*}|x_i\rangle.
	\end{equation*}
	The above equation can be represented in the matrix form so that $\mathbf{X}=\mathbf{Y}\mathbf{G}$, where
	\begin{equation*}
		\mathbf{X}=\begin{bmatrix}
			|x_1\rangle,\cdots,|x_N\rangle
		\end{bmatrix}=\begin{bmatrix}
			\pmb{\varphi}_1,\cdots,\pmb{\varphi}_N
		\end{bmatrix},\quad \mathbf{Y}=\begin{bmatrix}
			|y_1\rangle,\cdots,|y_N\rangle
		\end{bmatrix}\quad\text{and}\quad
		\mathbf{G}_{ji}=\frac{\langle x_j|x_i\rangle}{\lambda_i-\lambda_j^*},
	\end{equation*}
	$\mathbf{G}_{ji}$s denote the $(j,i)$ element of matrix $\mathbf{G}$.
	Hence, the Darboux matrix $\mathbf{T}^{(+)}(\lambda;x,t)$ can be rewritten
	\begin{equation}\label{matrix-form-T+}
		\mathbf{T}^{(+)}(\lambda;x,t)=\mathbb{I}_3-\mathbf{X}\mathbf{G}^{-1}(\lambda\mathbb{I}_3-\mathbf{D})^{-1}\mathbf{X}^\dagger,
	\end{equation}
	where $\mathbf{D}=\mathrm{diag}(\lambda_1^*,\cdots,\lambda_N^*)$.
	By equation \eqref{relation-hatM-M}, the recovery formula \eqref{recover-CNLS} can be rewritten as
	\begin{equation*}
		\mathbf{q}(x,t)=
		\lim_{\lambda\rightarrow\infty} \left(2\lambda \M(\lambda;x,t)\right)_{12}
		+\left(2\lambda \mathbf{T}^{(+)}(\lambda;x,t)\right)_{12}.
	\end{equation*}
	Together with equation \eqref{matrix-form-T+}, we can obtain the formula \eqref{matrix-form-T1}. 
	This completes the proof.
\end{proof}

By Theorem \ref{remove-N-poles}, the solution $\mathbf{q}(x,t)$ in equation \eqref{matrix-form-T1} can 
be recovered by the matrix 
$\M(\lambda;x,t)$ at $\lambda=\infty$ and  $\lambda=\lambda_i\in\mathcal{Z}$. 
To obtain the long-time behavior of the solution $\mathbf{q}(x,t)$ for the CNLS equation \eqref{CNLS}, 
it suffices to analyze the long-time behavior of the matrix 
$\M(\lambda;x,t)$ at infinity and $\mathcal{Z}$ with respect to the spectral parameter $\lambda_i$. 
In the next subsection, we apply INSD to obtain the long-time asymptotic behavior of solution $\M(\lambda;x,t)$ at $\lambda\in\mathcal{Z}\cup\{\infty\}$.

\subsection{Factorization of the jump matrix}

\noindent

An analogue of the classical NSD for RHPs was developed by Deift and 
Zhou \cite{deift1993steepest}. The key idea is to deform the jump matrix $\V(\lambda)$, which involves the oscillatory 
factor $\ee^{\pm 2\ii\theta}$ on the real axis, 
onto new contours in the plane where $\ee^{\pm 2\ii\theta}$
decay, as illustrated in figure \ref{region-decay-e}. And we can note that the jump matrix $\V(\lambda)$ has 
two distinct factorizations, which are needed in the contour deformation described in subsection \ref{contour-deformation},
\begin{equation*}
	\V(\lambda)=
	\begin{cases}
		\begin{pmatrix}
			1&-\mathrm{e}^{-2\ii t\theta}\mathbf{R}(\lambda)\\
			0&\mathbb{I}_2\end{pmatrix}
		\begin{pmatrix}1&0\\
			-\mathrm{e}^{2\ii t\theta}\mathbf{R}^\dagger(\lambda^*)&\mathbb{I}_2
		\end{pmatrix},\\
		\begin{pmatrix}1&0\\
			-\frac{\mathrm{e}^{2\ii t\theta}\mathbf{R}^\dagger(\lambda^*)}
			{1+\mathbf{R}
				(\lambda)\mathbf{R}^\dagger(\lambda^*)}&\mathbb{I}_2
		\end{pmatrix}
		\begin{pmatrix}
			1+\mathbf{R}(\lambda)\mathbf{R}^\dagger(\lambda^*)&0\\
			0&(\mathbb{I}_2+\mathbf{R}^\dagger(\lambda^*)\mathbf{R}(\lambda))^{-1}
		\end{pmatrix}
		\begin{pmatrix}1&-\frac{\mathrm{e}^{-2\ii t\theta}\mathbf{R}(\lambda)}
			{1+\mathbf{R}(\lambda)\mathbf{R}^\dagger(\lambda^*)}\\
			0&\mathbb{I}_2
		\end{pmatrix}.
	\end{cases}
\end{equation*}
We can see that the second factorization involves a diagonal matrix. 
Then we introduce a matrix function $\pmb{\delta}(\lambda)$ to conjugate the matrix $\M$, leading 
to a new RHP 
which has lower/upper factorization for the jump matrix on the real axis without the diagonal matrix.
Let $\pmb{\delta}(\lambda)$ be the solution of the matrix RHP
\begin{equation}\label{RHP-delta}
	\begin{cases}
		\pmb{\delta}_+(\lambda)=(\mathbb{I}_2+\mathbf{R}^\dagger\mathbf{R})\pmb{\delta}_-(\lambda),\quad&\lambda<\xi,\\
		\pmb{\delta}(\lambda)\rightarrow\mathbb{I}_2, &\lambda\rightarrow\infty,
	\end{cases}
\end{equation}
where $\xi=-\frac{x}{2t}$ is the stationary point of phase function $\theta(\lambda)$. 
Then the determinant of $\pmb{\delta}(\lambda)$ satisfies the following RHP:
\begin{equation}\label{RHP-det-delta}
	\begin{cases}
		\det\pmb{\delta}_+(\lambda)=\det\pmb{\delta}_-(\lambda)(1+|\mathbf{R}(\lambda)|^2),\quad  &\lambda<\xi,\\
		\det\pmb{\delta}(\lambda)\rightarrow1, &\lambda\rightarrow\infty.
	\end{cases}
\end{equation}

\begin{figure}[h]
	\centering
	\begin{tikzpicture}
		\fill[pink] (0,0) rectangle (2,2);\fill[pink] (-2,-2) rectangle (0,0);
		\fill[cyan!30] (-2,0) rectangle (0,2);\fill[cyan!30] (0,-2) rectangle (2,0);
		
		\draw[->,thick] (-2, 0) -- (-1, 0) ;
		\draw[thick] (-1, 0) -- (0, 0) ;
		\draw[->,thick] (0, 0) -- (1, 0) ;
		\draw[thick] (1, 0) -- (2, 0) ;
		
		\draw[dashed] (0, 1) -- (0, 0) ;
		\draw[dashed] (0, 2) -- (0, 1) ;
		\draw[dashed] (0, -2) -- (0, -1) ;
		\draw[dashed] (0, -1) -- (0, 0) ;
		\node at (0.3, -0.3)  {\small $\xi$};
		
		\node at (1, 1)  {\fontsize{7}{24}\selectfont $\ee^{2\ii t\theta}\ll1$};\node at (-1, -1)  {\fontsize{7}{24}\selectfont $\ee^{2\ii t\theta}\ll1$};
		\node at (-1, 1) {\fontsize{7}{24}\selectfont $\ee^{-2\ii t\theta}\ll1$};\node at (1, -1)  {\fontsize{7}{24}\selectfont $\ee^{-2\ii t\theta}\ll1$};
	\end{tikzpicture}
	\caption{\small{The regions of decay of the exponential factor $\ee^{\pm2\ii t\theta}$ for large $t>0$.}}
	\label{region-decay-e}
\end{figure}
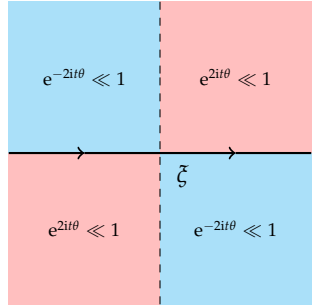
In terms of the positive definiteness of the jump matrix $\mathbb{I}_2+\mathbf{R}^\dagger\mathbf{R}(\lambda)$ 
and the vanishing lemma \cite{ablowitz2003complex}, it follows that $\pmb{\delta}$ exists and is unique. 
Additionally, $\det\pmb{\delta}$ can be determined by using the Plemelj formula \cite{ablowitz2003complex}
\begin{equation}\label{def-kappa}
	\begin{aligned}
		&\det\pmb{\delta}(\lambda)=\exp\left(\ii\int_{-\infty}^{\xi}\frac{\kappa(s)}{s-\lambda} \dd s\right),\\
		&\kappa(s)=-\frac{1}{2\pi}\ln\left(1+|\mathbf{R}(s)|^2\right).
	\end{aligned}
\end{equation}
Let $\mathbf{g}(\lambda)=[\pmb{\delta}^\dagger(\lambda^*)]^{-1}$, for $\lambda\in(-\infty,\xi)$, we have, 
\begin{equation*}
	\begin{aligned}
		\mathbf{g}_+(\lambda)
		&=\lim_{\varepsilon\rightarrow0_+}\mathbf{g}(\lambda+\ii\varepsilon)
		=\lim_{\varepsilon\rightarrow0_+}\left(\pmb{\delta}^\dagger(\lambda-\ii\varepsilon)\right)^{-1}\\
		&=\lim_{\varepsilon\rightarrow0_+}\left(\pmb{\delta}^\dagger(\lambda-\ii\varepsilon)\right)^{-1}
		=\left(\pmb{\delta}_-(\lambda)^\dagger\right)^{-1}\\
		&=\left(\mathbb{I}_2+\mathbf{R}^\dagger\mathbf{R}\right)(\pmb{\delta}_+^\dagger(\lambda))^{-1}
		=\left(\mathbb{I}_2+\mathbf{R}^\dagger\mathbf{R}\right)
		\lim_{\varepsilon\rightarrow0_+}\left(\pmb{\delta}^\dagger(\lambda+\ii\varepsilon)\right)^{-1}\\
		&=\left(\mathbb{I}_2+\mathbf{R}^\dagger\mathbf{R}\right)\lim_{\varepsilon\rightarrow0_+}\mathbf{g}(\lambda-\ii\varepsilon)
		=\left(\mathbb{I}_2+\mathbf{R}^\dagger\mathbf{R}\right)\mathbf{g}_-(\lambda).
	\end{aligned}
\end{equation*}
Noticing the uniqueness of the solution for the matrix RHP \eqref{RHP-delta}, we have
\begin{equation}\label{symmetry-delta}
	\pmb{\delta}(\lambda)=[\pmb{\delta}^\dagger(\lambda^*)]^{-1}.
\end{equation}
Inserting equation \eqref{symmetry-delta} into equation \eqref{RHP-delta}, we obtain
\begin{equation*}
	\begin{aligned}
		&
		|\pmb{\delta}_-(\lambda)|^2= 
		\begin{cases}
			2-\frac{|\mathbf{R}(\lambda)|^2}{1+|\mathbf{R}(\lambda)|^2},\\
			2,
		\end{cases} 
		&
		\begin{aligned}
			&\lambda<\xi,\\
			&\lambda>\xi,
		\end{aligned}\\
		&
		|\pmb{\delta}_+(\lambda)|^2= 
		\begin{cases}
			2+|\mathbf{R}(\lambda)|^2,\\
			2, 
		\end{cases}
		&
		\begin{aligned}
			&\lambda<\xi,\\
			&\lambda>\xi,
		\end{aligned}
	\end{aligned}
\end{equation*}
for $\lambda\in\mathbb{R}$. Similarly, we have
\begin{equation*}
	\det\pmb{\delta}(\lambda)(\det\pmb{\delta}(\lambda^*))^*=1,
\end{equation*}
then for real $\lambda$,
\begin{equation*}
	|\det\pmb{\delta}_+(\lambda)|=|\det\pmb{\delta}_-^{-1}(\lambda)|
	=\left(1+|\mathbf{R}(\lambda)|^2\right)^{\frac{1}{2}},
	\quad\lambda<\xi.
\end{equation*}
Hence, by the maximum principle, we obtain the boundedness of $\det\pmb{\delta}$ and $\pmb{\delta}$
\begin{equation}\label{bdd-delta}
	|\left(\det\pmb{\delta}(\lambda)\right)^{\pm1}|\leq const<\infty,
	\quad|\pmb{\delta}^{\pm1}(\lambda)|\leq const<\infty
\end{equation}
for any $\lambda\in\mathbb{C}$.

Now we utilize $\det\pmb{\delta}$ and $\pmb{\delta}$ to define a new function $\M^{(1)}(\lambda)$
\begin{equation}\label{def-Delta}
	\M^{(1)}(\lambda)= \M(\lambda)\pmb{\Delta}^{-1}(\lambda),
	\quad\pmb{\Delta}(\lambda)=
	\begin{pmatrix}
		\det\pmb{\delta}(\lambda)&0\\
		0&\pmb{\delta}^{-1}(\lambda)
	\end{pmatrix},
\end{equation}
and $\M^{(1)}(\lambda)$ satisfies the following RHP.

\begin{rhp}\label{rhp-M1}
	Find an analytic function $\M^{(1)}$: 
	$\mathbb{C}\setminus\mathbb{R}\rightarrow SL_{2}(\mathbb{C})$ 
	with the following properties
	\begin{itemize}
		\item[1.] $\M^{(1)}(\lambda)=\mathbb{I}_{3}+\mathcal{O}\left(\lambda^{-1}\right)$ as $\lambda\rightarrow\infty.$
		\item[2.] For each $\lambda\in\mathbb{R}$ (with $\mathbb{R}$ oriented left-to-right),
		$\M^{(1)}$ takes continuous boundary values $\M^{(1)}_{\pm}(\lambda)$ which satisfy 
		the jump relation: $\M_+^{(1)}(\lambda)=\M_-^{(1)}(\lambda)\V^{(1)}(\lambda)$ where
		\begin{equation*}
			\V^{(1)}(\lambda)=
			\begin{cases}
				\begin{pmatrix}
					1 & -\det\pmb{\delta}\ee^{-2\ii t\theta}\mathbf{R}(\lambda)\pmb{\delta}\\
					0 & \mathbb{I}_2
				\end{pmatrix}
				\begin{pmatrix}
					1 & 0\\
					-(\det\pmb{\delta})^{-1}\ee^{2\ii t\theta}\pmb{\delta}^{-1}\mathbf{R}^\dagger(\lambda) & \mathbb{I}_2
				\end{pmatrix},
				&\lambda\in(\xi,+\infty),\\
				\begin{pmatrix}
					1 & 0\\
					-(\det\pmb{\delta}_-)^{-1}\ee^{2\ii t\theta}\pmb{\delta}_-^{-1}
					\frac{\mathbf{R}^\dagger(\lambda)}{1+\mathbf{R}\mathbf{R}^\dagger} & \mathbb{I}_2
				\end{pmatrix}
				\begin{pmatrix}
					1 & -\det\pmb{\delta}_+\ee^{-2\ii t\theta}
					\frac{\mathbf{R}(\lambda)}{1+\mathbf{R}\mathbf{R}^\dagger}\pmb{\delta}_+\\
					0 & \mathbb{I}_2
				\end{pmatrix},
				&\lambda\in(-\infty,\xi).
			\end{cases}
		\end{equation*}
	\end{itemize}
\end{rhp}
Note that the jump matrix $\V^{(1)}$ admits the following factorization 
\begin{equation}
	\V^{(1)}=(\mathbb{I}_{3}-\mathbf{w}_1^-)^{-1}(\mathbb{I}_{3}+\mathbf{w}_1^+)
\end{equation}
where 
\begin{equation}
	\begin{aligned}
		\mathbf{w}_1
		&=(\mathbf{w}^-_1,\mathbf{w}^+_1)\\
		&=\begin{cases}
			\left(
			\begin{pmatrix}
				0&-\det\pmb{\delta}\ee^{-2\ii t\theta}
				\mathbf{R}(\lambda)\pmb{\delta}\\
				0&0\end{pmatrix},
			\begin{pmatrix}0&0\\
				-(\det\pmb{\delta})^{-1}\ee^{2\ii t\theta}\pmb{\delta}^{-1}\mathbf{R}^\dagger(\lambda)&0
			\end{pmatrix}
			\right),&\lambda>\xi ,\\
			\left(
			\begin{pmatrix}
				0&0\\
				-(\det\pmb{\delta}_-)^{-1}\ee^{2\ii t\theta}\pmb{\delta}_-^{-1}
				\frac{\mathbf{R}^\dagger(\lambda)}{1+\mathbf{R}\mathbf{R}^\dagger}&0
			\end{pmatrix},
			\begin{pmatrix}
				0&-\det\pmb{\delta}_+\ee^{-2\ii t\theta}
				\frac{\mathbf{R}(\lambda)}{1+\mathbf{R}\mathbf{R}^\dagger}\pmb{\delta}_+\\
				0&0
			\end{pmatrix}\right),&\lambda<\xi .
		\end{cases}
	\end{aligned}
\end{equation}
Then it is well known that the solvability of the RHP \ref{rhp-M1} is equivalent to 
the solvability of the following Beals-Coifman integral equation 
\begin{equation}\label{Beals-Coifman-mu1}
	\pmb{\mu}_1=\mathbb{I}_{3}+C_{\mathbf{w}_1}\pmb{\mu}_1,\qquad 
	C_{\mathbf{w}_1}\pmb{\mu}_1\equiv C_{\mathbb{R}}^+\pmb{\mu}_1{\mathbf{w}_1^-}+C_{\mathbb{R}}^-\pmb{\mu}_1{\mathbf{w}_1^+}
\end{equation}
where 
$C_{\mathbb{R}}^{\pm}$ is defined by
\begin{equation}
	(C_{\mathbb{R}}^\pm h)(z)=\lim_{z'\rightarrow z^\pm}(C_{\mathbb{R}} h)(z'),
\end{equation} 
and the Cauchy operator
\begin{equation}
	(C_{\mathbb{R}} h)(z)=\int\limits_{\mathbb{R}}\frac{h(s)}{s-z}\frac{ds}{2\pi \ii} 
	,\quad z\in\mathbb{C}\setminus\mathbb{R}.
\end{equation}

Now, we have the following lemmas to show 
that the operator $(1-C_{\mathbf{w}_1})^{-1}$ exists and 
is bounded in $L^2(\mathbb{R})$, which guarantees the existence of solutions $\M^{(1)}$. 

\begin{lemma}\label{bdd-w}
	For any factorization 
	$\V=(\V^-)^{-1}\V^+,\V^\pm,(\V^\pm)^{-1}\in L^\infty(\mathbb{R})$, the operator $1-C_{\mathbf{w}}$ with 
	$\mathbf{w}=(\mathbf{w}^-,\mathbf{w}^+)=(\mathbb{I}_3-\V^-,\V^+-\mathbb{I}_3)$ has a bounded inverse in $L^2(\mathbb{R})$ and 
	\begin{equation*}
		\|(1-C_\mathbf{w})^{-1}\|_{L^2(\mathbb{R})}\lesssim 1.
	\end{equation*}
\end{lemma}
A proof of Lemma \ref{bdd-w} is detailed in Appendix A for the sake of completeness.

\begin{lemma}\label{bbd-w1}
	The operator $(1-C_{\mathbf{w}_1})^{-1}$ exists in $L^2(\mathbb{R})$ and 
	the bound on $(1-C_{\mathbf{w}_1})^{-1}$ is equivalent to the bound on $(1-C_{\mathbf{w}})^{-1}$.
\end{lemma}

\begin{proof}
	By \cite[pp.1042]{deift2002long}, results follow from the analyticity and boundedness properties of 
	$\det\pmb{\delta}$ and $\pmb{\delta}$.
\end{proof}
It follows then by Proposition 2.6 in \cite{deift2002long} that RHP \ref{rhp-M1} has a 
unique solution and 
\begin{equation}
	\M^{(1)}(\lambda;x,t)=\mathbb{I}_{3}+C_{\mathbb{R}}(\pmb{\mu}_1(\mathbf{w}^-_1 + \mathbf{w}^+_1)),
\end{equation}
where $\pmb{\mu}$ is the unique solution 
of Beals-Coifman integral equation \eqref{Beals-Coifman-mu1}.
Moreover, we can obtain the expansion of $\M^{(1)}$ 
\begin{equation}
	\M^{(1)}(\lambda;x,t)=\mathbb{I}_3+\frac{\M_1^{(1)}(x,t)}{\lambda}+\oo(\lambda^{-2}),
	\quad \lambda\rightarrow\infty,
\end{equation}
where
\begin{equation}
	\M_1^{(1)}(x,t)=-\frac{1}{2\pi \ii}\int_{\mathbb{R}}\pmb{\mu}_1 \mathbf{w}_1
	=-\frac{1}{2\pi \ii}\int_{\mathbb{R}}\pmb{\mu}_1 (\mathbf{w}_1^-+\mathbf{w}_1^+).
\end{equation}

Next, we want to construct a new RHP whose solution $\M^{(2)}$ is a suitable approximation 
to solution $\M^{(1)}$. For a function $f$ defined on the real axis, we introduce the notation 
\begin{equation}
	[f](\lambda)=\frac{f(\xi)}{(1+\ii(\lambda-\xi))^2} ,\quad \lambda\in\mathbb{R},
\end{equation}
where $\xi=-x/(2t)$ denotes the stationary phase point for $\theta$.
Clearly, $[f](\xi)=f(\xi)$. Then we replace
$\mathbf{R}$ by $[\mathbf{R}]$ in RHP \ref{rhp-M1}, obtaining the following RHP 
which allows us to deform jump contour from the real axis into a new contour as described in next subsection.

\begin{rhp}\label{rhp-M2}
	Find an analytic function $\M^{(2)}$: 
	$\mathbb{C}\setminus\mathbb{R}\rightarrow SL_{2}(\mathbb{C})$ 
	with the following properties
	\begin{itemize}
		\item[1.] $\M^{(2)}(\lambda)=\mathbb{I}_{3}+\mathcal{O}\left(\lambda^{-1}\right)$ as $\lambda\rightarrow\infty.$
		\item[2.] For each $\lambda\in\mathbb{R}$ (with $\mathbb{R}$ oriented left-to-right),
		$\M^{(2)}$ takes continuous boundary values $\M^{(2)}_{\pm}(\lambda)$ which satisfy 
		the jump relation: $\M_+^{(2)}(\lambda)=\M_-^{(2)}(\lambda)\V^{(2)}(\lambda)$ where
		\begin{equation*}
			\V^{(2)}(\lambda)=
			\begin{cases}
				\begin{pmatrix}
					1 & -\det\pmb{\delta}\ee^{-2\ii t\theta}
					[\mathbf{R}](\lambda)\pmb{\delta}\\
					0 & \mathbb{I}_2
				\end{pmatrix}
				\begin{pmatrix}
					1 & 0\\
					-(\det\pmb{\delta})^{-1}\ee^{2\ii t\theta}
					\pmb{\delta}^{-1}[\mathbf{R}^\dagger](\lambda) & \mathbb{I}_2
				\end{pmatrix},
				&\lambda\in(\xi,+\infty),\\[6pt]
				\begin{pmatrix}
					1 & 0\\
					-(\det\pmb{\delta}_-)^{-1}\ee^{2\ii t\theta}\pmb{\delta}_-^{-1}
					[\frac{\mathbf{R}^\dagger(\lambda)}{1+\mathbf{R}\mathbf{R}^\dagger}] & \mathbb{I}_2
				\end{pmatrix}
				\begin{pmatrix}
					1 & -\det\pmb{\delta}_+\ee^{-2\ii t\theta}
					[\frac{\mathbf{R}(\lambda)}{1+\mathbf{R}\mathbf{R}^\dagger}]\pmb{\delta}_+\\
					0 & \mathbb{I}_2
				\end{pmatrix},
				&\lambda\in(-\infty,\xi).
			\end{cases}
		\end{equation*}
	\end{itemize}
\end{rhp}
Then we have a similar factorization for $\V^{(2)}$, 
$\V^{(2)}=(\mathbb{I}_{3}-\mathbf{w}_2^-)^{-1}(\mathbb{I}_{3}+\mathbf{w}_2^+)$, where 
\begin{equation}
	\begin{aligned}
		\mathbf{w}_2
		&=(\mathbf{w}^-_2,\mathbf{w}^+_2)\\
		&=\begin{cases}
			\left(
			\begin{pmatrix}
				0&-\det\pmb{\delta}\ee^{-2\ii t\theta}[\mathbf{R}]\pmb{\delta}\\
				0&0\end{pmatrix},
			\begin{pmatrix}0&0\\
				-(\det\pmb{\delta})^{-1}\ee^{2\ii t\theta}\pmb{\delta}^{-1}[\mathbf{R}^\dagger]&0
			\end{pmatrix}
			\right),&\lambda>\xi ,\\
			\left(
			\begin{pmatrix}
				0&0\\
				-(\det\pmb{\delta}_-)^{-1}\ee^{2\ii t\theta}\pmb{\delta}_-^{-1}
				[\frac{\mathbf{R}^\dagger}{1+\mathbf{R}\mathbf{R}^\dagger}]&0
			\end{pmatrix},
			\begin{pmatrix}
				0&-\det\pmb{\delta}_+\ee^{-2\ii t\theta}
				[\frac{\mathbf{R}}{1+\mathbf{R}\mathbf{R}^\dagger}]\pmb{\delta}_+\\
				0&0
			\end{pmatrix}\right),&\lambda<\xi .
		\end{cases}
	\end{aligned}
\end{equation}

\subsection{Contour deformation}\label{contour-deformation}
We now perform contour deformation on RHP \ref{rhp-M2}. Define the contours
\begin{equation*}
	\Sigma=\cup_{k=1}^4\Sigma_k,
	\quad\Sigma_k=\xi+\ee^{\frac{(2k-1)\pi\ii}{4}}\mathbb{R}_+,\quad k=1,2,3,4,
\end{equation*}
oriented with increasing real part and denote the six open sectors in $\mathbb{C}$, separated by $\mathbb{R}$ and the collection of $\Sigma_k$,
$k = 1,...,4$, by $\Omega_k$, $k = 1,...,6$ starting with the sector $\Omega_1$ between $[\xi,+\infty)$ 
and $\Sigma_1$ and numbered consecutively 
continuing counterclockwise, see figure \ref{region-Omega}. 
We introduce the trivial extension $\V^e$ of $\V^{(2)}$ onto $\Sigma\cup\mathbb{R}$ by setting $\V^e=\V^{(2)}$ 
on $\mathbb{R}$ and $\V^e=\mathbb{I}_3$ on $\Sigma$.
\begin{figure}[h]
	\centering
	\begin{tikzpicture}
		\draw[-] (-3, 0) -- (3, 0) ;
		\draw[->] (-1.5, 1.5) -- (1.5, -1.5) ;
		\draw[->] (-1.5, -1.5) -- (1.5, 1.5) ;
		\draw[->] (-3, 3) -- (-1.5, 1.5) ;
		\draw[->] (-3, -3) -- (-1.5, -1.5) ;
		\draw[-] (1.5, 1.5) -- (3, 3) ;
		\draw[-] (1.5, -1.5) -- (3, -3) ;
		\node at (1, 0.5)  {$\Omega_1$};
		\node at (0, 1)  {$\Omega_2$};
		\node at (-1, 0.5)  {$\Omega_3$};
		\node at (-1, -0.5)  {$\Omega_4$};
		\node at (0, -1)  {$\Omega_5$};
		\node at (1, -0.5)  {$\Omega_6$};
		\node at (3, 2.5)  {$\Sigma_1$};
		\node at (3, -2.5)  {$\Sigma_4$};
		\node at (-3, 2.5)  {$\Sigma_2$};
		\node at (-3, -2.5)  {$\Sigma_3$};
		
		\vspace{0.5cm} 
	\end{tikzpicture}
	\caption{\small{The contours $\Sigma_k$ and regions $\Omega_k$.}}
	\label{region-Omega}
\end{figure}
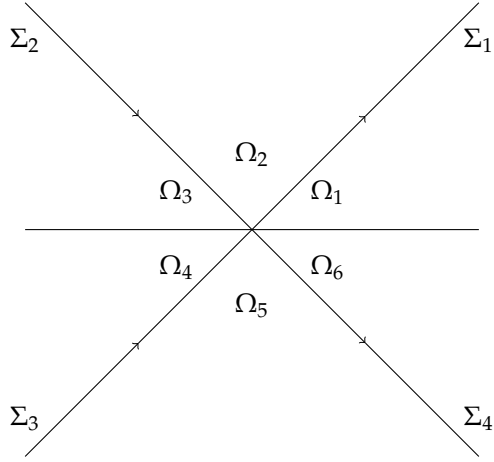

\begin{lemma}\label{bdd-w2-equ-bdd-we}
	(cf. \cite[pp.1023-1044]{deift2002long}) 
	Suppose $\V^e=(1-\mathbf{w}_e^-)^{-1}(1+\mathbf{w}_e^+)$, then $(1-C_{\mathbf{w}_2})^{-1}$ is bounded in $L^2(\mathbb{R})$ if and 
	only if $(1-C_{\mathbf{w}_e})^{-1}$ is bounded in $L^2(\mathbb{R}\cup\Sigma)$.
\end{lemma} 
We define a piecewise analytic function $\mathbf{\Xi}$ 
\begin{equation}\label{analysis-expand-Phi}
	\mathbf{\Xi}(\lambda)=
	\begin{cases}
		\begin{pmatrix}
			1 & 0\\
			(\det\pmb{\delta})^{-1}\ee^{2\ii t\theta}\pmb{\delta}^{-1}[\mathbf{R}^\dagger](\lambda) & \mathbb{I}_2
		\end{pmatrix},  
		&\lambda \in\Omega_1,\\
		\begin{pmatrix}
			1 & \det\pmb{\delta} \ee^{-2\ii t\theta}
			[\frac{\mathbf{R}(\lambda)}{1+\mathbf{R}\mathbf{R}^\dagger}]\pmb{\delta}\\
			0 & \mathbb{I}_2
		\end{pmatrix},
		&\lambda \in\Omega_3,\\
		\begin{pmatrix}
			1 & 0\\
			-(\det\pmb{\delta})^{-1}\ee^{2\ii t\theta}\pmb{\delta}^{-1}
			[\frac{\mathbf{R}^\dagger(\lambda)}{1+\mathbf{R}\mathbf{R}^\dagger}] & \mathbb{I}_2
		\end{pmatrix},
		&\lambda \in \Omega_4,\\
		\begin{pmatrix}
			1 & -\det\pmb{\delta} \ee^{-2\ii t\theta}[\mathbf{R}](\lambda)\pmb{\delta}\\
			0 & \mathbb{I}_2
		\end{pmatrix}  ,
		&\lambda \in\Omega_6,\\
		\mathbb{I}_3,
		&\lambda \in\Omega_2\cup\Omega_5,
	\end{cases} 
\end{equation}
and we set $\M^{(3)}=\M^{(2)}\mathbf{\Xi}$. 
Direct calculation shows that $\M^{(3)}$ satisfies the following RHP.

\begin{rhp}\label{rhp-M3}
	Find an analytic function $\M^{(3)}$: 
	$\mathbb{C}\setminus\Sigma\rightarrow SL_{2}(\mathbb{C})$ 
	with the following properties
	\begin{itemize}
		\item[1.] $\M^{(3)}(\lambda)=\mathbb{I}_{3}+\mathcal{O}\left(\lambda^{-1}\right)$ as $\lambda\rightarrow\infty.$
		\item[2.] For each $\lambda\in\Sigma$ ,
		$\M^{(3)}$ takes continuous boundary values $\M^{(3)}_{\pm}(\lambda)$ which satisfy 
		the jump relation: $\M_+^{(3)}(\lambda)=\M_-^{(3)}(\lambda)\V^{(3)}(\lambda)$ where
		\begin{equation*}
			\V^{(3)}(\lambda)=
			\begin{cases}
				\begin{pmatrix}
					1 & 0\\
					-(\det\pmb{\delta})^{-1}\ee^{2\ii t\theta}\pmb{\delta}^{-1}[\mathbf{R}^\dagger](\lambda) & \mathbb{I}_2
				\end{pmatrix},  
				&\lambda \in\Sigma_1,\\
				\begin{pmatrix}
					1 & -\det\pmb{\delta} \ee^{-2\ii t\theta}
					[\frac{\mathbf{R}(\lambda)}{1+\mathbf{R}\mathbf{R}^\dagger}]\pmb{\delta}\\
					0 & \mathbb{I}_2
				\end{pmatrix},
				&\lambda \in\Sigma_2,\\
				\begin{pmatrix}
					1 & 0\\
					-(\det\pmb{\delta})^{-1}\ee^{2\ii t\theta}
					[\frac{\mathbf{R}^\dagger(\lambda)}{1+\mathbf{R}\mathbf{R}^\dagger}] & \mathbb{I}_2
				\end{pmatrix},
				&\lambda \in \Sigma_3,\\
				\begin{pmatrix}
					1 & -\det\pmb{\delta} \ee^{-2\ii t\theta}[\mathbf{R}](\lambda)\pmb{\delta}\\
					0 & \mathbb{I}_2
				\end{pmatrix},
				&\lambda \in\Sigma_4.
			\end{cases}
		\end{equation*}
	\end{itemize}
\end{rhp}
Because $\mathbf{\Xi}$ is analytic and uniformly bounded in 
$\mathbb{C}\setminus(\mathbb{R}\cup\Sigma)$ for all $x\in\mathbb{R}$ and $t>0$, we can obtain the 
following lemma.

\begin{lemma}(cf. \cite[pp.1044-1045]{deift2002long})
	The operator $(1-C_{\mathbf{w}_e})^{-1}$ is bounded in $L^2(\mathbb{R}\cup\Sigma)$ if and 
	only if $(1-C_{\mathbf{w}_3})^{-1}$ is bounded in $L^2(\Sigma)$.
\end{lemma}

The key point to note is that RHP \ref{rhp-M2} has been transformed into RHP \ref{rhp-M3} 
with jump matrix $\V^{(3)}$, where all the exponential factors $\ee^{\pm \ii\theta}$ are now 
exponentially decreasing. Then we need to provide some estimates on $\det\pmb{\delta}(\lambda)$. 
By equation \eqref{RHP-det-delta}, the expression of $\det\pmb{\delta}$ can be rewritten as 
\begin{equation}
	\det\pmb{\delta}(\lambda)=
	\exp\left(C_{(-\infty,\xi)}\ln\left(1+\left|\mathbf{R}\right|^2\right)(\lambda)\right),\quad
	\lambda\in\mathbb{C}\setminus(-\infty,\xi). 
\end{equation}
Let $\chi(s)$ denote the characteristics function of the interval 
$(\xi-1,\xi)$. For 
$\lambda\in\mathbb{C}\setminus(-\infty,\xi)$,
\begin{equation*}
	\begin{aligned}
		&\quad C_{(-\infty,\xi)}\ln\left(1+\left|\mathbf{R}\right|^2\right)(\lambda) \\
		&=\int_{-\infty}^{\xi}\left(\ln\left(1+|\mathbf{R}(s)|^2\right)-
		\ln\left(1+|\mathbf{R}(\xi)|^2\right)
		\chi(s)(s-\xi+1)\right)\frac{\dd s}{2\pi \ii(s-\lambda)} \\
		&\quad+\ln\left(1+\left|\mathbf{R}(\xi)\right|^2\right)
		\int_{\xi-1}^{\xi}\frac{s-\xi+1}{s-\lambda}\frac{\dd s}{2\pi \ii} \\
		&=\beta(\lambda,\xi)+\ii\kappa(\xi)(1-(\lambda-\xi+1)\ln(\lambda-\xi+1) \\
		&\quad+(\lambda-\xi)\ln(\lambda-\xi)+\ln(\lambda-\xi)),
	\end{aligned}
\end{equation*}
where $\kappa(\xi)=-\frac{1}{2\pi}\ln\left(1+|\mathbf{R}(\xi)|^2\right)$ as defined in equation \eqref{def-kappa}, and 
\begin{equation*}
	\quad\beta(\lambda,\xi)
	=\int_{-\infty}^{\xi}\left(\ln(1+|\mathbf{R}(s)|^2)-\ln(1+|\mathbf{R}(\xi)|^2)
	\chi(s)(s-\xi+1)\right)\frac{ds}{2\pi \ii(s-\lambda)}.
\end{equation*}
We decompose $\det\pmb{\delta}$ as
\begin{equation*}
	\det\pmb{\delta}=\delta_0\delta_1,
\end{equation*}
where
\begin{equation*}
	\delta_0=\ee^{\beta(\xi,\xi)+\ii\kappa(\xi)}(\lambda-\xi)^{\ii\kappa(\xi)},
	\quad\delta_1=\ee^{\zeta(\lambda,\xi)},
\end{equation*}
and 
\begin{equation*}
	\zeta(\lambda,\xi)=(\beta(\lambda,\xi)-\beta(\xi,\xi))
	+\ii\kappa(\xi)\Big((\lambda-\xi)\ln(\lambda-\xi)-(\lambda-\xi+1)\ln(\lambda-\xi+1)\Big).
\end{equation*}

It follows from the fact 
\begin{equation}\label{boundedness-lamda-xi}
    \left|(\lambda-\xi)^{\ii\kappa(\xi)}\right|\leq 
	\ee^{-\pi\kappa(\xi)}=\sqrt{1+|\mathbf{R}(\xi)|^2}
   \end{equation}
that we can derive the boundedness of $\delta_0(\lambda)$. 
Together with the boundedness \eqref{bdd-delta} of 
$\det\pmb{\delta}(\lambda)$, we can also obtain the 
boundedness of $\delta_1(\lambda)$.
According to \cite[pp.1046-1047]{deift2002long}, we can obtain the following 
estimate for $\det\pmb{\delta}(\lambda)$
\begin{equation}\label{estimate-det-delta}
	\left|\det\pmb{\delta}(\lambda)-\delta_0(\lambda)\right|
	\lesssim \left|\ee^{\zeta(\lambda,\xi)}-1\right|
	=\left|\int_{0}^{1}\frac{\mathrm{d}}{\mathrm{d}s}
	\ee^{s\zeta(\lambda,\xi)}\mathrm{d}s\right|\lesssim 
	\left|\zeta(\lambda,\xi)\right|\max_{0\le s\le1}
	\left|\ee^{s\zeta(\lambda,\xi)}\right|\lesssim 
	\left|\lambda-\xi\right|^{1/2},
\end{equation}
as $\lambda\to\xi$ along the contours $\Sigma_2\cup\Sigma_4$. 
Here we use the $H^1$ property of $\beta(\lambda,\xi)$ and 
the properties of the logarithm to obtain 
\begin{equation}\label{estimate-zeta}
	\left|\zeta(\lambda,\xi)\right|\lesssim\left|\lambda-\xi\right|^{1/2}.
\end{equation}
By equation \eqref{estimate-det-delta}, we can derive the behavior of 
$\mathbf{R}(\xi)\pmb{\delta}(\lambda)$ 
as $\lambda\to\xi$ along the contours $\Sigma_2\cup\Sigma_4$, 
which is summarized in the following lemma.
\begin{lemma}\label{estimate-R-delta}
  As $\lambda\to\xi$ along the contour $\Sigma$, we have the following 
  estimates,
  \begin{equation}
    \begin{aligned}
      &\left|\mathbf{R}(\xi)\pmb{\delta}(\lambda)
    -\delta_0(\lambda)\pmb{A}(\xi)\right|
    \lesssim\left|\lambda-\xi\right|^{1/2},
    &\lambda\in\Sigma_2\cup\Sigma_4,\\
     &\left|\pmb{\delta}^{-1}(\lambda)\mathbf{R}^\dagger(\xi)
    -\delta_0^{-1}(\lambda)\pmb{A}^\dagger(\xi)\right|
    \lesssim\left|\lambda-\xi\right|^{1/2},
    &\lambda\in\Sigma_1\cup\Sigma_3,
    \end{aligned}
  \end{equation}
  where $\pmb{A}(\xi)$ is a constant defined 
  by equation \eqref{def-A-xi}.
\end{lemma}
\begin{proof}
  For the case on the contours $\Sigma_2\cup\Sigma_4$, 
  we first provide the estimate for 
  $\tilde{\pmb{\delta}}(\lambda)=\mathbf{R}(\xi)(\pmb{\delta}(\lambda)-\det\pmb{\delta}(\lambda))$. 
  It follows from equations \eqref{RHP-det-delta} and \eqref{RHP-delta} 
  that $\tilde{\pmb{\delta}}(\lambda)$ satisfies the following RHP
  \begin{equation}
    \begin{cases}\tilde{\pmb{\delta}}_+(\lambda)=\tilde{\pmb{\delta}}_-(\lambda)
      (1+|\mathbf{R}(\lambda)|^2)+\pmb{f}(\lambda)\pmb{\delta}_-(\lambda),
      &\quad \lambda\in(-\infty,\xi),\\
      \tilde{\pmb{\delta}}(\lambda)\to0,
      &\quad \lambda\to\infty,
    \end{cases}
  \end{equation}
  where 
  $\pmb{f}(\lambda)=\mathbf{R}(\xi)(\mathbf{R}^\dagger\mathbf{R}-\mathbf{R}\mathbf{R}^\dagger\mathbb{I}_2)(\lambda)$. 
  The solution $\tilde{\pmb{\delta}}(\lambda)$ for the above non-homogeneous RHP can be formulated as
  \begin{equation}\label{expression-tilde-delta}
    \begin{aligned}
      &\tilde{\pmb{\delta}}(\lambda)=\frac{X(\lambda)}{2\pi \ii}\int_{-\infty}^{\xi}
      \frac{\pmb{f}(s)\pmb{\delta}_-(s)}{X_+(s)(s-\lambda)}\mathrm{d}s,\\
      &X(\lambda)=\exp\left\{\frac1{2\pi \ii}\int_{-\infty}^{\xi}
      \frac{\ln\left(1+|\mathbf{R}(s)|^2\right)}{s-\lambda}\mathrm{d}s\right\}.
    \end{aligned}
  \end{equation}
 Moreover, $\tilde{\pmb{\delta}}(\lambda)$ can be rewritten by Cauchy 
 projection operator 
 \begin{equation}
  \tilde{\pmb{\delta}}(\lambda)=X(\lambda)
  C_{\Gamma}\pmb{g}(\lambda),\ \ 
  \pmb{g}(s)=X^{-1}_+(s)\pmb{f}(s)\pmb{\delta}_-(s).
 \end{equation}
 Here the contour $\Gamma$ denotes $(-\infty,\xi)$ and the 
 Cauchy projection operator $C_{\Gamma}\pmb{f}(\lambda)$ 
 is defined by 
\begin{equation}
  \begin{aligned}
    &C_{\Gamma}\pmb{f}(\lambda)=
  \frac{1}{2\pi \ii}\int_{\Gamma}\frac{\pmb{f}(s)}
  {s-\lambda}\dd s,
  \quad \lambda\in\mathbb{C}\setminus\Gamma,\\
  &C^\pm_{\Gamma}\pmb{f}(\lambda)=\lim\limits_{z\to\lambda_\pm}
    \frac{1}{2\pi \ii}\int_{\Gamma}\frac{\pmb{f}(s)}
    {s-z}\dd s,\quad \lambda\in\Gamma.
  \end{aligned}
\end{equation}
Then we show that $\pmb{g}(s)\in H^{1}(\Gamma)$. 
Since $\mathbf{R}(\lambda)$ lies in the Sobolev space $H^{2,1}(\mathbb{R})$ by Assumption \ref{assumption-q0}, 
$\pmb{f}(\lambda)$ also lies in $H^{2,1}(\mathbb{R})$. Together with 
equation \eqref{bdd-delta}, we can derive 
that $\pmb{g}(s)\in L^{2}(\Gamma)$. Then we consider the 
derivative $\pmb{g}^\prime(s)$ 
\begin{equation}
  \begin{aligned}
    \pmb{g}^\prime(s)&=-X^{-1}_+X^{\prime}_+X^{-1}_+(s)\pmb{f}(s)\pmb{\delta}_-(s)
    +X^{-1}_+(s)\pmb{f}^\prime(s)\pmb{\delta}_-(s)+
    X^{-1}_+(s)\pmb{f}(s)\pmb{\delta}_-^\prime(s)\\
    &=\mathrm{I}^{(1)}+\mathrm{I}^{(2)}+\mathrm{I}^{(3)}.
  \end{aligned}
\end{equation}
For $\mathrm{I}^{(1)}$, it follows from equation 
\eqref{expression-tilde-delta} that $X^{\prime}(s)$ satisfies 
the following RHP 
\begin{equation}
  	\begin{cases}
  	X^{\prime}_+(s)=
  	(1+\mathbf{R}\mathbf{R}^\dagger)X^{\prime}_-(s)
  	+n(s),\quad&s<\xi,\\
  	X^{\prime}(s)=\oo(s^{-2}), &s\rightarrow\infty,
  	\end{cases}
  \end{equation}
  where $n(s)=\left(\mathbf{R}^\prime\mathbf{R}^\dagger(s)+\mathbf{R}\mathbf{R}^{\dagger\prime}(s)\right)X_-(s)$. 
 Then the soluiton $X^{\prime}(s)$ can be determined in the 
 following form  
  \begin{equation}
  	X^{\prime}(s)=J(s)\left(\frac{1}{2\pi\ii}\int_{-\infty}^{\xi}
  	\frac{J_+^{-1}(x)n(x)}{x-s}\mathrm{d}x+
    \frac{A}{s-\xi}\right),
    \end{equation}
    where parameter $A$ will be determined below and 
    $J(s)$ satisfies the following RHP:
    \begin{equation}
  	\begin{cases}
      J_+(s)=
  	(1+\mathbf{R}\mathbf{R}^\dagger(s))J_-(s),\quad&s<\xi,\\
  	J(s)\rightarrow1, &s\rightarrow\infty.
  	\end{cases}
    \end{equation}
    By the uniqueness of the solution for RHP 
    \eqref{RHP-det-delta}, we have 
	$J(s)\equiv X(s)\equiv\det\pmb{\delta}(s)$. 
    Then it follows from 
    the geometric series expansion for $1/(x-s)$ that 
    \begin{equation}
      \frac{1}{2\pi\ii}\int_{-\infty}^{\xi}
  	\frac{J_+^{-1}(x)n(x)}{x-s}\mathrm{d}x
    =-\frac{1}{2\pi\ii}\int_{-\infty}^{\xi}
  	J_+^{-1}(x)n(x)\mathrm{d}x+\oo(s^{-2}),\,\ s\to\infty.
    \end{equation}
    According to the asymptotic behavior 
    $X^{\prime}(s)=\oo(s^{-2})$, we can derive that 
    \begin{equation}
      A=\frac{1}{2\pi\ii}\int_{-\infty}^{\xi}
  	J_+^{-1}(x)n(x)\mathrm{d}x=
    \frac{1}{2\pi\ii}\int_{-\infty}^{\xi}
  	\frac{\mathbf{R}^\prime\mathbf{R}^\dagger(s)+
    \mathbf{R}\mathbf{R}^{\dagger\prime}(s)}
    {1+\mathbf{R}\mathbf{R}^\dagger(s)}\mathrm{d}x=
    \frac{\ln\left(1+\mathbf{R}
    \mathbf{R}^\dagger(\xi)\right)}{2\pi \ii}.
    \end{equation}
    Then $\left\|\mathrm{I}^{(1)}\right\|_{L^2(\mathbb{R})}$ can 
    be decomposed into two parts 
    \begin{equation}
      \left\|\mathrm{I}^{(1)}\right\|_{L^2(\mathbb{R})}
      \lesssim \left\|C^+_{\Gamma} J_+^{-1}n(s)\right\|_{L^2(\mathbb{R})}
      +\left\|\frac{\pmb{f}(s)}{s-\xi}\right\|_{L^2(\mathbb{R})}.
    \end{equation}
    For the first part, we have 
    \begin{equation}
      \left\|C^+_{\Gamma} J_+^{-1}n(s)\right\|_{L^2(\mathbb{R})}
      \lesssim \left\|J_+^{-1}n(s)\right\|_{L^2(\Gamma)}
      \lesssim \left\|\mathbf{R}^\prime
	  \right\|_{L^2(\Gamma)}
    \end{equation}
    by the property of the Cauchy operator $C^\pm$ \cite{deift2002long}
    \begin{equation}\label{property-Cauchy}
      \left\|C^\pm_{\Gamma} h\right\|_{L^2(\mathbb{R})}
      \lesssim \left\| h\right\|_{L^2(\Gamma)}, 
      \quad \forall \ \ h\in L^2(\Gamma).
    \end{equation}
  Since $\pmb{f}(s)\in H^{2,1}(\mathbb{R})$ and 
  $\pmb{f}(\xi)=0$, we 
  can derive that 
  \begin{equation}
    \left|\pmb{f}(s)\right|=\left|\pmb{f}(s)-
    \pmb{f}(\xi)\right|\lesssim \left|s-\xi\right|.
  \end{equation}
  Hence we can estimate the second part by 
  \begin{equation}
    \left\|\frac{\pmb{f}(s)}{s-\xi}\right\|_{L^2(\mathbb{R})}^2
    \lesssim \int_{|s-\xi|<1}\frac{\left|\pmb{f}(s)\right|^2}
    {\left|s-\xi\right|^2}\mathrm{d}s+\int_{|s-\xi|>1}\frac{
      \left|\pmb{f}(s)\right|^2}
    {\left|s-\xi\right|^2}\mathrm{d}s
    \lesssim\int_{|s-\xi|<1}\frac{\left|s-\xi\right|^2}
    {\left|s-\xi\right|^2}\mathrm{d}s+
    \left\|\pmb{f}\right\|_{L^2(\mathbb{R})}^2\lesssim 1.
  \end{equation}
  Then we can conclude that $\mathrm{I}^{(1)}\in L^2(\mathbb{R})$. 
  By the fact that $\pmb{f}(s)\in H^{2,1}(\mathbb{R})$ and the 
  boundedness of $X^{-1}_+(s)\pmb{\delta}_-(s)$, we can also  
  obtain $\mathrm{I}^{(2)}\in L^2(\mathbb{R})$. 
  For $\mathrm{I}^{(3)}$, it follows from 
  equation \eqref{RHP-delta} that $\pmb{\delta}^{\prime}(s)$ 
  satisfies the following RHP:
    \begin{equation}
  	\begin{cases}
  	\pmb{\delta}^{\prime}_+(s)=
  	(\mathbb{I}_2+\mathbf{R}^\dagger\mathbf{R})\pmb{\delta}^{\prime}_-(s)
  	+\pmb{m}(s),\quad&s<\xi,\\
  	\pmb{\delta}^{\prime}(s)=\oo(s^{-2}), &s\rightarrow\infty.
  	\end{cases}
    \end{equation}
    where $\pmb{m}(s)=(\mathbf{R}^\dagger\mathbf{R}^{\prime}+(\mathbf{R}^\dagger)^{\prime}\mathbf{R})\pmb{\delta}_-(s)$.
    Analogous to the analysis for $X^{\prime}(s)$, 
    the solution $\pmb{\delta}^{\prime}(s)$ for 
    the above RHP can be expressed by
    \begin{equation}
  	\pmb{\delta}^{\prime}(s)=\frac{\pmb{Y}(s)}{2\pi\ii}\left(
    \int_{-\infty}^{\xi}
  	\frac{\pmb{Y}_+^{-1}(x)\pmb{m}(x)}{x-s}\mathrm{d}x+
    \frac{1}{s-\xi}\int_{-\infty}^{\xi}
  	\pmb{Y}_+^{-1}(x)\pmb{m}(x)\mathrm{d}x\right),
    \end{equation}
    where $\pmb{Y}(s)$ satisfies the following RHP:
    \begin{equation}
  	\begin{cases}
      \pmb{Y}_+(s)=
  	(\mathbb{I}_2+\mathbf{R}^\dagger\mathbf{R}(s))\pmb{Y}_-(s),\quad&s<\xi,\\
  	\pmb{Y}(s)\rightarrow\mathbb{I}_2, &s\rightarrow\infty.
  	\end{cases}
    \end{equation}
    By the uniqueness of the solution for RHP 
    \eqref{RHP-delta}, we have 
	$\pmb{\delta}(s)\equiv \pmb{Y}(s)$.
    Together with the property \eqref{property-Cauchy} of the Cauchy operator, 
    we can arrive at $\mathrm{I}^{(3)}\in L^2(\mathbb{R})$ by
    \begin{equation}
  	\begin{aligned}
  	  \left\|\mathrm{I}^{(3)}\right\|_{L^2(\mathbb{R})}
  	  &\lesssim \|C^-_{\Gamma}(\pmb{Y}_+^{-1}\pmb{m})(s)\|_{L^2(\mathbb{R})}
      +\left\|\frac{\pmb{f}(s)}{s-\xi}\right\|_{L^2(\mathbb{R})}\\
  	  &\lesssim \|\mathbf{R}^\dagger\mathbf{R}^{\prime}+
  	  (\mathbf{R}^\dagger)^{\prime}\mathbf{R}\|_{L^2(\Gamma)}+
  	  \left\|\pmb{f}\right\|_{L^2(\mathbb{R})}+
      \left(\int_{|s-\xi|<1}\frac{\left|s-\xi\right|^2}
    {\left|s-\xi\right|^2}\mathrm{d}s\right)^{1/2}.
  	\end{aligned}
    \end{equation}
    Consequently, we can derive that 
    $\pmb{g}^\prime(s)\in L^{2}(\mathbb{R})$ 
    and $\pmb{g}(s)\in H^{1}(\mathbb{R})$. 
    By the fact that $\pmb{g}(\xi)=0$ and Lemma 23.3 
	in \cite{beals1988direct}, we obtain 
    \begin{equation}\label{estimate-C-Gamma-g}
      \left|C_\Gamma\pmb{g}(\lambda)-
      C_\Gamma\pmb{g}(\xi)\right|=
      \left|\int_{\xi}^{\lambda}\frac{\mathrm{d}}
      {\mathrm{d}s}C_\Gamma\pmb{g}(s)\right|\lesssim 
      \left|\lambda-\xi\right|^{1/2}\left\|
        C_\Gamma\pmb{g}^\prime(s)\right\|_{L^2(\xi,\lambda)}
        \lesssim\left|\lambda-\xi\right|^{1/2}
        \left\|\pmb{g}^\prime(s)\right\|_{L^2(\Gamma)}.
    \end{equation}
  Then it follows from equation \eqref{estimate-C-Gamma-g} and 
  the estimate \eqref{estimate-det-delta} for 
  $\det\pmb{\delta}(\lambda)$ that 
  \begin{equation}
    \left|\tilde{\pmb{\delta}}(\lambda)-
    \delta_0(\lambda)C_\Gamma\pmb{g}(\xi)
    \right|\lesssim\left|\lambda-\xi\right|^{1/2}
  \end{equation}
  and hence 
  \begin{equation}
    \left|\mathbf{R}(\xi)\pmb{\delta}(\lambda)
    -\delta_0(\lambda)
    \left(\mathbf{R}(\xi)+C_\Gamma\pmb{g}(\xi)\right)\right|
    \lesssim\left|\lambda-\xi\right|^{1/2}.
  \end{equation}
  This implies that 
  \begin{equation}\label{def-A-xi}
	\pmb{A}(\xi)=\mathbf{R}(\xi)+C_\Gamma\pmb{g}(\xi).
  \end{equation}
  Then we can utilize the symmetry of $\pmb{\delta}(\lambda)$ to 
  provide the estimate for 
  $\pmb{\delta}^{-1}(\lambda)\mathbf{R}^\dagger(\xi)$ on the contours 
  $\Sigma_1\cup\Sigma_3$. It follows from equation 
  \eqref{symmetry-delta} that 
  \begin{equation}
    \pmb{\delta}^{-1}(\lambda)\mathbf{R}^\dagger(\xi)
    =\pmb{\delta}^\dagger(\lambda^*)\mathbf{R}^\dagger(\xi)
    =\left[\mathbf{R}(\xi)\pmb{\delta}(\lambda^*)\right]^\dagger
  \end{equation}
  and then 
  \begin{equation}
    \begin{aligned}
      \left|\pmb{\delta}^{-1}(\lambda)\mathbf{R}^\dagger(\xi)
      -\delta_0^{-1}(\lambda)\pmb{A}^\dagger(\xi)\right|
      &=\left|\left[\mathbf{R}(\xi)\pmb{\delta}(\lambda^*)-
      \delta_0(\lambda^*)\pmb{A}(\xi)\right]^\dagger\right|\\
      &\lesssim\left|\lambda-\xi\right|^{1/2},
      \ \ \lambda\in\Sigma_1\cup\Sigma_3.
    \end{aligned}
  \end{equation}
  Consequently we complete the proof.
\end{proof}

It follows form equations 
\eqref{estimate-det-delta} and Lemma 
\ref{estimate-R-delta} that we construct a solvable model RHP \ref{rhp-M4} 
whose solution $\M^{(4)}$ is a suitable approximation to $\M^{(3)}$.

\begin{rhp}\label{rhp-M4}
	Find an analytic function $\M^{(4)}$: 
	$\mathbb{C}\setminus\Sigma\rightarrow SL_{2}(\mathbb{C})$ 
	with the following properties
	\begin{itemize}
		\item[1.] $\M^{(4)}(\lambda)=\mathbb{I}_{3}+\mathcal{O}\left(\lambda^{-1}\right)$ as $\lambda\rightarrow\infty.$
		\item[2.] For each $\lambda\in\Sigma$ ,
		$\M^{(4)}$ takes continuous boundary values $\M^{(4)}_{\pm}(\lambda)$ which satisfy 
		the jump relation: $\M_+^{(4)}(\lambda)=\M_-^{(4)}(\lambda)\V^{(4)}(\lambda)$ where
		\begin{equation}\label{jump-parabolic-cylinder-model}
			\V^{(4)}(\lambda)=
			\begin{cases}
				\begin{pmatrix}
					1 & 0\\
					-\pmb{\delta}_0^{-2}\ee^{2\ii t\theta}
					\pmb{A}^\dagger(\xi) & \mathbb{I}_2
				\end{pmatrix},
				&\lambda \in\Sigma_1,\\
				\begin{pmatrix}
					1 & -\pmb{\delta}_0^2 \ee^{-2\ii t\theta}
					\frac{\pmb{A}(\xi)}{1+\mathbf{R}
					\mathbf{R}^\dagger(\xi)}\\
					0 & \mathbb{I}_2
				\end{pmatrix},
				&\lambda \in\Sigma_2,\\
				\begin{pmatrix}
					1 & 0\\
					-\pmb{\delta}_0^{-2}\ee^{2\ii t\theta}
					\frac{\pmb{A}^\dagger(\xi)}{1+\mathbf{R}
					\mathbf{R}^\dagger(\xi)} & \mathbb{I}_2
				\end{pmatrix},
				&\lambda \in \Sigma_3,\\
				\begin{pmatrix}
					1 & -\pmb{\delta}_0^2 \ee^{-2\ii t\theta}
					\pmb{A}(\xi)\\
					0 & \mathbb{I}_2
				\end{pmatrix},  
				&\lambda \in\Sigma_4.
			\end{cases}
		\end{equation}
	\end{itemize}
\end{rhp}
Next we will provide a proper estimate between jump matrices $\V^{(3)}(\lambda)$ and 
$\V^{(4)}(\lambda)$. Compared to the estimate in the 
scalar case \cite{deift2002long}, the main challenge lies in providing 
a proper estimate on $\pmb{\delta}$ without a closed form. 
To address this challenge, 
We decompose the estimate between 
$\V^{(3)}(\lambda)$ and $\V^{(4)}(\lambda)$ into two parts, as shown in equation 
\eqref{decompose-V3-V4}. The first part can be estimated 
following arguments similar to the scalar case \cite{deift2002long} and 
the second part is handled by applying Lemma 
\ref{estimate-R-delta}. 

\begin{lemma}\label{bdd-V3-V4}
	For $1\leq p\leq\infty$,
	\begin{equation*}
		\left\|\V^{(3)}-\V^{(4)}\right\|_{L^p(\Sigma)}\leq
		\frac c{t^{1/4+1/(2p)}}\quad\textit{uniformly for }
		t\geq1\mathrm{~and~all~}x\in\mathbb{R}\mathrm{~.}
	\end{equation*}
\end{lemma}
\begin{proof}
	We give the details of the estimate on the contour $\Sigma_4$ only, the situation 
	on the other rays in $\Sigma$ is similar.
	For $\lambda=\xi+ue^{-\ii\pi/4},u\ge0$,
	
	\begin{equation}\label{decompose-V3-V4}
		\begin{aligned}
			&\quad\left|\V^{(3)}(\lambda)-
			\V^{(4)}(\lambda)\right|\\
			&\le\left|\det\pmb{\delta}(\lambda)[\mathbf{R}](\lambda)\pmb{\delta}(\lambda)  
			\ee^{-2\ii t\theta}-\pmb{\delta}_0(\lambda)\mathbf{R}(\xi)\pmb{\delta}(\lambda)
			\ee^{-2\ii t\theta}\right|\\
			&+\left|\pmb{\delta}_0(\lambda)\mathbf{R}(\xi)\pmb{\delta}(\lambda)
			\ee^{-2\ii t\theta}
			-\mathbf{R}(\xi)\pmb{\delta}_0^2(\lambda)\pmb{A}(\xi)\ee^{-2\ii t\theta}\right|
			\equiv \mathrm{I}+\mathrm{II}.
		\end{aligned}
	\end{equation}
    For $\mathrm{I}$, we have 
	\begin{equation}
		\left|\mathrm{I}\right|\lesssim 
		\left|\frac{\pmb{\delta}_1(\lambda)}
		{(1+\ii(\lambda-\xi))^2}-1\right|\ee^{-2
		t\left|\lambda-\xi\right|^2}
		\lesssim \ee^{-2
		t\left|\lambda-\xi\right|^2}
		\left(\left|\pmb{\delta}_1(\lambda)-1\right|
		+\left|(\lambda-\xi)^2-2\ii(\lambda-\xi)\right|
		\right).
	\end{equation}
    It follows from equation 
	\eqref{estimate-zeta} and the fact 
	$\sup_{u>0}\,\ (2u+u^2)\ee^{-u^2/2}<c$ that 
	\begin{equation}
		\left|\mathrm{I}\right|\lesssim 
		\left|\lambda-\xi\right|^{1/2}
		\ee^{-2t\left|\lambda-\xi\right|^2}+t^{-1/2}
		\ee^{-t\left|\lambda-\xi\right|^2}
		\lesssim t^{-1/4}
		\ee^{-t\left|\lambda-\xi\right|^2}.
	\end{equation}
    By Lemma \ref{estimate-R-delta}, we can derive 
	\begin{equation}
		\left|\mathrm{II}\right|\lesssim
		\left|\lambda-\xi\right|^{1/2}
		\ee^{-2t\left|\lambda-\xi\right|^2}
		\lesssim t^{-1/4}
		\ee^{-t\left|\lambda-\xi\right|^2}.
	\end{equation}
	Hence we have 
	\begin{equation}
	\|\V^{(3)}(\lambda)-\V^{(4)}(\lambda)\|_{L^p(\Sigma_4)}
	\lesssim  t^{-1/4}\left(\int_{\Sigma_4}
	\ee^{-2tp\left|\lambda-\xi\right|^2}\mathrm{d}\lambda
	\right)^{1/p}
	\lesssim t^{-1/4-1/2p},\quad 1\le p\le\infty.
	\end{equation}
	Then the result follows immediately and 
	we complete the proof.
\end{proof}

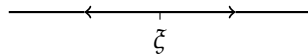
\begin{figure}[h]
	\centering
	\begin{tikzpicture}
		\vspace{0.5cm} 
		\draw[thick] (-2, 0) -- (-1, 0) ;
		\draw[->,thick] (0, 0) -- (-1, 0) ;
		\draw[->,thick] (0, 0) -- (1, 0) ;
		\draw[thick] (1, 0) -- (2, 0) ;
		
		\foreach \x/\label in {0/{$\xi$}} 
		\draw (\x, 0) -- (\x, -0.1) node[below] {\label}; 
	\end{tikzpicture}
	\caption{\small{The contours $\mathbb{R}_\xi$.}}
	\label{initial-contour}
\end{figure}
By Lemma \ref{bdd-V3-V4}, we can show that the 
operator $(1-C_{\mathbf{w}_3})^{-1}$ exists and is bounded in $L^2(\Sigma)$, which 
guarantees the existence of solution $\M^{(3)}$.

\begin{lemma}
	For sufficiently large $t>t_0$ , 
	the operator $(1-C_{\mathbf{w}_3})^{-1}$ exists and is bounded in $L^2(\Sigma)$ 
	for all $x\in\mathbb{R}$.
\end{lemma}

\begin{proof}
	By \cite[pp.1049]{deift2002long}, the bound on $(1-C_{\mathbf{w}_4})^{-1}$ in $L^2(\Sigma)$ is equivalent to 
	the bound on $(1-C_{\tilde{\mathbf{w}}})^{-1}$ in $L^2(\mathbb{R})$, where
	\begin{equation*}
		\tilde{\V}=
		\begin{pmatrix}
			1+|\mathbf{R}(\xi)|^2&-\mathbf{R}(\xi)\ee^{-2\ii t\theta}\\
			-\mathbf{R}^\dagger(\xi)\ee^{2\ii t\theta}&\mathbb{I}_2
		\end{pmatrix}.
	\end{equation*}
	But the bound on $(1-C_{\tilde{\mathbf{w}}})^{-1}$ can be ensured by the same argument as 
	$(1-C_{\mathbf{w}})^{-1}$ in Lemma \ref{bdd-w}. Hence $(1-C_{\mathbf{w}_4})^{-1}$ is bounded in $L^2(\Sigma)$.
	Set $\mathbf{w}_{3}=(\mathbf{w}^{-}_3=0,\mathbf{w}^{+}_3=\V^{(3)}-\mathbb{I}_3)$ and 
	$\mathbf{w}_{4}=(\mathbf{w}^{-}_4=0,\mathbf{w}^{+}_4=\V^{(4)}-\mathbb{I}_3)$. 
	By Lemma \ref{bdd-V3-V4},
	\begin{equation*}
		\|C_{\mathbf{w}_3}-C_{\mathbf{w}_4}\|_{L^2}
		=\sup_{\|\mathbf{f}\|_{L^2}=1}\|C^-_{\Sigma}\mathbf{f}(\V^{(3)}-\V^{(4)})\|_{L^2}\lesssim 
		\|\V^{(3)}-\V^{(4)}\|_{L^\infty}\lesssim t^{-\frac{1}{4}}.
	\end{equation*}

	Then we conclude by the resolvent identity that $(1-C_{\mathbf{w}_3})^{-1}$ exists in $L^2(\Sigma)$ 
	for sufficiently large $t>t_0$ and that
	\begin{equation*}
		\left\|(1-C_{{\mathbf{w}_{4}}})^{-1}\right\|_{{L^{2}(\Sigma)}}\lesssim1,
		\quad\left\|(1-C_{{\mathbf{w}_3}})^{-1}-(1-C_{{\mathbf{w}_4}})^{-1}\right\|_{{L^{2}}}
		\lesssim t^{-\frac{1}{4}},
	\end{equation*}
	uniformly for all $x\in\mathbb{R}$ and $t>t_0$. By \cite[Proposition 2.6 and Corollary 2.7]{deift2002long}, $(1-C_{\mathbf{w}_3})^{-1}$ exists and is bounded in $L^2(\Sigma)$ for any 
	factorization and hence the result follows immediately.
\end{proof}

From the previous results,  we observe that for sufficiently large $t>t_0$,  
$(1-C_{\mathbf{w}_i})^{-1}$ exists and is bounded in $L^2$, for $i=1,2,3,4$.

We now reverse the orientation
on $(-\infty,\xi)$ of RHP \ref{rhp-M1}, as shown in figure \ref{initial-contour} to obtain a contour 
$\mathbb{R}_\xi=(\ee^{\ii\pi}\mathbb{R}_++\xi)\cup(\mathbb{R}_++\xi)$ with associated jump matrix
$\hat{\V}^{(1)}=\V^{(1)}$ for $\lambda>\xi$ and $\hat{\V}^{(1)}=\left(\V^{(1)}\right)^{-1}$ 
for $\lambda<\xi$. 
From the form of $\hat{\V}^{(1)}$, we see that 
$\hat{\V}^{(1)}=(\mathbb{I}_{3}-\hat{\mathbf{w}}_1^-)^{-1}(\mathbb{I}_{3}+\hat{\mathbf{w}}_1^+)$, where
\begin{equation*}
	\begin{aligned}
		\hat{\mathbf{w}}_1
		&=(\hat{\mathbf{w}}^-_1,\hat{\mathbf{w}}^+_1)\\
		&=\begin{cases}
			\left(
			\begin{pmatrix}
				0&-\det\pmb{\delta} \ee^{-2\ii t\theta}\mathbf{R}(\lambda)\pmb{\delta}\\
				0&0\end{pmatrix},
			\begin{pmatrix}0&0\\
				-(\det\pmb{\delta})^{-1}\ee^{2\ii t\theta}\pmb{\delta}^{-1}\mathbf{R}^\dagger(\lambda)&0
			\end{pmatrix}
			\right),&\lambda>\xi ,\\
			\left(
			\begin{pmatrix}
				0&\det\pmb{\delta}_+\ee^{-2\ii t\theta}
				\frac{\mathbf{R}(\lambda)}{1+\mathbf{R}\mathbf{R}^\dagger}\pmb{\delta}_+\\
				0&0
			\end{pmatrix},
			\begin{pmatrix}
				0&0\\
				(\det\pmb{\delta}_-)^{-1}\ee^{2\ii t\theta}\pmb{\delta}_-^{-1}
				\frac{\mathbf{R}^\dagger(\lambda)}{1+\mathbf{R}\mathbf{R}^\dagger}&0
			\end{pmatrix}
			\right),&\lambda<\xi .
		\end{cases}
	\end{aligned}
\end{equation*}
Observe that if we reverse the orientation as above 
for the RHP \ref{rhp-M2}, we obtain a new RHP with 
contour $\mathbb{R}_\xi$ and jump matrix 
$\hat{\V}^{(2)}=(\mathbb{I}_{3}-\hat{\mathbf{w}}_2^-)^{-1}(\mathbb{I}_{3}+\hat{\mathbf{w}}_2^+)$ where 
$\hat{\mathbf{w}}_2=(\hat{\mathbf{w}}_2^-,\hat{\mathbf{w}}_2^+)$ is the same as $\hat{\mathbf{w}}_1$ except $\mathbf{R}$ 
is replaced by $[\mathbf{R}]$, etc. 
Then we can obtain the following expansion of $\hat{\M}^{(i)}$, for $i=1,2$,
\begin{equation}
	\hat{\M}^{(i)}(\lambda)=\mathbb{I}_3+\frac{\hat{\M}^{(i)}_1}{\lambda}
	+\oo(\lambda^{-2}),\quad \lambda\rightarrow\infty.
\end{equation}
Moreover, the extension $\hat{\M}^{(i)}$ of $\hat{\M}^{(i)}_\pm$ off $\mathbb{R}_\xi$ 
is the same as the extension $\M^{(i)}$ of $\M^{(i)}_\pm$ off $\mathbb{R}$, for $i=1,2$. 
Thus we can deduce
$\hat{\M}^{(i)}_1=\M^{(i)}_1$, for $i=1,2$. 
By Proposition 2.8 in \cite{deift2002long}, we can also obtain 
$C_{\hat{\mathbf{w}}_1}=C_{\mathbf{w}_1}$ and hence 
$\|(1-C_{\hat{\mathbf{w}}_i})^{-1}\|_{L^2(\mathbb{R}_\xi)}\lesssim1$, for $i=1,2$.
Extend $\mathbb{R}_\xi$ to a contour 
$\Gamma_{\xi}=\mathbb{R}\cup(\xi+\ee^{\ii\pi/2}\mathbb{R}_-)\cup(\xi+\ee^{-\ii\pi/2}\mathbb{R}_-)$ 
as shown in figure \ref{extend-contours}. 
As a complete contour, $\Gamma_{\xi}$ has the important property \cite{deift2002long}
\begin{equation}\label{complete-contour}
	C_{\Gamma_{\xi}}^+C_{\Gamma_{\xi}}^-
	=C_{\Gamma_{\xi}}^-C_{\Gamma_{\xi}}^+=0.
\end{equation}

\begin{figure}[h]
	\centering
	\begin{tikzpicture}
		\vspace{0.5cm} 
		\draw[thick] (-2, 0) -- (-1, 0) ;
		\draw[->,thick] (0, 0) -- (-1, 0) ;
		\draw[->,thick] (0, 0) -- (1, 0) ;
		\draw[thick] (1, 0) -- (2, 0) ;
		
		\draw[thick] (0, 1) -- (0, 0) ;
		\draw[->,thick] (0, 2) -- (0, 1) ;
		\draw[->,thick] (0, -2) -- (0, -1) ;
		\draw[thick] (0, -1) -- (0, 0) ;
		
		\node at (0.3, -0.3)  {$\xi$};
	\end{tikzpicture}
	\caption{\small{The contours $\Gamma_\xi$.}}
	\label{extend-contours}
\end{figure}
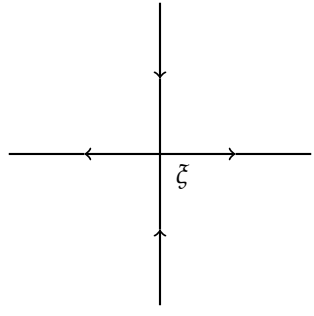
		
		

Our goal is to prove that $\hat{\M}^{(1)}_1$ is close to $\hat{\M}^{(2)}_1$ with the above
property \eqref{complete-contour} of complete contour and the following lemmas.

\begin{lemma}\label{lemma1-CNLS}
	Suppose $1\times2$ vector function $\mathbf{f}\in H^{1,1}(\mathbb{R})$. Then for any $2\leq p<\infty$ and for all $t\ge0$,
	\begin{equation*}
		\begin{cases}
			\|C^-_{(\xi,+\infty)\rightarrow\Gamma_\xi}
			(\det\pmb{\delta})^{-1}\ee^{2\ii t\theta}\pmb{\delta}^{-1}\mathbf{f}^\dagger\|_{L^p}
			\lesssim t^{-1/(2p)},\\
			\|C^-_{(-\infty,\xi)\rightarrow\Gamma_\xi}
			(\det\pmb{\delta}_-)^{-1}\ee^{2\ii t\theta}\pmb{\delta}_-^{-1}
			\mathbf{f}^\dagger|_{L^p}
			\lesssim t^{-1/(2p)},\\
			\|C^+_{(\xi,+\infty)\rightarrow\Gamma_\xi}
			\det\pmb{\delta} \ee^{-2\ii t\theta}\mathbf{f}\pmb{\delta}\|_{L^p}
			\lesssim t^{-1/(2p)},\\
			\|C^+_{(-\infty,\xi)\rightarrow\Gamma_\xi}
			\det\pmb{\delta}_+\ee^{-2\ii t\theta}
			\mathbf{f}\pmb{\delta}_+\|_{L^p}
			\lesssim t^{-1/(2p)},
		\end{cases}
	\end{equation*}
	where $\Gamma_{\xi}=\mathbb{R}\cup(\xi+\ee^{\ii\pi/2}\mathbb{R}_-)\cup(\xi+\ee^{-\ii\pi/2}\mathbb{R}_-)$.
\end{lemma}

\begin{lemma}\label{lemma2-CNLS}
	Suppose $1\times2$ vector function $\mathbf{f}\in H^{1,1}(\mathbb{R})$, $\mathbf{f}(\xi)=0$ and 
	g is a $2\times2$ vector function in the Hardy space $H^q(\mathbb{C}\setminus\mathbb{R})$. 
	Then for any $2\leq q\le\infty$ and for all $t\ge0$,
	\begin{equation*}
		\begin{cases}
			\|C^-_{(\xi,+\infty)\rightarrow\Gamma_\xi}
			\ee^{2\ii t\theta}g_+\mathbf{f}^\dagger\|_{L^2}
			\lesssim t^{-1/2+1/q}\|g\|_{H^q(\mathbb{C}\setminus\mathbb{R})},\\
			\|C^-_{(-\infty,\xi)\rightarrow\Gamma_\xi}
			\ee^{-2\ii t\theta}fg_-|_{L^2}
			\lesssim t^{-1/2+1/q}\|g\|_{H^q(\mathbb{C}\setminus\mathbb{R})},\\
			\|C^+_{(\xi,+\infty)\rightarrow\Gamma_\xi}
			\ee^{-2\ii t\theta}fg_-\|_{L^2}
			\lesssim t^{-1/2+1/q}\|g\|_{H^q(\mathbb{C}\setminus\mathbb{R})},\\
			\|C^+_{(-\infty,\xi)\rightarrow\Gamma_\xi}
			\ee^{2\ii t\theta}g_+\mathbf{f}^\dagger\|_{L^2}
			\lesssim t^{-1/2+1/q}\|g\|_{H^q(\mathbb{C}\setminus\mathbb{R})}.
		\end{cases}
	\end{equation*}
\end{lemma}
The proofs of Lemma \ref{lemma1-CNLS} and Lemma \ref{lemma2-CNLS} 
are detailed in Appendix A for the sake of completeness. 
Then we need the following $L^P$ bound on $(1-C_{\mathbf{w}_1})^{-1}$.

\begin{lemma}\label{Lp-w1}\cite[Proposition 4.5]{deift2002long}
	Suppose that $\mathbf{R}$ is a continuous function on
	$\mathbb{R}$ such that
	\begin{equation*}
		\lim_{\lambda\to\infty}\mathbf{R}(\lambda)=0\quad\mathrm{and}
		\quad\|\mathbf{R}\|_{L^\infty(\mathbb{R})}<\infty.
	\end{equation*}
	Then for any $p\ge2$, there exists $t_0$ such that for $t\ge t_0$ and 
	$x\in\mathbb{R}$, $(1-C_{\mathbf{w}})^{-1}$ exists in $L^p(\mathbb{R})$ and 
	\begin{equation*}
		\|(1-C_{\mathbf{w}})^{-1}\|_{L^p(\mathbb{R})\rightarrow L^p(\mathbb{R})}\lesssim1.
	\end{equation*}
\end{lemma} 
By Lemma \ref{bbd-w1}, Lemma \ref{Lp-w1} and Proposition 2.8 in \cite{deift2002long}, 
we conclude that for $\mathbf{R}(\lambda)\in H^{1,1}(\mathbb{R})$ and $2\le p<\infty$,
\begin{equation}\label{lp-hat-w1}
	\|(1-C_{\hat{\mathbf{w}}_1})^{-1}\|_{L^p(\mathbb{R}_\xi)\rightarrow L^p(\mathbb{R}_\xi)}\lesssim1.
\end{equation}
Using equation \eqref{lp-hat-w1}, we obtain the following results.

\begin{lemma}\label{lemma3-CNLS}
	For any $2\leq p<\infty$ and for t sufficiently large,
	\begin{equation*}
		\|\hat{\pmb{\mu}}_1-\mathbb{I}_3\|_{L^p}\lesssim t^{-1/(2p)}.
	\end{equation*}
\end{lemma}
\begin{proof}
	The result follows from Lemma \ref{lemma1-CNLS} 
	\begin{equation*}
		\begin{aligned}
			\|\hat{\pmb{\mu}}_1-\mathbb{I}_3\|_{L^p}
			&=\|(1-C_{\hat{\mathbf{w}}_1})^{-1}(C_{\hat{\mathbf{w}}_1}\mathbb{I}_3)\|_{L^p}\\
			&\le \|(1-C_{\hat{\mathbf{w}}_1})^{-1}\|_{L^p(\mathbb{R}_\xi)\rightarrow L^p(\mathbb{R}_\xi)}
			\|C_{\hat{\mathbf{w}}_1}\mathbb{I}_3\|_{L^p}\lesssim t^{-1/(2p)}.
		\end{aligned}
	\end{equation*}
\end{proof}

\begin{lemma}\label{lemma4-CNLS}
	For any $2\leq p<\infty$ and for t sufficiently large,
	\begin{equation*}
		\|C^\pm_{\mathbb{R}_\xi}\hat{\pmb{\mu}}_1(\hat{\mathbf{w}}_1^\mp-\hat{\mathbf{w}}_2^\mp)\|_{L^2}
		\lesssim t^{-\frac{1}{2}+\frac{1}{2p}}.
	\end{equation*}
\end{lemma} 
\begin{proof}
	By triangularity and $\hat{\pmb{\mu}}_1=\mathbb{I}_3+C_{\hat{\mathbf{w}}_1}\hat{\pmb{\mu}}_1$,
	\begin{equation*}
		\begin{aligned}
			C^\pm_{\mathbb{R}_\xi}\hat{\pmb{\mu}}_1(\hat{\mathbf{w}}_1^\mp-\hat{\mathbf{w}}_2^\mp)
			&=C^\pm_{\mathbb{R}_\xi}(\hat{\mathbf{w}}_1^\mp-\hat{\mathbf{w}}_2^\mp)
			C^\pm_{\mathbb{R}_\xi}(C_{\hat{\mathbf{w}}_1}\hat{\pmb{\mu}}_1)(\hat{\mathbf{w}}_1^\mp-\hat{\mathbf{w}}_2^\mp)\\
			&=C^\pm_{\mathbb{R}_\xi}(\hat{\mathbf{w}}_1^\mp-\hat{\mathbf{w}}_2^\mp)
			+C^\pm_{\mathbb{R}_\xi}(C^\mp\pmb{\mu}_1\mathbf{w}_1^\pm)(\hat{\mathbf{w}}_1^\mp-\hat{\mathbf{w}}_2^\mp)
			\equiv\mathrm{I}+\mathrm{II}.
		\end{aligned}
	\end{equation*}
	By Lemma \ref{lemma2-CNLS}, 
	we obtain that $\|\mathrm{I}\|_{L^2}\lesssim t^{-1/2}$. For $\mathrm{II}$, 
	by Lemma \ref{lemma3-CNLS} and Lemma \ref{lemma1-CNLS},
	\begin{equation*}
		\begin{aligned}
			\|C^\mp_{\mathbb{R}_\xi}\hat{\pmb{\mu}}_1\hat{\mathbf{w}}_1^\pm\|_{L^2}
			&\le\|C^\mp_{\mathbb{R}_\xi}(\hat{\pmb{\mu}}_1-\mathbb{I}_3)\hat{\mathbf{w}}_1^\pm\|_{L^2}+
			\|C^\mp_{\mathbb{R}_\xi} \hat{\mathbf{w}}_1^\pm\|_{L^2}\\
			&\lesssim \|\hat{\pmb{\mu}}_1-\mathbb{I}_3\|_{L^2}+t^{-1/(2p)}
			\lesssim t^{-1/(2p)}.
		\end{aligned}
	\end{equation*}
	Together with Lemma \ref{lemma2-CNLS}, we obtain that 
	$\|\mathrm{II}\|_{L^2}\lesssim t^{-1/2+1/p}t^{-1/2p}\lesssim t^{-\frac{1}{2}+\frac{1}{(2p)}}$. 
	The result follows immediately.
\end{proof}
Next we can show that $\int_{\mathbb{R}_\xi}\hat{\pmb{\mu}}_2 \hat{\mathbf{w}}_2$ is a good enough 
approximation to $\int_{\mathbb{R}_\xi}\hat{\pmb{\mu}}_1 \hat{\mathbf{w}}_1$ for large $t$.

\begin{lemma}\label{bdd-M1-M2}
	For any $2\leq p<\infty$ and for t sufficiently large,
	\begin{equation*}
		\left|\int_{\mathbb{R}_\xi}\hat{\pmb{\mu}}_1 \hat{\mathbf{w}}_1-
		\int_{\mathbb{R}_\xi}\hat{\pmb{\mu}}_2 \hat{\mathbf{w}}_2\right|
		\lesssim t^{{-\frac{3}{4}+\frac{1}{(2p)}}}.
	\end{equation*}
\end{lemma} 
A proof of Lemma \ref{bdd-M1-M2} is provided in Appendix A for the sake of completeness.

Set $\mathbf{w}_{3}=(\mathbf{w}^{-}_3=0,\mathbf{w}^{+}_3=\V^{(3)}-\mathbb{I}_3)$ and 
$\mathbf{w}_{4}=(\mathbf{w}^{-}_4=0,\mathbf{w}^{+}_4=\V^{(4)}-\mathbb{I}_3)$ 
on $\Sigma$. Then we will show that $\int_\Sigma\pmb{\mu}_4\mathbf{w}_4$ is a good enough 
approximation to $\int_\Sigma\pmb{\mu}_3\mathbf{w}_3$ for large $t$.

\begin{lemma}\label{bdd-M3-M4}
	For t sufficiently large,
	\begin{equation*}
		\left|\int_{\Sigma}\pmb{\mu}_3 \pmb{w}_3-\int_{\Sigma}\pmb{\mu}_4 \pmb{w}_4\right|
		\lesssim t^{{-\frac{3}{4}}}.
	\end{equation*}
\end{lemma} 

\begin{proof}
	We decompose the estimate into three parts as follows:
	\begin{equation}
		\begin{aligned}
			\int_{\Sigma}\pmb{\mu}_3 \pmb{w}_3-\int_{\Sigma}\pmb{\mu}_4 \pmb{w}_4
			& =\int_{\Sigma} \pmb{w}_3-\pmb{w}_4+
			\int_{\Sigma}(\pmb{\mu}_4-\mathbb{I}_{3})(\pmb{w}_3-\pmb{w}_4)+\int_{\Sigma}(\pmb{\mu}_3-\pmb{\mu}_4)\pmb{w}_3 \\
			&\equiv\mathrm{I}+\mathrm{II}+\mathrm{III}.
		\end{aligned}
	\end{equation}
	For $\mathrm{I}$, by the Lemma \ref{bdd-V3-V4},
	\begin{equation*}
		|\mathrm{I}|\le\left\|\V^{(3)}-\V^{(4)}\right\|_{L^1(\Sigma)}\lesssim t^{-3/4}.
	\end{equation*}
	For $\mathrm{II}$, we have
	\begin{equation*}
		\|\pmb{\mu}_4-\mathbb{I}_{3}\|_{L^2(\Sigma)}=\|C^-_{\Sigma}(\V^{(4)}-\mathbb{I}_3)\|_{L^2(\Sigma)}
		\lesssim\|\V^{(4)}-\mathbb{I}_3\|_{L^2(\Sigma)}.
	\end{equation*}
	we consider $\lambda=\xi+ue^{-\ii\pi/4}\in\Sigma_4$ only, the situation 
	on the other rays in $\Sigma$ is similar. 
	By $|\pmb{\delta}_0|=|\ee^{\ii\kappa(\xi)\ln(\lambda-\xi)}|=|\ee^{\kappa(\xi)\pi/4}|\lesssim1$, 
	\begin{equation*}
		\|\V^{(4)}-\mathbb{I}_3\|_{L^2(\Sigma_4)}\lesssim(\int_{0}^{+\infty}\ee^{-2tu^2}\dd u)^{1/2}
		\lesssim t^{-1/4}.
	\end{equation*}
	Therefore 
	$\|\pmb{\mu}_4-\mathbb{I}_{3}\|_{L^2(\Sigma)}\lesssim\|\V^{(4)}-\mathbb{I}_3\|_{L^2(\Sigma)}\lesssim t^{-1/4}$. 
	Then 
	\begin{equation*}
		\begin{aligned}
			|\mathrm{II}|
			&\le\|\pmb{\mu}_4-\mathbb{I}_{3}\|_{L^2(\Sigma)}\|\pmb{w}_3-\pmb{w}_4\|_{L^2(\Sigma)}\\
			&\lesssim t^{-1/4}\left\|\V^{(3)}-\V^{(4)}\right\|_{L^2(\Sigma)}\\
			&\lesssim t^{-3/4}.
		\end{aligned}
	\end{equation*}
	For $\mathrm{III}$, we compute
	\begin{equation*}
		\begin{aligned}
			\pmb{\mu}_3-\pmb{\mu}_4&=(1-C_{\pmb{w}_3})^{-1}C_{\pmb{w}_3-\pmb{w}_4}\pmb{\mu}_4\\
			&=(1-C_{\pmb{w}_3})^{-1}C_{\Sigma}^-\pmb{\mu}_4(\V^{(3)}-\V^{(4)}),
		\end{aligned}
	\end{equation*}
	then by Lemma \ref{bdd-V3-V4} 
	\begin{equation*}
		\begin{aligned}
			\|\pmb{\mu}_3-\pmb{\mu}_4\|_{L^2(\Sigma)}
			&\lesssim\|\pmb{\mu}_4(\V^{(3)}-\V^{(4)})\|_{L^2(\Sigma)}\\
			&\lesssim\|(\pmb{\mu}_4-\mathbb{I}_3)(\V^{(3)}-\V^{(4)})\|_{L^2(\Sigma)}
			+\|\V^{(3)}-\V^{(4)}\|_{L^2(\Sigma)}\\
			&\lesssim\|(\pmb{\mu}_4-\mathbb{I}_3)\|_{L^2(\Sigma)}
			\|\V^{(3)}-\V^{(4)}\|_{L^\infty(\Sigma)}+t^{-1/2}\\
			&\lesssim t^{-1/4}t^{-1/4}+t^{-1/2}\lesssim t^{-1/2}.
		\end{aligned}
	\end{equation*}
	Therefore by $\|\V^{(3)}-\mathbb{I}_3\|_{L^2(\Sigma_4)}\lesssim t^{-1/4}$ 
	as $\|\V^{(4)}-\mathbb{I}_3\|_{L^2(\Sigma_4)}$,
	\begin{equation*}
		|\mathrm{III}|\le\|\pmb{\mu}_3-\pmb{\mu}_4\|_{L^2(\Sigma)}\|\pmb{w}_3\|_{L^2(\Sigma)}
		\lesssim t^{-1/2}t^{-1/4}\lesssim t^{-3/4}.
	\end{equation*}
	Adding the estimates for $\mathrm{I}$, $\mathrm{II}$, and $\mathrm{III}$, the result follows.
\end{proof}

By Beals and Coifman \cite{beals1984scattering}, we know that 
\begin{equation*}
	\M^{(4)}(x,t)=\mathbb{I}_{3}+C_{\Sigma}(\pmb{\mu}_4(\pmb{w}^-_4 + \pmb{w}^+_4))
\end{equation*}
solves RHP \ref{rhp-M4}, where $\pmb{\mu}_4$ is the solution of 
$(1-C_{\pmb{w}_4})\pmb{\mu}_4=\mathbb{I}_{3}$. Thus, we can obtain the expansion of 
$\M^{(4)}(\lambda;x,t)$
\begin{equation}\label{residue-M4}
	\begin{aligned}
		&\M^{(4)}(\lambda;x,t)=\mathbb{I}_3+
		\frac{\M_1^{(4)}(x,t)}{\lambda}+\oo(\lambda^{-2}),\ \ \lambda\to\infty\\
		&\M_1^{(4)}(x,t)=\frac{\ii}{2\pi}\int_{\Sigma}\pmb{\mu}_4 \pmb{w}_4.
	\end{aligned}
\end{equation}
But the RHP \ref{rhp-M4} can be solved explicitly in the terms of parabolic 
cylinder function, which appears frequently in the literature of long-time asymptotic calculations for 
integrable nonlinear waves \cite{dieng2008long,grunert2009long,jenkins2017global,liujiaqi2018long}. The jump matrices \eqref{jump-parabolic-cylinder-model} 
are exactly the jump matrices of the parabolic cylinder model problem 
which was first introduced and solved in \cite{its1981asymptotics}, and later applied to the mKdV equation in \cite{deift1993steepest}.
Set $y=\varphi(\lambda)=2\sqrt{t}(\lambda-\xi)$ and 
$\M^{(5)}(y;x,t)=\M^{(4)}(\varphi^{-1}(y);x,t)$, 
it follows RHP 3.4 that $\M^{(5)}(y):=\M^{(5)}(y;x,t)$ satisfies the following RHP.

\begin{rhp}\label{solvable-RHP-model}
	Find an analytic function $\M^{(5)}$: 
	$\mathbb{C}\setminus\Sigma^{\prime}\rightarrow SL_{2}(\mathbb{C}),\Sigma^{\prime}=\cup_{k=1}^{4}\ee^{\frac{(2k-1)\pi\ii}{4}}\mathbb{R}_+$ 
	with the following properties
	\begin{itemize}
		\item[1.] $\M^{(5)}(y)=\mathbb{I}_{3}+\mathcal{O}\left(y^{-1}\right)$ as $y\rightarrow\infty.$
		\item[2.] For each $y\in\Sigma^\prime$ ,
		$\M^{(5)}$ takes continuous boundary values $\M^{(5)}_{\pm}(y)$ which satisfy 
		the jump relation: $\M_+^{(5)}(y)=\M_-^{(5)}(y)\V^{(5)}(y)$ where
		\begin{equation*}
			\V^{(5)}(y)=
			\begin{cases}
				\begin{pmatrix}
					1 & 0\\
					-\eta^{-2} y^{-2\ii\kappa(\xi)} 
					\ee^{\frac{1}{2}\ii y^2} 
					\pmb{A}^\dagger(\xi) & \mathbb{I}_2
				\end{pmatrix},  
				&y \in \ee^{\frac{\ii\pi}{4}}\mathbb{R}_+,\\
				\begin{pmatrix}
					1 & -\eta^2 y^{2\ii\kappa(\xi)} \ee^{-\frac{1}{2}\ii y^2}
					\frac{\pmb{A}(\xi)}{1+\mathbf{R}\mathbf{R}^\dagger(\xi)}\\
					0 & \mathbb{I}_2
				\end{pmatrix},
				&y \in \ee^{-\frac{\ii\pi}{4}}\mathbb{R}_-,\\
				\begin{pmatrix}
					1 & 0\\
					-\eta^{-2}y^{-2\ii\kappa(\xi)} \ee^{\frac{1}{2}\ii y^2}
					\frac{\pmb{A}^\dagger(\xi)}{1+\mathbf{R}\mathbf{R}^\dagger(\xi)} & \mathbb{I}_2
				\end{pmatrix},
				&y \in \ee^{\frac{\ii\pi}{4}}\mathbb{R}_-,\\
				\begin{pmatrix}
					1 & -\eta^2 y^{2\ii\kappa(\xi)} \ee^{-\frac{1}{2}\ii y^2}
					\pmb{A}(\xi)\\
					0 & \mathbb{I}_2
				\end{pmatrix},  
				&y \in \ee^{-\frac{\ii\pi}{4}}\mathbb{R}_+,
			\end{cases}
		\end{equation*}
		where $\eta=(2\sqrt{t})^{-\ii\kappa(\xi)}\ee^{\ii t\xi^2}\ee^{-\beta(\xi,\xi)-\ii\kappa(\xi)}$.
	\end{itemize}
\end{rhp}

\subsection{Solving model problem}
Since it is a standardized procedure to solve a 
$2\times2$ model RHP \cite{deift1993steepest}, we consider reducing RHP \ref{solvable-RHP-model} to 
a $2\times2$ model RHP. We set 
$-\pmb{A}^\dagger(\xi)=[a,b]^{\top}$ and then define
\begin{equation}\label{decrease-dimension-weber-equ-B}
	\mathbf{A}=\begin{pmatrix}
		a^{*} & b^{*}\\
		-b & a
	\end{pmatrix}|\pmb{A}(\xi)|^{-1},\quad
	\mathbf{B}=\begin{pmatrix}
		1 & 0\\
		0 & \mathbf{A}
	\end{pmatrix}.
\end{equation}
Clearly, $\mathbf{A}\mathbf{A}^\dagger=\mathbb{I}_2$ and 
$\mathbf{B}\mathbf{B}^\dagger=\mathbb{I}_3$. Let 
$\M^{(6)}(y):=\M^{(6)}(y;x,t)=\mathbf{B}\M^{(5)}(y)\mathbf{B}^{-1}$ and then 
$\M^{(6)}(y)$ satisfies the following RHP.

\begin{rhp}\label{rhp-M6}
	Find an analytic function $\M^{(6)}$: 
	$\mathbb{C}\setminus\Sigma^{\prime}\rightarrow SL_{2}(\mathbb{C}),\Sigma^{\prime}=\cup_{k=1}^{4}\ee^{\frac{(2k-1)\pi}{4}}\mathbb{R}_+$ 
	with the following properties
	\begin{itemize}
		\item[1.] $\M^{(6)}(y)=\mathbb{I}_{3}+\mathcal{O}\left(y^{-1}\right)$ as $y\rightarrow\infty.$
		\item[2.] For each $y\in\Sigma^\prime$ ,
		$\M^{(6)}$ takes continuous boundary values $\M^{(6)}_{\pm}(y)$ which satisfy 
		the jump relation: $\M_+^{(6)}(y)=\M_-^{(6)}(y)\V^{(6)}(y)$ where
		\begin{equation*}
			\V^{(6)}(y)=
			\begin{cases}
				\begin{pmatrix}
					1 & 0 & 0\\
					\eta^{-2} y^{-2\ii\kappa(\xi)} \ee^{\frac{1}{2}\ii y^2} 
					|\pmb{A}^\dagger(\xi)| & 1 & 0\\
					0 & 0 & 1
				\end{pmatrix},  
				&y \in \ee^{\frac{\ii\pi}{4}}\mathbb{R}_+,\\
				\begin{pmatrix}
					1 & \eta^2 y^{2\ii\kappa(\xi)} 
					\ee^{-\frac{1}{2}\ii y^2} 
					\frac{|\pmb{A}(\xi)|}{1+|\mathbf{R}^\dagger(\xi)|^2} & 0\\
					0 & 1 & 0\\
					0 & 0 & 1
				\end{pmatrix}, 
				&y \in \ee^{-\frac{\ii\pi}{4}}\mathbb{R}_-,\\
				\begin{pmatrix}
					1 & 0 & 0\\
					\eta^{-2} y^{-2\ii\kappa(\xi)} \ee^{\frac{1}{2}
						\ii y^2} \frac{|\pmb{A}^\dagger(\xi)|}{1+|\mathbf{R}^\dagger(\xi)|^2} & 1 & 0\\
					0 & 0 & 1
				\end{pmatrix},
				&y \in \ee^{\frac{\ii\pi}{4}}\mathbb{R}_-,\\
				\begin{pmatrix}
					1 & \eta^2 y^{2\ii\kappa(\xi)} \ee^{-\frac{1}{2}\ii y^2} 
					|\pmb{A}(\xi)| & 0\\
					0 & 1 & 0\\
					0 & 0 & 1
				\end{pmatrix}, 
				&y \in \ee^{-\frac{\ii\pi}{4}}\mathbb{R}_+,
			\end{cases}
		\end{equation*}
		where $\eta=(2\sqrt{t})^{-\ii\kappa(\xi)}\ee^{\ii t\xi^2}\ee^{-\beta(\xi,\xi)-\ii\kappa(\xi)}$.
	\end{itemize}
\end{rhp}

Next, we can transform RHP \ref{rhp-M6} into a $2\times2$ model RHP and 
$\M^{(6)}(y)$ needs to be rewritten in the following block form
\begin{equation*}
	\M^{(6)}(y)=
	\begin{pmatrix}
		\mathbf{N}(y) & \M^{(6)}_{12}(y)\\
		\M^{(6)}_{21}(y) & \M^{(6)}_{22}(y)
	\end{pmatrix},
\end{equation*}
where $\M^{(6)}_{22}(y)$ is scalar.
It follows from RHP \ref{rhp-M6} that 
$\M^{(6)}_{12}(y)=0$, $\M^{(6)}_{21}(y)=0$, 
$\M^{(6)}_{22}(y)=1$ and $\mathbf{N}(y)$ satisfies the following RHP.

\begin{rhp}\label{rhp-N}
	Find an analytic function $N$: 
	$\mathbb{C}\setminus\Sigma^{\prime}\rightarrow SL_{2}(\mathbb{C}),\Sigma^{\prime}=\bigcup_{k=1}^{4}\ee^{\frac{(2k-1)\pi\ii}{4}}\mathbb{R}_+$ 
	with the following properties
	\begin{itemize}
		\item[1.] $\mathbf{N}(y)=\mathbb{I}_{2}+\mathcal{O}\left(y^{-1}\right)$ as $y\rightarrow\infty.$
		\item[2.] For each $y\in\Sigma^\prime$ ,
		$N$ takes continuous boundary values $N_{\pm}(y)$ which satisfy 
		the jump relation: $N_+^{(6)}(y)=N_-^{(6)}(y)\V_N(y)$ where
		\begin{equation*}
			\V_N(y)=
			\begin{cases}
				\begin{pmatrix}
					1 & 0 \\
					\eta^{-2} y^{-2\ii\kappa(\xi)} \ee^{\frac{1}{2}\ii y^2}
					 |\pmb{A}^\dagger(\xi)| & 1 
				\end{pmatrix},  
				&y \in \ee^{\frac{\ii\pi}{4}}\mathbb{R}_+,\\
				\begin{pmatrix}
					1 & \eta^2 y^{2\ii\kappa(\xi)} 
					\ee^{-\frac{1}{2}\ii y^2} \frac{|\pmb{A}(\xi)|}{1+|\mathbf{R}^\dagger(\xi)|^2} \\
					0 & 1 
				\end{pmatrix}, 
				&y \in \ee^{-\frac{\ii\pi}{4}}\mathbb{R}_-,\\
				\begin{pmatrix}
					1 & 0 \\
					\eta^{-2} y^{-2\ii\kappa(\xi)} 
					\ee^{\frac{1}{2}\ii y^2} \frac{|\pmb{A}^\dagger(\xi)|}{1+|\mathbf{R}^\dagger(\xi)|^2} & 1
				\end{pmatrix},
				&y \in \ee^{\frac{\ii\pi}{4}}\mathbb{R}_-,\\
				\begin{pmatrix}
					1 & \eta^2 y^{2\ii\kappa(\xi)} 
					\ee^{-\frac{1}{2}\ii y^2} |\pmb{A}(\xi)| \\
					0 & 1
				\end{pmatrix}, 
				&y \in \ee^{-\frac{\ii\pi}{4}}\mathbb{R}_+,
			\end{cases}
		\end{equation*}
		where $\eta=(2\sqrt{t})^{-\ii\kappa(\xi)}\ee^{\ii t\xi^2}\ee^{-\beta(\xi,\xi)-\ii\kappa(\xi)}$.
	\end{itemize}
\end{rhp}
The solution $\mathbf{N}$ is described in Appendix B and we can obtain the expansion of $\mathbf{N}$ 
\begin{equation}
	\mathbf{N}(y)=\mathbb{I}_2+\frac{\mathbf{N}_1}{y}+\oo(y^{-2}),
\end{equation}
where, using equations \eqref{relation-beta-N}, \eqref{def-beta} and \eqref{relation-b12-b21}, 
we have 
\begin{equation}
	\mathbf{N}_1=\begin{pmatrix}
		0&-\ii\eta^{2}\beta_{12}\\
		\ii\eta^{-2}\beta_{21}&0
	\end{pmatrix},\,
	\beta_{12}=
	\frac{|\pmb{A}(\xi)|}{\sqrt{2\pi}}\ee^{\frac{\pi\kappa_{\xi}}{2}+
		\frac{5\pi \ii}{4}}\kappa_{\xi}\Gamma(\ii\kappa_{\xi}),\,
	\beta_{21}=-\beta_{12}^*.
\end{equation}
Then we can also obtain the expansion of $\M^{(6)}$
\begin{equation}
	\begin{aligned}
		&\M^{(6)}(y)=\mathbb{I}_3+\frac{\M^{(6)}_{1}}{y}+\oo(y^{-2}),
		&\M^{(6)}_{1}=\begin{pmatrix}
			\mathbf{N}_1 & 0\\
			0 & 0
		\end{pmatrix}.
	\end{aligned}
\end{equation}
It follows from $\M^{(6)}(y)=\mathbf{B}\M^{(5)}(y)\mathbf{B}^{-1}$ that 
we can deduce the expansion of $\M^{(5)}$
\begin{equation}\label{M5_1}
	\M^{(5)}(y)=\mathbb{I}_3+\frac{\M^{(5)}_{1}}{y}+\oo(y^{-2}),
	\quad \M^{(5)}_{1}=\mathbf{B}^{-1}\M^{(6)}_1\mathbf{B}.
\end{equation}
Actually, our ultimate goal is to recover the soluiton $\mathbf{q}(x,t)$ for 
CNLS equation \eqref{CNLS} by the recovery formula \eqref{matrix-form-T1}. Therefore, 
we focus on the behavior of $\left(\M^{(5)}_{1}\right)_{12}$
\begin{equation}
	\left(\M^{(5)}_{1}\right)_{12}=[-\ii\eta^2\beta_{12},0]\mathbf{A}=
	\frac{\eta^2e^{\frac{\pi\kappa_{\xi}}{2}-
			\frac{\pi \ii}{4}}\kappa_{\xi}\Gamma
			(\ii\kappa_{\xi})}{\sqrt{2\pi}}\pmb{A}(\xi).
\end{equation}
By the definition of rescaled variable $y=2\sqrt{t}(\lambda-\xi)$, we can obtain the expansion 
of $\M^{(4)}(\lambda)$
\begin{equation}
	\begin{aligned}
		\M^{(4)}(\lambda)=\M^{(5)}(y)
		&=\mathbb{I}_3+\frac{\M^{(5)}_1}{y}+\mathcal{O}(y^{-2}),&y\rightarrow\infty,\\
		&=\mathbb{I}_3+\frac{\M^{(5)}_1}{2\sqrt{t}\lambda}+\mathcal{O}(\lambda^{-2})
		,&\lambda\rightarrow\infty.
	\end{aligned}
\end{equation}
Together with equation \eqref{residue-M4}, we find
\begin{equation}\label{M4_1-M5_1}
	\M_1^{(4)}(x,t)=\frac{\ii}{2\pi}\int_{\Sigma}\pmb{\mu}_4 \pmb{w}_4=\frac{1}{2\sqrt{t}}\M^{(5)}_1(x,t).
\end{equation}
Inverting the above sequence of transformations, we have 
\begin{equation}\label{inverse-rhp}
	\M(\lambda)=\M^{(1)}(\lambda)\pmb{\Delta}(\lambda),\quad
	\M^{(2)}(\lambda)=\M^{(3)}(\lambda)\mathbf{\Xi}(\lambda),\quad
	\M^{(4)}(\lambda)=\mathbf{B}^{-1}
	\begin{pmatrix}
		\mathbf{N}(\varphi(\lambda))&0\\
		0&1
	\end{pmatrix}\mathbf{B},
\end{equation}
and then, by Lemma \ref{bdd-M1-M2}, 
\begin{equation}
	\begin{aligned}
		2(\M_1)_{12}&=2\left(\hat{\M}^{(1)}_1\right)_{12}
		=\left(\frac{\ii}{\pi}\int_{\mathbb{R}_\xi}\hat{\pmb{\mu}}_1 \hat{w}_1\right)_{12}\\
		&=\left(\frac{\ii}{\pi}\int_{\mathbb{R}_\xi}\hat{\pmb{\mu}}_2 \hat{w}_2\right)_{12}+\mathcal{O}(t^{-\frac{3}{4}+\frac{1}{(2p)}})\\
		&=2\left(\M^{(2)}_1\right)_{12}+\mathcal{O}(t^{-\frac{3}{4}+\frac{1}{(2p)}}).
	\end{aligned}
\end{equation}
Since $\mathbf{\Xi}(\lambda)=\mathbb{I}_3$ in region $\Omega_2$, we obtain 
$\M^{(3)}=\M^{(2)}\mathbf{\Xi}=\M^{(2)}=\hat{\M}^{(2)}$ which implies 
\begin{equation}
	\begin{aligned}
		2(\M_1)_{12}
		&=2\left(\M^{(3)}_1\right)_{12}+\mathcal{O}(t^{-\frac{3}{4}+\frac{1}{(2p)}})\\
		&=\left(\frac{\ii}{\pi}\int_{\Sigma}\pmb{\mu}_3 \pmb{w}_3\right)_{12}+\mathcal{O}(t^{-\frac{3}{4}+\frac{1}{(2p)}})\\
		&=\left(\frac{\ii}{\pi}\int_{\Sigma}\pmb{\mu}_4 \pmb{w}_4\right)_{12}
		+\mathcal{O}(t^{-\frac{3}{4}+\frac{1}{(2p)}}),
	\end{aligned}
\end{equation}
by Lemma \ref{bdd-M3-M4}.
Together with equations \eqref{residue-M4}, \eqref{M5_1} and \eqref{M4_1-M5_1}, we obtain
\begin{equation}
	2(\M_1)_{12}
	=\frac{\ee^{\frac{\pi\kappa_{\xi}}{2}-
			\frac{\pi \ii}{4}}\eta^2\kappa_{\xi}
			\Gamma(\ii\kappa_{\xi})\pmb{A}(\xi)}{\sqrt{2\pi t}}+\mathcal{O}(t^{-\frac{3}{4}+\frac{1}{(2p)}}).
\end{equation}
Hence, by recovery formula \eqref{matrix-form-T1}, we have 
\begin{equation}\label{recovery-formula2}
	\mathbf{q}(x,t)
	=\frac{\ee^{\frac{\pi\kappa_{\xi}}{2}-
			\frac{\pi \ii}{4}}\eta^2\kappa_{\xi}\Gamma
			(\ii\kappa_{\xi})\pmb{A}(\xi)}{\sqrt{2\pi t}}
	-2\mathbf{X}_1\mathbf{G}^{-1}(\mathbf{X}_2)^\dagger+\mathcal{O}(t^{-\frac{3}{4}+\frac{1}{(2p)}}).
\end{equation}
But for matrices $\mathbf{X}$ and $\mathbf{G}$, we must obtain $\M(\lambda;x,t)$ evaluated at $\lambda=\lambda_i$.
It follows from equation \eqref{inverse-rhp} that
\begin{equation}
	\begin{aligned}
		&\M(\lambda_i)=\M^{(1)}(\lambda_i)\pmb{\Delta}(\lambda_i),\quad
		\M^{(2)}(\lambda_i)=\M^{(3)}(\lambda_i)\mathbf{\Xi}(\lambda_i),\quad\\
		&\M^{(4)}(\lambda_i)=B^{-1}
		\begin{pmatrix}
			\mathbf{N}(y_i)&0\\
			0&1
		\end{pmatrix}B.
	\end{aligned}
\end{equation}
where $y_i=2\sqrt{t}(\lambda_i-\xi)$. The solution $\mathbf{N}(y)$ can be  given by
\begin{equation}
	\mathbf{N}(y)=\eta^{\pmb{\sigma}_3}\pmb{\Psi}(y)
	y^{-\ii\kappa_\xi\pmb{\sigma}_3}
	\mathrm{\ee}^{\frac14\ii y^2\pmb{\sigma}_3}
	\eta^{-\pmb{\sigma}_3}.
\end{equation}
By equation \eqref{asymptotic-Da} and equation \eqref{weber-equ-Psi}, we obtain 
the long-time asymptotic behavior of $\mathbf{N}(y_i)$
\begin{equation}
	\mathbf{N}(y_i)=\mathbb{I}_2+t^{-1/2}\mathbf{P}(y_i)+\oo(t^{-1}),
\end{equation}
\begin{equation}
	\mathbf{P}(y_i)=
	\begin{cases}
		\begin{pmatrix}
			0&\frac{\eta^2\beta_{12}\ee^{-\frac34\pi \ii}}{2(\lambda_i-\xi)}\\
			\frac{\eta^{-2}\beta_{21} \ee^{\frac34\pi \ii}}{2(\lambda_i-\xi)}&0
		\end{pmatrix},
		&\mathrm{arg}(y_i)\in[0,\frac{\pi}{4}],\\
		\begin{pmatrix}
			0&\frac{\eta^2\beta_{12}\ee^{-\frac34\pi \ii}}{2(\lambda_i-\xi)}\\
			\frac{\eta^{-2}\beta_{21}\ee^{-\frac14\pi \ii}}{2(\lambda_i-\xi)}&0
		\end{pmatrix},
		&\mathrm{arg}(y_i)\in[\frac{\pi}{4},\frac{3\pi}{4}],\\
		\begin{pmatrix}
			0&\frac{\eta^2\beta_{12}\ee^{\frac14\pi \ii}}{2(\lambda_i-\xi)}\\
			\frac{\eta^{-2}\beta_{21}\ee^{-\frac14\pi \ii}}{2(\lambda_i-\xi)}&0
		\end{pmatrix},
		&\mathrm{arg}(y_i)\in[\frac{3\pi}{4},\pi].\\
	\end{cases}
\end{equation}
Then we have 
\begin{equation*}
	\M^{(4)}(\lambda_i)=\mathbb{I}_3+
	t^{-1/2}\tilde{\mathbf{P}}_i+\oo(t^{-1}),
\end{equation*}
\begin{equation}\label{def-tilde-P}
	\tilde{\mathbf{P}}_i=
	\begin{pmatrix}
		0&\frac{-\mathbf{P}_{12}(\varphi(\lambda_i))
		\pmb{A}(\xi)}{\left|\pmb{A}(\xi)\right|}\\
		\frac{-\mathbf{P}_{21}(\varphi(\lambda_i))\pmb{A}^\dagger(\xi)}
		{\left|\pmb{A}(\xi)\right|}&0
	\end{pmatrix}.
\end{equation}
It follows from equation \eqref{analysis-expand-Phi} that
\begin{equation}
	\mathbf{\Xi}(\lambda_i)=\mathbb{I}_3+\oo(\ee^{-ct}).
\end{equation}
By \cite[Lemma 5.18]{deift2011long}, we can deduce 
\begin{equation}
	\begin{aligned}
		\left|\M^{(3)}(\lambda_i;x,t)-\M^{(4)}(\lambda_i;x,t)\right|&
		\lesssim \oo(t^{-\frac{3}{4}}),\\
		\left|\M^{(1)}(\lambda_i;x,t)-\M^{(2)}(\lambda_i;x,t)\right|&
		\lesssim \oo(t^{-\frac{3}{4}+\frac{1}{2p}})
	\end{aligned}
\end{equation}
from Lemma \ref{bdd-M1-M2} and Lemma \ref{bdd-M3-M4}.
Consequently, we obtain the long-time asymptotic behavior of 
$\M(\lambda_i;x,t)$, 
\begin{equation}\label{estimate-M-lambda-i}
	\begin{aligned}
		\M(\lambda_i;x,t)=\left(\mathbb{I}_3+
		t^{-1/2}\tilde{\mathbf{P}}_i\right)\pmb{\Delta}(\lambda_i)
		+\oo(t^{-\frac{3}{4}+\frac{1}{(2p)}}).
	\end{aligned}
\end{equation}
In conclusion, Theorem \ref{main-result} follows from equations 
\eqref{recovery-formula2} and \eqref{estimate-M-lambda-i}.

\section*{Funding}
Liming Ling is supported by the National Natural Science Foundation of China (Grant No. 12471236), Guangzhou Science and Technology Plan (Grant No. 2024A04J6245), and
Guangdong Natural Science Foundation (Grant No. 2025A1515011868); Xiaoen Zhang is supported by the National Natural Science Foundation of China (Grant No. 12471237), the Science and Technology Support Plan for Youth Innovation of Colleges and Universities of Shandong Province of China (No. 2023KJ090).

\section*{Conflict of interest}
The authors declare that they have no conflict of interest with other people or organizations that may inappropriately influence the author's work.

\section*{Competing interests}
The authors declare that they have no known competing financial interests or personal relationships that could have appeared to influence the work reported in this paper.
\section*{Data availability statement}
We do not analyse or generate any datasets, because our work proceeds within a theoretical and mathematical approach. One can obtain the relevant materials from the references below.

\section*{Appendix A}\label{trivial-extend}
\setcounter{equation}{0} 
\renewcommand{\theequation}{A.\arabic{equation}} 
The scalar version of the following lemmas has already been proposed in previous literature \cite{deift2002long,deift2011long}. 
We extend these lemmas trivially to the vector case and provide a more detailed proof for the sake of completeness.

\begin{proof}[Proof of Lemma \ref{bdd-w}]
	Our proof follows that found in \cite[Lemma 5.2]{deift2011long}.
	Define $\mathbf{w}=(\mathbf{w}_-,\mathbf{w}_+)=(0,\V-\mathbb{I}_3)$ and $\hat{\mathbf{w}}=(\hat{\mathbf{w}}_-,\hat{\mathbf{w}}_+)=(0,\V^{-1}-\mathbb{I}_3)$, then we will show 
	that $1-C_{\mathbf{w}}$ is a Fredholm operator with index 0. For $\mathbf{f}\in L^2(\mathbb{R})$, 
	\begin{equation*}
		\begin{aligned}
			C_{\mathbf{w}}(C_{\hat{\mathbf{w}}}\mathbf{f}) 
			&=C^{-}_{\mathbb{R}}[C^{-}_{\mathbb{R}}(\mathbf{f}(\V^{-1}-\mathbb{I}_3))(\V-\mathbb{I}_3)] \\
			&=C^{-}_{\mathbb{R}}[(C^{+}_{\mathbb{R}}(\mathbf{f}(\V^{-1}-\mathbb{I}_3))-\mathbf{f}(\V^{-1}-\mathbb{I}_3))(\V-\mathbb{I}_3)] \\
			&=C^{-}_{\mathbb{R}}[C^{+}_{\mathbb{R}}(\mathbf{f}(\V^{-1}-\mathbb{I}_3))(\V-\mathbb{I}_3)+
			\mathbf{f}(\V-\mathbb{I}_3+\V^{-1}-\mathbb{I}_3)] \\
			&=P\mathbf{f}+C_{\mathbf{w}}\mathbf{f}+C_{\hat{\mathbf{w}}}\mathbf{f}.
		\end{aligned}
	\end{equation*}
	where $P\mathbf{f}=C^{-}_{\mathbb{R}}[C^{+}_{\mathbb{R}}(\mathbf{f}(\V^{-1}-\mathbb{I}_3))(\V-\mathbb{I}_3)]$. 
	As $\V$ is continuous and $\V-\mathbb{I}_3\rightarrow0$ as 
	$\lambda\rightarrow\infty$, we can approximate $\V-\mathbb{I}_3$ by a sequence 
	of rational functions of the following form 
	\begin{equation*}
		\sum_{i=1}^n\frac{{\bf A}_i}{z-a_i},\quad a_i\in\mathbb{C}\setminus\mathbb{R}.
	\end{equation*}
	By the Cauchy theorem, 
	\begin{equation*}
		\begin{aligned}
			P_i \mathbf{f}(z)
			&=C^{-}_{\mathbb{R}}\left(C^{+}_{\mathbb{R}}(\mathbf{f}(\V^{-1}-\mathbb{I}_3))(s)\frac{1}{s-a_i}\right)(z)\\
			&=\lim_{\varepsilon\downarrow0}\int_{\mathbb{R}}
			\frac{C^{+}_{\mathbb{R}}(\mathbf{f}(\V^{-1}-\mathbb{I}_3))(s)}{(s-a_i)(s-(z-\ii\varepsilon))}
			\frac{\dd s}{2\pi \ii}\\
			&=
			\begin{cases}
				0,& a_i\in\mathbb{C}^-,\\
				\frac{C^{+}(\mathbf{f}(\V^{-1}-\mathbb{I}_3))(a_i)}{a_i-z},&a_i\in\mathbb{C}^+.
			\end{cases}
		\end{aligned}
	\end{equation*}
	This implies operator $P_i$ has either rank 0 or 1 and hence is compact in 
	$L^2(\mathbb{R})$. Together with 
	$P=\displaystyle\lim_{n\rightarrow\infty}\sum_{i=1}^n{\bf A}_iP_i$, $P$ is operator limit of 
	compact operators and hence is compact. It follows from 
	$(1-C_\mathbf{w})(1-C_{\hat{\mathbf{w}}})=1+P$ that $(1-C_\mathbf{w})(1-C_{\hat{\mathbf{w}}})-1$ is compact. 
	Similarly, we can obain that $(1-C_{\hat{\mathbf{w}}})(1-C_\mathbf{w})-1$ is also compact. 
	Therefore, $1-C_\mathbf{w}$ is a Fredholm operator. For $0\le\alpha\le1$, set 
	\begin{equation*}
		\V_\alpha=\begin{pmatrix}
			1+\alpha^2|\mathbf{R}|^2&-\alpha\mathbf{R}\ee^{-2\ii t\theta}\\
			-\alpha\mathbf{R}^\dagger \ee^{2\ii t\theta}&\mathbb{I}_2
		\end{pmatrix}.
	\end{equation*}
	Clearly $\V_0=\mathbb{I}_3$ and $\V_1=\V$. Then we can show that 
	$1-C_{\mathbf{w}_\alpha}$ is a Fredholm operator as same as $1-C_\mathbf{w}$ 
	and conclude that ${\rm ind}(1-C_{\mathbf{w}_1})=ind(1-C_{\mathbf{w}_0})=0$. Hence 
	$1-C_{\mathbf{w}}$ is a Fredholm operator with index 0. Then we want to 
	show that $(1-C_{\mathbf{w}})\mathbf{f}=0$, $\mathbf{f}\in L^2(\mathbb{R})$ implies that $\mathbf{f}=0$. 
	Set ${\bf m}_\pm=C^\pm_{\mathbb{R}}(\mathbf{f}(\V-\mathbb{I}_3))$. Since $(1-C_{\mathbf{w}})\mathbf{f}=0$, we have 
	\begin{equation*}
		\mathbf{f}=C_{\mathbf{w}}\mathbf{f}=C^-_{\mathbb{R}}(\mathbf{f}(\V-\mathbb{I}_3))={\bf m}_-,
	\end{equation*}
	and 
	\begin{equation*}
		\begin{aligned}
			{\bf m}_+&
			=C^+_{\mathbb{R}}(\mathbf{f}(\V-\mathbb{I}_3))=C^-_{\mathbb{R}}
			(\mathbf{f}(\V-\mathbb{I}_3))+\mathbf{f}(\V-\mathbb{I}_3)\\
			&={\bf m}_-+{\bf m}_-(\V-\mathbb{I}_3)={\bf m}_-\V.
		\end{aligned}
	\end{equation*}
	By the Cauchy theorem,
	\begin{equation*}
		0=\int_\mathbb{R}{\bf m}_+{\bf m}_-^\dagger =\int_\mathbb{R}{\bf m}_-\V {\bf m}_-^\dagger.
	\end{equation*}
	But $\V$ is clearly strictly positive and hence $\mathbf{f}={\bf m}_-=0$. Thus 
	$1-C_\mathbf{w}$ is injective. It follow form 
	$1-C_{\mathbf{w}}$ is a Fredholm operator with index 0 that 
	$1-C_{\mathbf{w}}$ is bijective and $(1-C_{\mathbf{w}})^{-1}$ exists. 
	By Corollary 2.7 in \cite{deift2002long}, we conclude that 
	$1-C_{\mathbf{w}}$ is invertible in $L^2(\mathbb{R})$ for any factorization. 
	To complete the proof, we must show that $(1-C_{\mathbf{w}})^{-1}$ is bounded 
	in $L^2(\mathbb{R})$.
	
	By proposition 2.6 in \cite{deift2002long}, 
	we suppose that $\M_\pm\in\partial C(L^2)$ solve $IRHP2_{L^2}$ defined in \cite[pp.1035]{deift2002long} 
	which implies $\M_+=\M_-\V+\mathbf{f}$ for $\mathbf{f}\in L^2(\mathbb{R})$. As above, we have 
	\begin{equation*}
		0=\int_\mathbb{R}\M_+\M_-^\dagger =\int_\mathbb{R}\M_-\V\M_-^\dagger
		+\int_\mathbb{R}\mathbf{f}\M_-^\dagger.
	\end{equation*}
	Since $\V$ is clearly strictly positive, set $\lambda_0>0$ be the smallest eigenvalue of $\V$. 
	Then we have 
	\begin{equation*}
		\lambda_0\|\M_-\|^2_{L^2}\le
		\left|\int_\mathbb{R}\M_-\V\M_-^\dagger\right|
		=\left|\int_\mathbb{R}\mathbf{f}\M_-^\dagger\right|\le
		\|\mathbf{f}\|_{L^2}\|\M_-\|_{L^2}.
	\end{equation*}
	Hence $\|\M_\pm\|_{L^2}\lesssim\|\mathbf{f}\|_{L^2}$, which implies 
	$\|(1-C_\mathbf{w})^{-1}\|_{L^2(\mathbb{R})}\lesssim 1$ for any factorization of $\V$ by 
	proposition 2.6 and corollary 2.7 in \cite{deift2002long}. 
	This completes the proof.
\end{proof}

\begin{proof}[Proof of Lemma \ref{lemma1-CNLS}]
	Our proof follows that found in \cite[Lemma 4.3]{deift2002long}.
	Consider the first inequality. The other cases are similar. For convenience, 
	we assume that $x=0$. We only need to show
	\begin{equation*}
		\|C^-_{\mathbb{R}_+\rightarrow\Gamma}
		(\det\pmb{\delta})^{-1}\ee^{2\ii t\lambda^2}\pmb{\delta}^{-1}\mathbf{f}^\dagger\|_{L^p},
	\end{equation*}
	where $\Gamma=\mathbb{R}\cup(\ee^{\ii\pi/2}\mathbb{R}_-)\cup(\ee^{-\ii\pi/2}\mathbb{R}_-)$. 
	Then the result follows immediately by the translation of 
	$\lambda\rightarrow\lambda-\xi$ since the $H^{1,1}$ and $L^\infty$ norm of 
	$\mathbf{f}$ is invariant under the translation. By Fourier theory, 
	$\mathbf{f}^\dagger(\lambda)=1/\sqrt{2\pi}\int \ee^{-iy\lambda}\check{\mathbf{f}}^\dagger(y)\dd y$. For any 
	$\varepsilon>0$, 
	\begin{equation*}
		\begin{aligned}
			C^-_{\mathbb{R}_+\rightarrow\Gamma}\ee^{-\varepsilon\lambda}
			(\det\pmb{\delta})^{-1}\ee^{2\ii t\lambda^2}\pmb{\delta}^{-1}\mathbf{f}^\dagger
			&=\frac1{\sqrt{2\pi}}\int \ee^{-\ii\frac{y^2}{8t}}F_1\check{\mathbf{f}}^\dagger(y)\dd y
			+\frac1{\sqrt{2\pi}}\int \ee^{-\ii\frac{y^2}{8t}}F_2\check{\mathbf{f}}^\dagger(y)\dd y\\
			&\equiv \mathrm{I}+\mathrm{II},
		\end{aligned}
	\end{equation*}
	where 
	\begin{equation*}
		\begin{aligned}
			&F_1(y)=C^-_{\mathbb{R}_+\rightarrow\Gamma}
			(\ee^{-\varepsilon\lambda}
			(\det\pmb{\delta})^{-1}\ee^{2\ii t(\lambda-\frac{y}{4t})^2}
			\chi_{(0,a)}\pmb{\delta}^{-1}),\\
			&F_2(y)=C^-_{\mathbb{R}_+\rightarrow\Gamma}
			(\ee^{-\varepsilon\lambda}
			(\det\pmb{\delta})^{-1}\ee^{2\ii t(\lambda-\frac{y}{4t})^2}
			\chi_{(a,\infty)}\pmb{\delta}^{-1}),
		\end{aligned}
	\end{equation*}
	and $a=\max(0,\frac{y}{4t})$. The factor $\ee^{-\varepsilon\lambda}$ is included to 
	ensure that $\|F_2\|_{L^p}\lesssim\|\ee^{-\varepsilon\lambda}\|_{L^p(a,\infty)}$ and 
	hence $F_2(y)$ exists in $L^p$.
	
	Assume first that $p>2$. For $y>0$, then $a=\frac{y}{4t}$, 
	\begin{equation*}
		\|F_1\|_{L^p(\Gamma)}\lesssim \|(\det\pmb{\delta})^{-1}\|_{L^\infty(\mathbb{R}_+)}
		\|\pmb{\delta}^{-1}\|_{L^\infty(\mathbb{R}_+)}
		\|\chi_{(0,a)}\|_{L^p}
		\lesssim \frac{|y|^{1/p}}{t^{1/p}},
	\end{equation*}
	and hence 
	\begin{equation*}
		\begin{aligned}
			\|\mathrm{I}\|_{L^p(\Gamma)}
			&\lesssim t^{-1/p}
			\int_{0}^{\infty}y^{1/p}|\check{\mathbf{f}}^\dagger(y)|\dd y\\
			&\lesssim t^{-1/p}
			(\int_{0}^{1}y^{1/p}|\check{\mathbf{f}}^\dagger(y)|\dd y+
			\int_{1}^{\infty}y|\check{\mathbf{f}}^\dagger(y)|y^{1/p-1}\dd y)\\
			&\lesssim t^{-1/p}
			(\|\check{\mathbf{f}}^\dagger\|_{L^2}\|y^{1/p}\|_{L^2(0,1)}+
			\|y\check{\mathbf{f}}^\dagger\|_{L^2}\|y^{1/p-1}\|_{L^2(1,\infty)})\\
			&\lesssim t^{-1/p}\|\mathbf{f}\|_{H^{1,1}}.
		\end{aligned}
	\end{equation*}
	For $p=2$, the part $\mathrm{I}$ can be rewritten as 
	\begin{equation*}
		\frac1{\sqrt{2\pi}}C^-_{\mathbb{R}_+\rightarrow\Gamma}
		(\int_{4t\lambda}^{\infty} \ee^{-\varepsilon\lambda}
		(\det\pmb{\delta})^{-1}\ee^{2\ii t\lambda^2-iy\lambda}
		\pmb{\delta}^{-1}\check{\mathbf{f}}^\dagger\dd y).
	\end{equation*}
	By Hardy inequality, we get 
	\begin{equation*}
		\begin{aligned}
			\|\int_{4t\lambda}^{\infty}| \check{\mathbf{f}}^\dagger|\dd y\|_{L^2(\mathbb{R}_+)}
			&=\|\int_{\lambda}^{\infty}\frac{| 4ty\check{\mathbf{f}}^\dagger(4ty)|}{y}\dd y\|_{L^2(\mathbb{R}_+)}\\
			&\le 2\|| \check{\mathbf{f}}^\dagger(4ty)|4ty\|_{L^2(\mathbb{R}_+)}\\
			&\le 2\frac{1}{\sqrt{4t}}\|| y\check{\mathbf{f}}^\dagger|\|_{L^2(\mathbb{R}_+)}\lesssim t^{-1/2},
		\end{aligned}
	\end{equation*}
	and hence 
	\begin{equation*}
		\|\mathrm{I}\|_{L^2(\Gamma)}\lesssim
		\|\int_{4t\lambda}^{\infty} \ee^{-\varepsilon\lambda}
		(\det\pmb{\delta})^{-1}\ee^{2\ii t\lambda^2-iy\lambda}
		\pmb{\delta}^{-1}\check{\mathbf{f}}^\dagger\dd y\|_{L^2(\mathbb{R}_+)}\lesssim
		\|\int_{4t\lambda}^{\infty} | \check{\mathbf{f}}^\dagger|\dd y\|_{L^2(\mathbb{R}_+)}
		\lesssim t^{-1/2}.
	\end{equation*}
	For $F_2$, first consider the case when $y<0$, and hence $a=0$. Then for $p\ge2$, by 
	Cauchy theorem,
	\begin{equation*}
		\begin{aligned}
			\|F_2\|_{L^p(\Gamma)}
			&=\|C_{\ee^{\ii\pi/4}\mathbb{R}_+\rightarrow\Gamma}
			(\ee^{-\varepsilon\lambda}
			(\det\pmb{\delta})^{-1}\ee^{2\ii t(\lambda-\frac{y}{4t})^2}
			\pmb{\delta}^{-1})\|_{L^p}\\
			&\lesssim\|\ee^{-\varepsilon\lambda}
			(\det\pmb{\delta})^{-1}\ee^{2\ii t(\lambda-\frac{y}{4t})^2}
			\pmb{\delta}^{-1}\|_{L^p(\ee^{\ii\pi/4}\mathbb{R}_+)}\\
			&\lesssim \|\ee^{2\ii t\lambda^2}
			\ee^{-\ii y\lambda}\|_{L^p(\ee^{\ii\pi/4}\mathbb{R}_+)}.
		\end{aligned}
	\end{equation*}
	Since $y<0$ and $\lambda\in \ee^{\ii\pi/4}\mathbb{R}_+$, it follows from 
	${\rm Re}(-\ii y\lambda)<0$ that 
	\begin{equation*}
		\|F_2\|_{L^p(\Gamma)}\lesssim\|\ee^{2\ii t\lambda^2}
		\|_{L^p(\ee^{\ii\pi/4}\mathbb{R}_+)}\lesssim t^{-1/(2p)}.
	\end{equation*}
	When $y>0$, then $a=\frac{y}{4t}$ and for $p\ge2$, again by Cauchy theorem,
	\begin{equation*}
		\begin{aligned}
			\|F_2\|_{L^p(\Gamma)}
			&=\|C^-_{(y/4t,\infty)\rightarrow\Gamma}
			(\ee^{-\varepsilon\lambda}
			(\det\pmb{\delta})^{-1}\ee^{2\ii t(\lambda-\frac{y}{4t})^2}
			\pmb{\delta}^{-1})\|_{L^p}\\
			&=\|C_{(y/4t+\ee^{\ii\pi/4}\mathbb{R}_+)\rightarrow\Gamma}
			(\ee^{-\varepsilon\lambda}
			(\det\pmb{\delta})^{-1}\ee^{2\ii t(\lambda-\frac{y}{4t})^2}
			\pmb{\delta}^{-1})\|_{L^p}\\
			&\lesssim\|\ee^{2\ii t(\lambda-\frac{y}{4t})^2}
			\|_{L^p(y/4t+\ee^{\ii\pi/4}\mathbb{R}_+)}\\
			&\lesssim t^{-1/(2p)}.
		\end{aligned}
	\end{equation*}
	Thus for $p\ge2$, we have $\|F_2\|_{L^p(\Gamma)}\lesssim t^{-1/(2p)}$. Then
	\begin{equation*}
		\begin{aligned}
			\|\mathrm{II}\|_{L^p(\Gamma)}
			&\lesssim t^{-1/(2p)}\int_{-\infty}^{\infty}|\check{\mathbf{f}}^\dagger(y)|\dd y\\
			&\lesssim t^{-1/(2p)}(\int_{|y|\le1}|\check{\mathbf{f}}^\dagger(y)|+\int_{|y|>1}|y|^{-1}|y\check{\mathbf{f}}^\dagger(y)|)\dd y\\
			&\lesssim t^{-1/(2p)}\|\mathbf{f}\|_{H^{1,1}}.
		\end{aligned}
	\end{equation*}
	Letting $\varepsilon\downarrow0$ and adding the estimates for $\mathrm{I}$ and $\mathrm{II}$, 
	the result follows immediately.
\end{proof}

\begin{proof}[Proof of Lemma \ref{lemma2-CNLS}]
	Our proof follows that found in \cite[Lemma 4.3 ]{deift2002long}.
	We only consider the first inequality, the other cases are similar. Again, 
	we assume that $x=0$. We only need to show
	\begin{equation*}
		\|C^-_{\mathbb{R}_+\rightarrow\Gamma}
		\ee^{2\ii t\lambda}g_+\mathbf{f}^\dagger\|_{L^2}
		\lesssim t^{-1/2+1/q}\|g\|_{H^q(\mathbb{C}\setminus\mathbb{R})}.
	\end{equation*}
	It follows from $\mathbf{f}(0)=0$ and Fourier theorem that $\int\check{\mathbf{f}}^\dagger\dd y=0$ and 
	\begin{equation*}
		\begin{aligned}
			C^-_{\mathbb{R}_+\rightarrow\Gamma}
			\ee^{-\varepsilon\lambda}\ee^{2\ii t\lambda^2}g_+\mathbf{f}^\dagger
			&=\frac1{\sqrt{2\pi}}C^-_{\mathbb{R}_+\rightarrow\Gamma}
			\ee^{-\varepsilon\lambda}\ee^{2\ii t\lambda^2}g_+\int (\ee^{-iy\lambda}-1)\check{\mathbf{f}}^\dagger\dd y\\
			&=\frac1{\sqrt{2\pi}}\int \ee^{-\ii\frac{y^2}{8t}}F_1\check{\mathbf{f}}^\dagger(y)\dd y
			+\frac1{\sqrt{2\pi}}\int \ee^{-\ii\frac{y^2}{8t}}F_2\check{\mathbf{f}}^\dagger(y)\dd y\\
			&\equiv \mathrm{I}+\mathrm{II},
		\end{aligned}
	\end{equation*}
	where 
	\begin{equation*}
		\begin{aligned}
			&F_1(y)=C^-_{\mathbb{R}_+\rightarrow\Gamma}
			(\ee^{-\varepsilon\lambda}
			\ee^{2\ii t(\lambda-\frac{y}{4t})^2}
			\chi_{(0,a)}(1-\ee^{iy\lambda})g_+),\\
			&F_2(y)=C^-_{\mathbb{R}_+\rightarrow\Gamma}
			(\ee^{-\varepsilon\lambda}
			\ee^{2\ii t(\lambda-\frac{y}{4t})^2}
			\chi_{(a,\infty)}(1-\ee^{iy\lambda})g_+),
		\end{aligned}
	\end{equation*}
	and $a=\max(0,\frac{y}{4t})$. For $y>0$, $\mathrm{I}$ can be rewritten as 
	\begin{equation*}
		\frac1{\sqrt{2\pi}}C^-_{\mathbb{R}_+\rightarrow\Gamma}
		(g_+\int_{4t\lambda}^{\infty} \ee^{-\varepsilon\lambda}\ee^{2\ii t\lambda^2}
		(\ee^{-iy\lambda}-1)\check{\mathbf{f}}^\dagger(y)\dd y).
	\end{equation*}
	Then for $q^{'}\ge2$ which satisfies $\frac{1}{q}+\frac{1}{q^{'}}=\frac{1}{2}$, we have
	\begin{equation*}
		\begin{aligned}
			\|\mathrm{I}\|_{L^2(\Gamma)}
			&\lesssim \|g_+\int_{4t\lambda}^{\infty} \check{\mathbf{f}}^\dagger(y)\dd y\|_{L^2(\mathbb{R}_+)}\\
			&\lesssim \|g_+\|_{L^q(\mathbb{R}_+)}
			\|\int_{4t\lambda}^{\infty} \check{\mathbf{f}}^\dagger(y)\dd y\|_{L^{q^{'}}(\mathbb{R}_+)}\\
			&\lesssim \|g\|_{H^q(\mathbb{C}\setminus\mathbb{R})}
			\left(\int_{0}^{\infty}\left( \int_0^\infty|\check{\mathbf{f}}^\dagger(y)|\dd y\right)^{q^\prime-2}
			\left(\int_{4t\lambda}^\infty|\check{\mathbf{f}}^\dagger(y)|\dd y\right)^2d\lambda\right)^{1/q^{\prime}}\\
			&\lesssim \|g\|_{H^q(\mathbb{C}\setminus\mathbb{R})}
			\left\|\int_{4t\lambda}^\infty|\check{\mathbf{f}}^\dagger(y)|\dd y\right\|_{L^2(\mathbb{R}_+)}^{2/q^{\prime}}.
		\end{aligned}
	\end{equation*}
	Together with Hardy inequality
	\begin{equation*}
		\left\|\int_{4t\lambda}^\infty|\check{\mathbf{f}}^\dagger(y)|\dd y\right\|_{L^2(\mathbb{R}_+)}\lesssim
		t^{-1/2},
	\end{equation*}
	then we get
	\begin{equation*}
		\|\mathrm{I}\|_{L^2(\Gamma)}\lesssim\|g\|_{H^q(\mathbb{C}\setminus\mathbb{R})}t^{-1/q^{'}}
		\lesssim \|g\|_{H^q(\mathbb{C}\setminus\mathbb{R})}t^{-1/2+1/q}.
	\end{equation*}
	For $F_2$, first consider the case when $y<0$, then $a=0$ and 
	\begin{equation*}
		\begin{aligned}
			\|F_2\|_{L^2(\Gamma)}
			&=\|C_{(\ee^{\ii \pi/4}\mathbb{R}_+)\rightarrow\Gamma}
			(\ee^{-\varepsilon\lambda}
			\ee^{2\ii t(\lambda-\frac{y}{4t})^2}
			(1-\ee^{iy\lambda})g_+)\|_{L^2(\Gamma)}\\
			&\lesssim\|\ee^{2\ii t\lambda^2}
			(\ee^{-iy\lambda}-1)g_+\|_{L^2(\ee^{\ii \pi/4}\mathbb{R}_+)}.
		\end{aligned}
	\end{equation*}
	Let $\lambda=ue^{\ii \pi/4},u\ge0$, then
	\begin{equation*}
		\begin{aligned}
			\|\mathrm{II}\|_{L^2(\Gamma)}
			&=\|\frac1{\sqrt{2\pi}}\int_{-\infty}^{0} 
			\ee^{-\ii\frac{y^2}{8t}}F_2\check{\mathbf{f}}^\dagger(y)\dd y\|_{L^2(\Gamma)}\\
			&\lesssim\int_{-\infty}^{0} 
			\left(\int_0^\infty du\left|g(\ee^{\ii\frac\pi4}u)\right|^2
			\ee^{-4tu^2}\left|\ee^{-ie^{\ii\frac\pi4}u y}-1\right|^2\right)
			^{\frac12}|\check{\mathbf{f}}^\dagger(y)|\dd y.
		\end{aligned}
	\end{equation*}
	Since ${\rm Re}(-ie^{\ii\frac\pi4}u y)<0$, then
	\begin{equation*}
		\left|\ee^{-ie^{\ii\frac\pi4}u y}-1\right|
		=\left|\int_{0}^{-ie^{\ii\frac\pi4}u y}\ee^zdz\right|\le|uy|.
	\end{equation*}
	Hence 
	\begin{equation*}
		\begin{aligned}
			\|\mathrm{II}\|_{L^2(\Gamma)}
			&\lesssim\left(\int_0^\infty du\left|g(\ee^{\ii\frac\pi4}u)\right|^2
			\ee^{-4tu^2}\right)
			^{\frac12}\int_{-\infty}^{-t^{\frac{1}{2}}} 
			|\check{\mathbf{f}}^\dagger(y)|\dd y\\
			&+\left(\int_0^\infty du\left|g(\ee^{\ii\frac\pi4}u)\right|^2
			\ee^{-4tu^2}u^2\right)
			^{\frac12}\int_{-t^{\frac{1}{2}}}^{0} 
			|y\check{\mathbf{f}}^\dagger(y)|\dd y\\
			&\lesssim \left(\int_0^\infty dx\left|g(\ee^{\ii\frac\pi4}\frac{x}{\sqrt{t}})
			\right|^2e^{-4x^2}\right)
			^{\frac12}t^{-\frac{1}{4}}t^{-\frac{1}{4}}
			+\left(\int_0^\infty dx\left|g(\ee^{\ii\frac\pi4}\frac{x}{\sqrt{t}})
			\right|^2e^{-4x^2}x^2\right)
			^{\frac12}t^{-\frac{3}{4}}t^{\frac{1}{4}}\\
			&\lesssim t^{-\frac{1}{2}}
			\|g(\frac{x}{\sqrt{t}})\|_{L^q(\ee^{\ii\frac\pi4}\mathbb{R}_+)}
			\lesssim t^{-\frac{1}{2}+\frac{1}{2q}}\|g\|_{L^q(\ee^{\ii\frac\pi4}\mathbb{R}_+)}
			\lesssim t^{-\frac{1}{2}+\frac{1}{2q}}
			\|g\|_{H^q(\mathbb{C}\setminus\mathbb{R})}.
		\end{aligned}
	\end{equation*}
	When $y>0$, then $a=y/4t$ and 
	\begin{equation*}
		\begin{aligned}
			\|F_2\|_{L^2(\Gamma)}
			&=\|C^-_{(y/4t,\infty)\rightarrow\Gamma}
			(\ee^{-\varepsilon\lambda}
			\ee^{2\ii t(\lambda-\frac{y}{4t})^2}
			(1-\ee^{iy\lambda})g_+)\|_{L^2}\\
			&\lesssim\|C_{(y/4t+\ee^{\ii \pi/4}\mathbb{R}_+)\rightarrow\Gamma}
			(\ee^{2\ii t(\lambda-\frac{y}{4t})^2}
			(1-\ee^{iy\lambda})g_+)\|_{L^2}\\
			&\lesssim\|(\ee^{2\ii t(\lambda-\frac{y}{4t})^2}
			(1-\ee^{iy\lambda})g_+)\|_{L^2(y/4t+\ee^{\ii \pi/4}\mathbb{R}_+)}.
		\end{aligned}
	\end{equation*}
	Let $\lambda=\frac{y}{4t}+ue^{\ii \pi/4},u\ge0$, then
	\begin{equation*}
		\begin{aligned}
			\|\mathrm{II}\|_{L^2(\Gamma)}
			&\lesssim\int_{0}^{\infty} 
			\left(\int_0^\infty du\left|g(\frac{y}{4t}+ue^{\ii \pi/4})\right|^2
			\ee^{-4tu^2}\left|1-\ee^{iy(\frac{y}{4t}+ue^{\ii \pi/4})}\right|^2\right)
			^{\frac12}|\check{\mathbf{f}}^\dagger(y)|\dd y\\
			&\lesssim\int_{t^{\frac12}}^{\infty} 
			\left(\int_{0}^\infty du\left|g(\frac{y}{4t}+ue^{\ii \pi/4})\right|^2
			\ee^{-4tu^2}\right)^{\frac12}|\check{\mathbf{f}}^\dagger(y)|\dd y\\
			&+\int_{0}^{t^{\frac12}} 
			\left(\int_{0}^\infty du\left|g(\frac{y}{4t}+ue^{\ii \pi/4})\right|^2
			\ee^{-4tu^2}
			(|\ee^{-\ii\frac{y^2}{4t}}-1|^2+|1-\ee^{iyue^{\ii \pi/4}}|)
			\right)
			^{\frac12}|\check{\mathbf{f}}^\dagger(y)|\dd y.
		\end{aligned}
	\end{equation*}
	Together with $|\ee^{-\ii\frac{y^2}{4t}}-1|\lesssim t^{-1}y^2$,
	\begin{equation*}
		\begin{aligned}
			\|\mathrm{II}\|_{L^2(\Gamma)}
			&\lesssim t^{-\frac14}
			\left(\int_{0}^\infty dx\left|g(\frac{y}{4t}+\frac{x}{\sqrt{t}}\ee^{\ii \pi/4})\right|^2
			\ee^{-4tx^2}\right)^{\frac12}
			\int_{t^{\frac12}}^{\infty}y^{-1} |y\check{\mathbf{f}}^\dagger(y)|\dd y\\
			&+t^{-1}t^{-\frac14}
			\left(\int_{0}^\infty dx\left|g(\frac{y}{4t}+\frac{x}{\sqrt{t}}\ee^{\ii \pi/4})\right|^2
			\ee^{-4tx^2}\right)^{\frac12}
			\int_{0}^{t^{\frac12}} y |y\check{\mathbf{f}}^\dagger(y)|\dd y\\
			&+t^{-\frac34}
			\left(\int_{0}^\infty dx\left|g(\frac{y}{4t}+\frac{x}{\sqrt{t}}\ee^{\ii \pi/4})\right|^2
			\ee^{-4tx^2}x^2\right)^{\frac12}
			\int_{0}^{t^{\frac12}} |y\check{\mathbf{f}}^\dagger(y)|\dd y\\
			&\lesssim t^{-\frac12}
			\|g(\frac{y}{4t}+\frac{x}{\sqrt{t}}\ee^{\ii \pi/4})\|_{L^q(\mathbb{R}_+)}
			\lesssim t^{-\frac12+\frac{1}{2q}}
			\|g\|_{L^q(\frac{y}{4t}+\ee^{\frac{\ii\pi}4}\mathbb{R}_+)}
			\lesssim t^{-1/2+1/2q}\|g\|_{H^q(\mathbb{C}\setminus\mathbb{R})}.
		\end{aligned}
	\end{equation*}
	Letting $\varepsilon\downarrow0$ and adding the estimates for $\mathrm{I}$ and $\mathrm{II}$, 
	the result follows immediately.
\end{proof}

\begin{proof}[Proof of Lemma \ref{bdd-M1-M2}]
	Our proof follows that found in \cite[Lemma 4.8]{deift2002long}.
	\begin{equation}
		\begin{aligned}
			\int\hat{\pmb{\mu}}_1 \hat{\mathbf{w}}_1-\int\hat{\pmb{\mu}}_2 \hat{\mathbf{w}}_2
			& =\int \hat{\mathbf{w}}_1-\hat{\mathbf{w}}_2+
			\int(\hat{\pmb{\mu}}_2-\mathbb{I}_{3})(\hat{\mathbf{w}}_1-\hat{\mathbf{w}}_2)+\int(\hat{\pmb{\mu}}_1-\hat{\pmb{\mu}}_2)\hat{\mathbf{w}}_1 \\
			&\equiv\mathrm{I}+\mathrm{II}+\mathrm{III}.
		\end{aligned}
	\end{equation}
	For $\mathrm{I}$, it can be decomposed into four parts 
	which have the same bound of $t^{-\frac{3}{4}}$:
	\begin{equation*}
		\begin{aligned}
			\int \hat{\mathbf{w}}_1-\hat{\mathbf{w}}_2
			&=\int_{\xi}^{+\infty}\hat{\mathbf{w}}_1^--\hat{\mathbf{w}}_2^--\int_{-\infty}^{\xi}\hat{\mathbf{w}}_1^--\hat{\mathbf{w}}_2^-\\
			&+\int_{\xi}^{+\infty}\hat{\mathbf{w}}_1^+-\hat{\mathbf{w}}_2^+-\int_{-\infty}^{\xi}\hat{\mathbf{w}}_1^+-\hat{\mathbf{w}}_2^+,
		\end{aligned}
	\end{equation*}
	we only give the rest of the details for the first term, the other terms 
	are easily inferred.
	
	Set $\mathbf{f}(\lambda)=\mathbf{R}(\lambda)-[\mathbf{R}](\lambda)\in H^{1,1}(\mathbb{R})$, 
	then 
	\begin{equation}
		\left|\int_{\xi}^{+\infty}\hat{\mathbf{w}}_1^--\hat{\mathbf{w}}_2^-\right|
		\leq\left|\int_{\xi}^{\xi+t^{-1/2}}
		\ee^{-2\ii t\theta}\det\pmb{\delta}_-\mathbf{f}\pmb{\delta}_-\right|
		+\left|\int_{\xi+t^{-1/2}}^{+\infty}
		\ee^{-2\ii t\theta}\det\pmb{\delta}_-\mathbf{f}\pmb{\delta}_-\right|.
	\end{equation}
	Since $\mathbf{f}(\lambda)\in H^{1,1}(\mathbb{R})$ and $\mathbf{f}(\xi)=0$, we have that 
	$\left|\mathbf{f}(\lambda)\right|\leq\left|\lambda-\xi\right|^{1/2}\|\mathbf{f}\|_{H^{1,1}}$. Then 
	\begin{equation*}
		\begin{aligned}
			\left|\int_{\xi}^{\xi+t^{-1/2}}
			\ee^{-2\ii t\theta}\det\pmb{\delta}_-\mathbf{f}\pmb{\delta}_-\dd\lambda\right|
			&=\left|\int_{0}^{1}\ee^{-2i\lambda_1^2+2\ii t\xi^2}[\det\pmb{\delta}_-\mathbf{f}\pmb{\delta}_-]
			(\frac{\lambda_1}{\sqrt{t}}+\xi)\dd\lambda_1t^{-\frac{1}{2}}\right|\\   
			&\leq t^{-\frac{1}{2}}\|\det\pmb{\delta}_-\|_{L^{\infty}}\|\pmb{\delta}_-\|_{L^{\infty}}
			\int_{0}^{1}\left|\mathbf{f}(\frac{\lambda_1}{\sqrt{t}}+\xi)\right|\dd\lambda_1\\
			&\lesssim t^{-\frac{1}{2}}\int_{0}^{1}\left|\lambda_1t^{-1/2}\right|^{1/2}\dd\lambda_1\\
			&\lesssim t^{-\frac{3}{4}}.
		\end{aligned}
	\end{equation*}
	For the second term, 
	\begin{equation*}
		\begin{aligned}
			&\quad\int_{\xi+t^{-1/2}}^{+\infty}\ee^{-2\ii t\theta}\det\pmb{\delta}_-\mathbf{f}\pmb{\delta}_-\dd\lambda\\
			&=\int_{\xi+t^{-1/2}}^{+\infty}\det\pmb{\delta}_-\mathbf{f}\pmb{\delta}_-\frac{\dd\ee^{-2\ii t\theta}}{-4\ii t(\lambda-\xi)}\\
			&=\frac{1}{4i\sqrt{t}}[\det\pmb{\delta}_-\mathbf{f}\pmb{\delta}_-](\xi+t^{-1/2})\\
			&+
			\frac{1}{4\ii t}\int_{\xi+t^{-1/2}}^{+\infty}\ee^{-2\ii t\theta}
			(\frac{(\det\pmb{\delta}_-)^{\prime}\mathbf{f}\pmb{\delta}_-}{\lambda-\xi}+
			\frac{\det\pmb{\delta}_-\mathbf{f}^{\prime}\pmb{\delta}_-}{\lambda-\xi}+
			\frac{\det\pmb{\delta}_-\mathbf{f}(\pmb{\delta}_-)^{\prime}}{\lambda-\xi}-
			\frac{\det\pmb{\delta}_-\mathbf{f}\pmb{\delta}_-}{(\lambda-\xi)^2})\dd\lambda\\
			&\equiv\mathrm{II}_1+\mathrm{II}_2+\mathrm{II}_3+\mathrm{II}_4+\mathrm{II}_5.
		\end{aligned}
	\end{equation*}
	For $\mathrm{II}_1$, 
	\begin{equation*}
		\left|\mathrm{II}_1\right|\lesssim 
		t^{-1/2}\|\det\pmb{\delta}_-\|_{L^{\infty}}\|\pmb{\delta}_-\|_{L^{\infty}}\left|\mathbf{f}(\xi+t^{-1/2})\right|
		\lesssim t^{-\frac{3}{4}}.
	\end{equation*}
	For $\mathrm{II}_2$, we consider the boundedness of $(\det\pmb{\delta}_-)^{\prime}$ 
	as $(X^{-1}_+)^{\prime}$ in Lemma \ref{bdd-V3-V4}, then we have 
	\begin{equation*}
		\begin{aligned}
			&|(\det\pmb{\delta}_-)^{\prime}|\le\mathrm{II}^{(1)}+\mathrm{II}^{(2)},\\
			&|\mathrm{II}^{(1)}|\lesssim\frac{1}{|s-\xi|},\quad
			\|\mathrm{II}^{(2)}\|_{L^2}\lesssim1.
		\end{aligned}
	\end{equation*}
	
	Therefore
	\begin{equation*}
		\begin{aligned}
			\left|\mathrm{II}_2\right|
			&\lesssim t^{-1}\|\pmb{\delta}_-\|_{L^{\infty}}
			\int_{\xi+t^{-1/2}}^{+\infty}
			\frac{(\left|\mathrm{II}^{(1)}\right|+\left|\mathrm{II}^{(2)}\right|)\left|\mathbf{f}\right|}
			{\left|\lambda-\xi\right|}\dd\lambda\\
			&\lesssim t^{-1}(\int_{\xi+t^{-1/2}}^{+\infty}
			\frac{\left|\lambda-\xi\right|^{1/2}}{\left|\lambda-\xi\right|^2}d\lambda
			+\|\mathrm{II}^{(2)}\|_{L^{2}}\|\mathbf{f}\|_{L^{\infty}}
			(\int_{\xi+t^{-1/2}}^{+\infty}\frac{1}{\left|\lambda-\xi\right|^2}\dd\lambda)^{1/2})\\
			&\lesssim t^{-\frac{3}{4}}.
		\end{aligned}
	\end{equation*}
	For $\mathrm{II}_3$, 
	\begin{equation*}
		\left|\mathrm{II}_3\right|\lesssim t^{-1}
		\|\mathbf{f}^{\prime}\|_{L^{2}}
		(\int_{\xi+t^{-1/2}}^{+\infty}\frac{1}{\left|\lambda-\xi\right|^2}\dd\lambda)^{1/2}
		\lesssim t^{-\frac{3}{4}}.
	\end{equation*}
	For $\mathrm{II}_4$, 
	by the boundedness of 
	$\|\pmb{\delta}_-^{\prime}\|_{L^2}$ in Lemma \ref{bdd-V3-V4},
	we have 
	\begin{equation*}
		\left|\mathrm{II}_4\right|
		\lesssim t^{-1}\|\pmb{\delta}^{\prime}_-\|_{L^{2}}
		(\int_{\xi+t^{-1/2}}^{+\infty}\frac{1}{\left|\lambda-\xi\right|^2}\dd\lambda)^{1/2}
		\lesssim t^{-\frac{3}{4}}.
	\end{equation*}
	For $\mathrm{II}_5$, 
	\begin{equation*}
		\left|\mathrm{II}_5\right|\lesssim t^{-1}
		\int_{\xi+t^{-1/2}}^{+\infty}
		\frac{\left|\lambda-\xi\right|^{1/2}}{\left|\lambda-\xi\right|^2}\dd\lambda
		\lesssim t^{-\frac{3}{4}}.
	\end{equation*}
	By the above calculations,
	\begin{equation*}
		\left|\int_{\xi+t^{-1/2}}^{+\infty}
		\ee^{-2\ii t\theta}\det\pmb{\delta}_-\mathbf{f}\pmb{\delta}_-\right|\lesssim t^{-\frac{3}{4}}.
	\end{equation*}
	Therefore $\left|\mathrm{I}\right|\lesssim t^{-\frac{3}{4}}$.
	
	For $\mathrm{II}$, by triangularity,
	\begin{equation*}
		\begin{aligned}
			\mathrm{II}
			&=\int(C_{\hat{\mathbf{w}}_2}\hat{\pmb{\mu}}_2)(\hat{\mathbf{w}}_1-\hat{\mathbf{w}}_2)\\
			&=\int (C^+_{\mathbb{R}_{\xi}} \hat{\pmb{\mu}}_2 \hat{\mathbf{w}}_2^-)(\hat{\mathbf{w}}_1^+-\hat{\mathbf{w}}_2^+)+
			\int (C^- _{\mathbb{R}_{\xi}}\hat{\pmb{\mu}}_2 \hat{\mathbf{w}}_2^+)(\hat{\mathbf{w}}_1^--\hat{w}_2^-)\\
			&\equiv\mathrm{II}_++\mathrm{II}_- .
		\end{aligned}
	\end{equation*}
	Extending the integrand trivially from $\mathbb{R}_\xi$ to a complete contour
	$\Gamma_{\xi}=\mathbb{R}\cup(\xi+\ee^{\ii\pi/2}\mathbb{R}_-)\cup(\xi+\ee^{-\ii\pi/2}\mathbb{R}_-)$ 
	and using the following properties
	\begin{equation}\label{property-C}
		C_{\Gamma_{\xi}}^+C_{\Gamma_{\xi}}^-
		=C_{\Gamma_{\xi}}^-C_{\Gamma_{\xi}}^+=0,\quad
		C^+_{\Gamma_{\xi}}-C^-_{\Gamma_{\xi}}=1,
	\end{equation}
	we can have 
	\begin{equation*}
		\begin{aligned}
			|\mathrm{II}_+|
			&=\left|\int (C^+_{\Gamma_{\xi}} \hat{\pmb{\mu}}_2 \hat{\mathbf{w}}_2^-)(C^+_{\Gamma_{\xi}}-C^-_{\Gamma_{\xi}})(\hat{\mathbf{w}}_1^+-\hat{\mathbf{w}}_2^+)\right|\\
			&=\left|-\int (\hat{\mathbf{w}}_1^+-\hat{\mathbf{w}}_2^+)C^-_{\Gamma_{\xi}}C^+_{\Gamma_{\xi}} (\hat{\pmb{\mu}}_2 \hat{\mathbf{w}}_2^-)
			-\int (C^+_{\Gamma_{\xi}} \hat{\pmb{\mu}}_2 \hat{\mathbf{w}}_2^-)C^-_{\Gamma_{\xi}}(\hat{\mathbf{w}}_1^+-\hat{\mathbf{w}}_2^+)\right|\\
			&=\left|\int (C^+_{\Gamma_{\xi}} \hat{\pmb{\mu}}_2 \hat{\mathbf{w}}_2^-)C^-_{\Gamma_{\xi}}(\hat{\mathbf{w}}_1^+-\hat{\mathbf{w}}_2^+)\right|.
		\end{aligned}
	\end{equation*}
	By Lemma \ref{lemma3-CNLS} and Lemma \ref{lemma1-CNLS}, 
	\begin{equation*}
		\begin{aligned}
			\|C^+ _{\Gamma_{\xi}}\hat{\pmb{\mu}}_2 \hat{\mathbf{w}}_2^-\|_{L^2}
			&\leq\|C^+_{\Gamma_{\xi}} (\hat{\pmb{\mu}}_2-\mathbb{I}_3) \hat{\mathbf{w}}_2^-\|_{L^2}+
			\|C^+ _{\Gamma_{\xi}} \hat{\mathbf{w}}_2^-\|_{L^2}\\
			&\lesssim\|\hat{\pmb{\mu}}_2-\mathbb{I}_3\|_{L^2}+t^{-\frac{1}{4}}\\
			&\lesssim t^{-\frac{1}{4}}.
		\end{aligned}
	\end{equation*}
	By Lemma \ref{lemma2-CNLS},
	\begin{equation*}
		\|C^-_{\Gamma_{\xi}}(\hat{\mathbf{w}}_1^+-\hat{\mathbf{w}}_2^+)\|_{L^2}
		\lesssim t^{-\frac{1}{2}}.
	\end{equation*}

	Hence 
	\begin{equation*}
		\|\mathrm{II}_+\|_{L^2}\leq\|C^+_{\Gamma_{\xi}} \hat{\pmb{\mu}}_2 \hat{\mathbf{w}}_2^-\|_{L^2}
		\|C^-_{\Gamma_{\xi}}(\hat{\mathbf{w}}_1^+-\hat{\mathbf{w}}_2^+)\|_{L^2}
		\lesssim t^{-\frac{3}{4}}.
	\end{equation*}
	
	The estimate for $\mathrm{II}_-$ is similar and therefore 
	$\left|\mathrm{II}\right|\lesssim t^{-\frac{3}{4}}$. Then we compute
	\begin{equation*}
		\hat{\pmb{\mu}}_1-\hat{\pmb{\mu}}_2
		=(1-C_{\hat{\mathbf{w}}_2})^{-1}C_{\hat{\mathbf{w}}_1-\hat{\mathbf{w}}_2}\hat{\pmb{\mu}}_1
		=C_{\hat{\mathbf{w}}_1-\hat{\mathbf{w}}_2}\hat{\pmb{\mu}}_1+C_{\hat{\mathbf{w}}_2}h ,
	\end{equation*}
	where $\pmb{h}=(1-C_{\hat{\mathbf{w}}_2})^{-1}C_{\hat{\mathbf{w}}_1-\hat{\mathbf{w}}_2}\hat{\pmb{\mu}}_1$. 
	Again by triangularity and equation \eqref{property-C},
	\begin{equation*}
		\begin{aligned}
			\mathrm{III}
			&=\int(C_{\hat{\mathbf{w}}_1-\hat{\mathbf{w}}_2}\hat{\pmb{\mu}}_1)\hat{\mathbf{w}}_1+\int(C_{\hat{\mathbf{w}}_2}\pmb{h})\hat{\mathbf{w}}_1 \\
			&=-\int\left(C^+_{\mathbb{R}_\xi}\hat{\pmb{\mu}}_1(\hat{\mathbf{w}}_1^--\hat{\mathbf{w}}_2^-)\right)C^-_{\mathbb{R}_\xi}
			\hat{\mathbf{w}}_1^++
			\int\left(C^-_{\mathbb{R}_\xi}\hat{\pmb{\mu}}_1(\hat{\mathbf{w}}_1^+-\hat{\mathbf{w}}_2^+)\right)C^+_{\mathbb{R}_\xi}\hat{\mathbf{w}}_1^- \\
			&\quad-\int\left(C^+_{\mathbb{R}_\xi}\pmb{h}\mathbf{w}_2^-\right)C^-_{\mathbb{R}_\xi}\hat{\mathbf{w}}_1^++
			\int\left(C^-_{\mathbb{R}_\xi}\pmb{h}\mathbf{w}_2^+\right)C^+_{\mathbb{R}_\xi}\hat{\mathbf{w}}_1^-\\
			&\equiv \hat{\mathrm{III}}_++\hat{\mathrm{III}}_-
			+\tilde{\mathrm{III}}_++\tilde{\mathrm{III}}_- .
		\end{aligned}
	\end{equation*}
	By Lemma \ref{lemma1-CNLS} and Lemma \ref{lemma4-CNLS},
	\begin{equation*}
		\left|\hat{\mathrm{III}}_+\right|\leq\|C^+_{\mathbb{R}_\xi}\hat{\pmb{\mu}}_1(\hat{\mathbf{w}}_1^--\hat{\mathbf{w}}_2^-)\|_{L^2}
		\|C^-_{\mathbb{R}_\xi}\hat{\mathbf{w}}_1^+\|_{L^2}\lesssim t^{-\frac{1}{2}+\frac{1}{(2p)}}t^{-\frac{1}{4}}
		\lesssim t^{-\frac{3}{4}+\frac{1}{(2p)}},
	\end{equation*}
	with a similar estimate for $\mathrm{III}_-^{'}$. Again by Lemma \ref{lemma4-CNLS}, 
	$\|C^{\pm}_{\mathbb{R}_\xi}\hat{\pmb{\mu}}_1(\hat{\mathbf{w}}_1^{\mp}-\hat{\mathbf{w}}_2^{\mp})\|_{L^2}\lesssim t^{-\frac{1}{2}+\frac{1}{(2p)}}$ 
	and hence $\|\pmb{h}\|_{L^2}\lesssim t^{-\frac{1}{2}+\frac{1}{(2p)}}$. Together with 
	Lemma \ref{lemma1-CNLS}, this implies that 
	$\left|\tilde{\mathrm{III}}_{\pm}\right|\lesssim t^{-\frac{3}{4}+\frac{1}{(2p)}}$ and 
	therefore $\left|\mathrm{III}\right|\lesssim t^{-\frac{3}{4}+\frac{1}{(2p)}}$. 
	Adding the estimates for $\mathrm{I}$, $\mathrm{II}$, and $\mathrm{III}$, the result follows immediately.
\end{proof}

\section*{Appendix B}
\setcounter{equation}{0} 
\renewcommand{\theequation}{B.\arabic{equation}} 
In the Appendix B, we will provide a detailed procedure for solving RHP \ref{rhp-N} and 
obtaining the long-time asymptotic behavior of the solution $\mathbf{N}(y)$. 
Set $\pmb{\sigma}_3=
\begin{pmatrix}
	1&0\\
	0&-1
\end{pmatrix}$, $\kappa_\xi=\kappa(\xi)$  and 
$\pmb{\Psi}(y)=\eta^{-\pmb{\sigma}_3}\mathbf{N}(y)\eta^{\pmb{\sigma}_3}y^{\ii\kappa_\xi
	\pmb{\sigma}_3}\mathrm{\ee}^{-\frac14\ii y^2\pmb{\sigma}_3}$, 
then we have 
\begin{equation*}
	\pmb{\Psi}_+(y)=\pmb{\Psi}_-(y)\mathbf{v}_0,\quad \mathbf{v}_0=y^{-\ii\kappa_\xi\hat{\pmb{\sigma}}_3}
	\mathrm{e}^{\frac14\ii y^2\hat{\pmb{\sigma}}_3}
	\eta^{-\hat{\pmb{\sigma}}_3}\V_{N}(y).
\end{equation*}
Clearly, the jump matrix $\mathbf{v}_0$ is constant along each ray. Thus
\begin{equation*}
	\frac{\mathrm{d}\pmb{\Psi}_+(y)}{\mathrm{d}y}=
	\frac{\mathrm{d}\pmb{\Psi}_-(y)}{\mathrm{d}y}\mathbf{v}_0,
\end{equation*}
and $\frac{\mathrm{d}\pmb{\Psi}(y)}{\mathrm{d}y}\pmb{\Psi}^{-1}(y)$ has no jump discontinuity 
along any of the rays. Moreover, it follows from the expansion
$\mathbf{N}(y)=\mathbb{I}_3+\frac{\mathbf{N}_1}{y}+O(y^{-2}),\quad y\rightarrow\infty$ that 
\begin{equation*}
	\begin{aligned}
		\frac{\mathrm{d}\pmb{\Psi}}{\mathrm{d}y}\pmb{\Psi}^{-1}
		&=\eta^{-\hat{\pmb{\sigma}}_3}\left(\frac{\dd \mathbf{N}}{\dd y}\mathbf{N}^{-1}-
		\frac{1}{2}\ii y\mathbf{N}\pmb{\sigma}_3\mathbf{N}^{-1}+
		\ii\kappa_\xi y^{-1}\mathbf{N}\pmb{\sigma}_3\mathbf{N}^{-1}\right)\\
		&=\oo(y^{-1})-\frac{1}{2}\ii y\pmb{\sigma}_3-\frac12\ii\eta^{-\hat{\pmb{\sigma}}_3}[\mathbf{N}_1,\pmb{\sigma}_3].
	\end{aligned}
\end{equation*}
By the Liouville theorem, we can obtain 
\begin{equation}\label{2*2Psi-equation}
	\frac{\mathrm{d}\pmb{\Psi}(y)}{\mathrm{d}y}+\frac12\ii y\pmb{\sigma}_3\pmb{\Psi}(y)
	=\beta\pmb{\Psi}(y),
\end{equation}
where 
\begin{equation}\label{relation-beta-N}
	\beta=-\frac12\ii\eta^{-\hat{\pmb{\sigma}}_3}[\mathbf{N}_1,\pmb{\sigma}_3]
	=\begin{pmatrix}
		0&\ii\eta^{-2}[\mathbf{N}_1]_{12}\\
		-\ii\eta^{2}[\mathbf{N}_1]_{21}&0
	\end{pmatrix}=
	\begin{pmatrix}
		0&\beta_{12}\\
		\beta_{21}&0
	\end{pmatrix}.
\end{equation}
We can further show that the solution $\mathbf{N}(y)$ for RHP \ref{rhp-N} is unique, 
and then obtain the symmetry of $\mathbf{N}(y)$ 
\begin{equation}
	\mathbf{N}^{-1}(y)=\mathbf{N}^\dagger(y^*),
\end{equation}
which implies 
\begin{equation}\label{relation-b12-b21}
	\beta_{12}=-\beta^*_{21}.
\end{equation}

From equation \eqref{2*2Psi-equation} we obain
\begin{equation}\label{relation-Psi-ij}
	\begin{aligned}
		\frac{\mathrm{d}\pmb{\Psi}_{11}(y)}{\mathrm{d}y}+\frac \ii2y\pmb{\Psi}_{11}& =\beta_{12}\pmb{\Psi}_{21},\\
		\frac{\mathrm{d}\pmb{\Psi}_{12}(y)}{\mathrm{d}y}+\frac \ii2y\pmb{\Psi}_{12}& =\beta_{12}\pmb{\Psi}_{22},\\
		\frac{\mathrm{d}\pmb{\Psi}_{21}(y)}{\mathrm{d}y}-\frac \ii2y\pmb{\Psi}_{21}& =\beta_{21}\pmb{\Psi}_{11},\\
		\frac{\mathrm{d}\pmb{\Psi}_{22}(y)}{\mathrm{d}y}-\frac \ii2y\pmb{\Psi}_{22}& =\beta_{21}\pmb{\Psi}_{12},
	\end{aligned}
\end{equation}
and
\begin{equation}\label{psi11-psi22}
	\begin{aligned}
		\frac{\mathrm{d}^2\pmb{\Psi}_{11}(y)}{\mathrm{d}y^2}& =
		\left(-\frac \ii2-\frac14y^2+\beta_{12}\beta_{21}\right)\pmb{\Psi}_{11}(y), \\
		\frac{\mathrm{d}^2\pmb{\Psi}_{22}(y)}{\mathrm{d}y^2}& 
		=\left(\frac \ii2-\frac14y^2+\beta_{12}\beta_{21}\right)\pmb{\Psi}_{22}(y).
	\end{aligned}
\end{equation}

As is well known, the Weber equation
\begin{equation}\label{weber-eq}
	\frac{\mathrm{d}^2g(\zeta)}{\mathrm{d}\zeta^2}+
	\left(a+\frac12-\frac{\zeta^2}4\right)g(\zeta)=0
\end{equation}
has the solution
\begin{equation}\label{soliton-weber}
	g(\zeta)=c_1D_a(\zeta)+c_2D_a(-\zeta),
\end{equation}
where $D_a(\cdot)$ denotes the standard parabolic-cylinder function \cite{NIST_DLMF_2016} and satisfies
\begin{equation}\label{parabolic-cylinder-function}
	\begin{aligned}
		&\frac{\mathrm{d}D_a(\zeta)}{\mathrm{d}\zeta}+\frac\zeta2D_a(\zeta)-aD_{a-1}(\zeta)=0,\\
		&D_a(\pm\zeta)=\frac{\Gamma(a+1)\mathrm{e}^{\frac{\ii\pi a}2}}
		{\sqrt{2\pi}}D_{-a-1}(\pm \ii\zeta)+\frac{\Gamma(a+1)\mathrm{e}^{-\frac{\ii\pi a}2}}
		{\sqrt{2\pi}}D_{-a-1}(\mp \ii\zeta).
	\end{aligned}
\end{equation}
From Whittaker and Watson \cite{Whittaker_Watson_1927}, we obain the asymptotic properties of $D_a(\cdot)$
\begin{equation}\label{asymptotic-Da}
	D_a(\zeta)=
	\begin{cases}
		\zeta^a \ee^{-\frac{\zeta^2}{4}}(1+\mathcal{O}(\zeta^{-2})),&\quad|\arg(\zeta)|<\frac{3\pi}{4},\\
		\zeta^a \ee^{-\frac{\zeta^2}{4}}(1+\mathcal{O}(\zeta^{-2}))
		-\frac{\sqrt{2\pi}}{\Gamma(-a)}\mathrm e^{a\pi \ii
			+\frac{\zeta^2}{4}}\zeta^{-a-1}(1+\mathcal{O}(\zeta^{-2})),&\quad\frac{\pi}{4}<\arg(\zeta)<\frac{5\pi}{4},\\
		\zeta^a\mathrm e^{-\frac{\zeta^2}{4}}(1+\mathcal{O}(\zeta^{-2}))
		-\frac{\sqrt{2\pi}}{\Gamma(-a)}\mathrm e^{-a\pi \ii
			+\frac{\zeta^2}{4}}\zeta^{-a-1}(1+\mathcal{O}(\zeta^{-2})),&\quad-\frac{5\pi}{4}<\arg(\zeta)<-\frac{\pi}{4},
	\end{cases}
\end{equation}
where $\Gamma$ is the Gamma function.
Set $a=-\ii\beta_{12}\beta_{21}$, it follows from equation \eqref{psi11-psi22} and 
equation \eqref{soliton-weber} that 
\begin{equation}
	\begin{aligned}
		&\pmb{\Psi}_{11}(y) =c_1D_{-a}\left(\mathrm{e}^{-\frac{3\pi \ii}{4}}y\right)
		+c_2D_{-a}\left(\mathrm{e}^{\frac{\pi \ii}{4}}y\right), \\
		&\pmb{\Psi}_{22}(y) =c_3D_a\left(\mathrm{e}^{\frac{3\pi \ii}4}y\right)+
		c_4D_a\left(\mathrm{e}^{-\frac{\pi \ii}4}y\right) .
	\end{aligned}
\end{equation}

As $y\rightarrow\infty$, 
$\pmb{\Psi}(y)y^{-\ii\kappa_\xi\pmb{\sigma}_3}\mathrm{e}^{\frac14\ii y^2\pmb{\sigma}_3}=\eta^{-\pmb{\sigma}_3}\mathbf{N}(y)\eta^{\pmb{\sigma}_3}\rightarrow\mathbb{I}_2$ and hence
\begin{equation}\label{asymptotic-Psi}
	\pmb{\Psi}_{11}(y)\rightarrow y^{\ii\kappa_\xi}\mathrm{e}^{-\frac14\ii y^2},\quad
	\pmb{\Psi}_{22}(y)\rightarrow y^{-\ii\kappa_\xi}\mathrm{e}^{\frac14\ii y^2}.
\end{equation}
If $\mathrm{arg}(y)\in[-\frac{\pi}{4},\frac{3\pi}{4}]$, we have
\begin{equation}
	D_a(y\ee^{-\frac{\pi}{4}\ii})\to 
	\ee^{-\frac{\pi}{4}\ii a} y^a \ee^{\frac14\ii y^2},
\end{equation}
and then $c_3=0$, $c_4=\ee^{\frac{\pi}{4}k_\xi}$, hence
\begin{equation*}
	\pmb{\Psi}_{22}(y)=\ee^{\frac{\pi\kappa_\xi}{4}}
	D_a(y\ee^{-\frac{\pi}{4}\ii}),\quad a=-\ii\kappa_\xi,\quad 
	\kappa_\xi=\beta_{12}\beta_{21},
\end{equation*}
by equation \eqref{asymptotic-Da} and equation \eqref{asymptotic-Psi}. Consequently, 
it follows from equation \eqref{relation-Psi-ij} and equation \eqref{parabolic-cylinder-function} that
\begin{equation*}
	\pmb{\Psi}_{12}(y)=\beta_{12}\ee^{\frac{\pi}{4}\kappa_\xi-\frac34\pi \ii}D_{a-1}(y\ee^{-\frac{\pi}{4}\ii}).
\end{equation*}
If $\arg(y)\in[\frac{3\pi}{4},\frac{7\pi}{4}]$, we can obain similar conclusions
\begin{equation*}
	\pmb{\Psi}_{22}(y)=\ee^{\frac{-3\pi\kappa_\xi}{4}}D_a(y\ee^{\frac{3\pi}{4}\ii}),\quad
	\pmb{\Psi}_{12}(y)=\beta_{12}\ee^{-\frac{3\pi\kappa_\xi}{4}+\frac14\pi \ii}D_{a-1}(y\ee^{\frac{3\pi}{4}\ii}).
\end{equation*}
If $\arg(y)\in[-\frac{3\pi}{4},\frac{\pi}{4}]$, we have
\begin{equation}
	D_{-a}(y\ee^{\frac{\pi}{4}\ii})\to 
	\ee^{-\frac{\pi}{4}ia} y^{-a} \ee^{-\frac14iy^2},
\end{equation}
and then $c_1=0$, $c_2=\ee^{\frac{\pi}{4}k_\xi}$, hence
\begin{equation*}
	\pmb{\Psi}_{11}(y)=\ee^{\frac{\pi\kappa_\xi}{4}}D_{-a}(y\ee^{\frac{\pi}{4}\ii}).
\end{equation*}
Consequently, 
it follows from equation \eqref{relation-Psi-ij} and equation \eqref{parabolic-cylinder-function} that
\begin{equation*}
	\pmb{\Psi}_{21}(y)=\beta_{21} \ee^{\frac{\pi}{4}\kappa_\xi+\frac34\pi \ii}D_{-a-1}(y\ee^{\frac{\pi}{4}\ii}).
\end{equation*}
If $\arg(y)\in[\frac{\pi}{4},\frac{5\pi}{4}]$, we can also obain similar conclusions
\begin{equation*}
	\pmb{\Psi}_{11}(y)=
	\ee^{-\frac{3\pi}{4}\kappa_\xi}D_{-a}(y\ee^{-\frac{3\pi \ii}{4}}),\quad
	\pmb{\Psi}_{21}(y)=\beta_{21}\ee^{-\frac{3\pi}{4}\kappa_\xi-\frac14\pi \ii}D_{-a-1}(y\ee^{-\frac{3\pi}{4}\ii}).
\end{equation*}

Hence, the unique solution $\pmb{\Psi}(y)$ to problem \eqref{2*2Psi-equation} can be given by
\begin{equation}\label{weber-equ-Psi}
	\pmb{\Psi}(y)=
	\begin{cases}
		\begin{bmatrix}
			\ee^{\frac{\pi\kappa_\xi}{4}}D_{-a}(y\ee^{\frac{\pi}{4}\ii})&\beta_{12}\ee^{\frac{\pi}{4}\kappa_\xi-\frac34\pi \ii}D_{a-1}(y\ee^{-\frac{\pi}{4}\ii})\\
			\beta_{21} \ee^{\frac{\pi}{4}\kappa_\xi+\frac34\pi \ii}D_{-a-1}(y\ee^{\frac{\pi}{4}\ii})&\ee^{\frac{\pi\kappa_\xi}{4}}D_a(y\ee^{-\frac{\pi}{4}\ii})
		\end{bmatrix},
		&\mathrm{arg}(y)\in[-\frac{\pi}{4},\frac{\pi}{4}],\\
		\begin{bmatrix}
			\ee^{-\frac{3\pi}{4}\kappa_\xi}D_{-a}(y\ee^{-\frac{3\pi \ii}{4}})&\beta_{12}\ee^{\frac{\pi}{4}\kappa_\xi-\frac34\pi \ii}D_{a-1}(y\ee^{-\frac{\pi}{4}\ii})\\
			\beta_{21}\ee^{-\frac{3\pi}{4}\kappa_\xi-\frac14\pi \ii}D_{-a-1}(y\ee^{-\frac{3\pi}{4}\ii})&\ee^{\frac{\pi\kappa_\xi}{4}}D_a(y\ee^{-\frac{\pi}{4}\ii})
		\end{bmatrix},
		&\mathrm{arg}(y)\in[\frac{\pi}{4},\frac{3\pi}{4}],\\
		\begin{bmatrix}
			\ee^{-\frac{3\pi}{4}\kappa_\xi}D_{-a}(y\ee^{-\frac{3\pi \ii}{4}})&\beta_{12}\ee^{-\frac{3\pi\kappa_\xi}{4}+\frac14\pi \ii}D_{a-1}(y\ee^{\frac{3\pi}{4}\ii})\\
			\beta_{21}\ee^{-\frac{3\pi}{4}\kappa_\xi-\frac14\pi \ii}D_{-a-1}(y\ee^{-\frac{3\pi}{4}\ii})&\ee^{\frac{-3\pi\kappa_\xi}{4}}D_a(y\ee^{\frac{3\pi}{4}\ii})
		\end{bmatrix},
		&\mathrm{arg}(y)\in[\frac{3\pi}{4},\frac{5\pi}{4}],\\
		\begin{bmatrix}
			\ee^{\frac{\pi\kappa_\xi}{4}}D_{-a}(y\ee^{\frac{\pi}{4}\ii})&\beta_{12}\ee^{-\frac{3\pi\kappa_\xi}{4}+\frac14\pi \ii}D_{a-1}(y\ee^{\frac{3\pi}{4}\ii})\\
			\beta_{21} \ee^{\frac{\pi}{4}\kappa_\xi+\frac34\pi \ii}D_{-a-1}(y\ee^{\frac{\pi}{4}\ii})&\ee^{\frac{-3\pi\kappa_\xi}{4}}D_a(y\ee^{\frac{3\pi}{4}\ii})
		\end{bmatrix},
		&\mathrm{arg}(y)\in[-\frac{3\pi}{4},-\frac{\pi}{4}].\\
	\end{cases}
\end{equation}
Along the ray $\arg(y)=-\frac\pi4$, we have 
$\pmb{\Psi}_{12}^+=|\pmb{A}(\xi)|\pmb{\Psi}_{11}^-+\pmb{\Psi}_{12}^-$ and hence
\begin{equation}\label{coefficient1}
	|\pmb{A}(\xi)|\ee^{\frac{\pi}{4}\kappa}D_{-a}(y\ee^{\frac{\pi}{4}\ii})=
	-\beta_{12}\ee^{\frac{\pi}{4}\kappa-\frac34\pi \ii}D_{a-1}(y\ee^{-\frac{\pi}{4}\ii})
	+\beta_{12}\ee^{-\frac{3\pi}{4}\kappa+\frac14\pi \ii}D_{a-1}(y\ee^{\frac{3\pi}{4}\ii}).
\end{equation}
It follows from equation \eqref{parabolic-cylinder-function} that 
\begin{equation}\label{coefficient2}
	\begin{aligned}
		|\pmb{A}(\xi)|\ee^{\frac{\pi}{4}\kappa_\xi}D_{-a}(y\ee^{\frac{\pi}{4}\ii})
		=&|\pmb{A}(\xi)|\ee^{\frac{\pi}{4}\kappa_\xi}
		\frac{1}{\sqrt{2\pi}}\Gamma(-a+1)\ee^{-\frac\pi2\ii a}D_{a-1}(y\ee^{\frac{3\pi}{4}\ii})\\
		&+|\pmb{A}(\xi)|\ee^{\frac{\pi}{4}\kappa_\xi}
		\frac{1}{\sqrt{2\pi}}\Gamma(-a+1)\ee^{\frac\pi2\ii a}D_{a-1}(y\ee^{-\frac{\pi}{4}\ii}).
	\end{aligned}
\end{equation}
Comparing the coefficient in equation \eqref{coefficient1} and equation \eqref{coefficient2}, we have
\begin{equation}\label{def-beta}
	\beta_{12}=
	\frac{|\pmb{A}(\xi)|}{\sqrt{2\pi}}\ee^{\frac{\pi\kappa_{\xi}}{2}+
		\frac{5\pi \ii}{4}}\kappa_{\xi}\Gamma(\ii\kappa_{\xi}).
\end{equation}
\newpage
\bibliographystyle{unsrt}
\bibliography{reference}

@article{deift1994long,
author = {Deift, P.A. and Zhou, X.},
title = {{Long-time asymptotics for integrable systems. Higher order theory}},
volume = {165},
journal = {Communications in Mathematical Physics},
number = {1},
publisher = {Springer},
pages = {175--191},
year = {1994},
}

@article{deift2002long,
  title={{Long-time asymptotics for solutions of the NLS equation with initial data in a weighted Sobolev space}},
  author={Deift, P.A. and Zhou, X.},
  journal={Communications on Pure and Applied Mathematics},
  year={2003},
  volume={56},
  number={8},
  pages={1029--1077},
  publisher={Wiley Online Library}
}

@article{deift2011long,
  title={{Long-time asymptotics for solutions of the NLS equation with a delta potential and even initial data}},
  author={Deift, P.A. and Park, J.},
  journal={International Mathematics Research Notices},
  volume={2011},
  number={24},
  pages={5505--5624},
  year={2011},
  publisher={OUP}
}

@article{manakov1974theory,
  title={{On the theory of two-dimensional stationary self-focusing of electromagnetic waves}},
  author={Manakov, S.V.},
  journal={Soviet Physics-JETP},
  volume={38},
  number={2},
  pages={248--253},
  year={1974},
  publisher={AIP}
}

@book{novikov1984theory,
  title={{Theory of solitons: the inverse scattering method}},
  author={Novikov, S. and Manakov, S.V. and Pitaevskii, L.P. and Zakharov, V.E.},
  year={1984},
  publisher={Soviet Physics JETP}
}

@article{busch2001dark,
  title={{Dark-bright solitons in inhomogeneous Bose-Einstein condensates}},
  author={Busch, T. and Anglin, J.R.},
  journal={Physical Review Letters},
  volume={87},
  number={1},
  pages={010401},
  year={2001},
  publisher={APS}
}

@book{agrawal2000nonlinear,
  title={{Nonlinear fiber optics}},
  author={Agrawal, G.P.},
  title={Nonlinear Science at the Dawn of the 21st Century},
  pages={195--211},
  year={2000},
  publisher={Springer Berlin Heidelberg}
}

@article{scott1984launching,
  title={{Launching a Davydov soliton: I. Soliton analysis}},
  author={Scott, A.C.},
  journal={Physica Scripta},
  volume={29},
  number={3},
  pages={279},
  year={1984},
  publisher={IOP Publishing}
}

@article{deift1993steepest,
  title={{A steepest descent method for oscillatory Riemann--Hilbert problems. Asymptotics for the mKdV equation}},
  author={Deift, P.A. and Zhou, X.},
  journal={Annals of Mathematics},
  year={1993},
  volume = {137},
  number = {2},
  pages = {295--368},
  publisher = {Princeton University Press},
}

@article{Liu_2019,
year = {2019},
publisher = {IOP},
volume = {32},
number = {3},
pages = {1012--1043},
author = {Liu, J.Q.},
title = {{$L^2$-Sobolev space bijectivity of the inverse scattering of a $3\times3$ AKNS system}},
journal = {Nonlinearity}
}

@article{zhouscatteringtransforms1989,
author = {Zhou, X.},
title = {Direct and inverse scattering transforms with arbitrary spectral singularities},
journal = {Communications on Pure and Applied Mathematics},
volume = {42},
number = {7},
pages = {895-938},
publisher = {Springer},
year = {1989}
}

@book{ablowitz2003complex,
  title={Complex variables: introduction and applications},
  author={Ablowitz, M.J. and Fokas, A.S.},
  year={2003},
  publisher = {McGraw-Hill}
}

@article{beals1984scattering,
  title={Scattering and inverse scattering for first order systems},
  author={Beals, R. and Coifman, R.R.},
  journal={Communications on Pure and Applied Mathematics},
  volume={37},
  number={1},
  pages={39--90},
  year={1984},
  publisher={Wiley Online Library}
}

@article{its1981asymptotics,
  title={{Asymptotics of solutions of the nonlinear Schrödinger equation and isomonodromic deformations of systems of linear differential equations}},
  author={Its, A.R.},
  journal={Dokl. Akad. Nauk SSSR},
  volume={261},
  number={1},
  pages={14--18},
  year={1981},
  organization={Russian Academy of Sciences}
}

@article{dieng2008long,
  title={{Long-time asymptotics for the NLS equation via $\bar{\partial}$ methods}},
  author={Dieng, M.and McLaughlin, KDT-R},
  journal={arXiv preprint arXiv:0805.2807},
  year={2008}
}

@article{grunert2009long,
  title={{Long-time asymptotics for the Korteweg--de Vries equation via nonlinear steepest descent}},
  author={Grunert, K. and Teschl, G.},
  journal={Mathematical Physics, Analysis and Geometry},
  volume={12},
  number={3},
  pages={287--324},
  year={2009},
  publisher={Springer}
}

@article{jenkins2017global,
  author = {Jenkins, R. and Liu, J.Q. and Perry, P.A. and Sulem, C.},
  title = {{Global well-posedness for the derivative non-linear Schrödinger equation}},
  journal = {Communications in Partial Differential Equations},
  volume = {43},
  number = {8},
  pages = {1151--1195},
  year = {2018},
  publisher = {Taylor $\&$ Francis}
}

@article{liujiaqi2018long,
  author = {Liu, J.Q. and Perry, P.A. and Sulem, C.},
  title = {{Long-time behavior of solutions to the derivative nonlinear Schrödinger equation for soliton-free initial data}},
  journal = {Annales de l'I.H.P. Analyse non linéaire},
  pages = {217--265},
  publisher = {Elsevier},
  volume = {35},
  number = {1},
  year = {2018},
}

@article{NIST_DLMF_2016,
  author = {Lozier, D.W.},
  title = {{NIST Digital Library of Mathematical Functions}},
  year = {2003},
  journal = {Annals of Mathematics and Artificial Intelligence},
  volume = {38},
  pages = {105-119},
  publisher = {Springer}
}

@book{Whittaker_Watson_1927,
  author = {Whittaker, E.T. and Watson, G.N.},
  title = {A Course of Modern Analysis},
  publisher = {Cambridge University Press},
  year = {1927}
}

@article{zakharov1976asymptotic,
  author = {Zakharov, V.E. and Manakov, S.V.},
  title = {Asymptotic behavior of non-linear wave systems integrated by the inverse scattering method},
  journal = {Soviet Physics Journal},
  volume = {21},
  number = {5},
  pages = {940--944},
  year = {1976}
}

@article{McLaughlin2006,
	author = {McLaughlin, K.T.-R. and Miller, P.D.},
	title = {{The $\bar{\partial}$ steepest descent method and the asymptotic behavior of polynomials orthogonal on the unit circle with fixed and exponentially varying non-analytic weights}},
	year = {2006},
	journal = {International Mathematics Research Papers},
	publisher = {Elsevier B.V.},
  volume = {38},
  number = {5},
  pages = {1835--1866},
}

@article{McLaughlin2008,
    author = {McLaughlin, K.T.-R. and Miller, P.D.},
    title = {{The steepest descent method for orthogonal polynomials on the real line with varying weights}},
    journal = {International Mathematics Research Notices},
    volume = {2008},
    pages = {rnn075},
    year = {2008},
    publisher={Wiley Online Library}
}

@article{Chiao1964,
	author = {Chiao, R.Y. and Garmire, E. and Townes, C.H.},
	title = {Self-trapping of optical beams},
	year = {1964},
	journal = {Physical Review Letters},
	volume = {13},
	number = {15},
	pages = {479-482},
	doi = {10.1103/PhysRevLett.13.479},
	publisher={Elsevier B.V.}
}

@article{Zakharov1968,
	author = {Zakharov, V.E.},
	title = {Stability of periodic waves of finite amplitude on the surface of a deep fluid},
	year = {1968},
	volume = {9},
	number = {2},
	pages = {190-194},
	publisher={Elsevier B.V.},
  journal={Journal of Fluid Mechanics}
}

@article{YAN2011,
title = {Vector financial rogue waves},
journal = {Physics Letters A},
volume = {375},
number = {48},
pages = {4274-4279},
year = {2011},
author = {Yan, Z.Y.},
publisher={Springer}
}

@article{Guo2011,
	author = {Guo, B.L. and Ling, L.M.},
	title = {{Rogue wave, breathers and bright-dark-rogue solutions for the coupled Schrödinger equations}},
	year = {2011},
	journal = {Chinese Physics Letters},
	volume = {28},
	number = {11},
	publisher={Elsevier B.V.}
}

@article{Mu2015,
author = {Mu, G. and Qin, Z.Y. and Grimshaw, R.},
title = {{Dynamics of Rogue Waves on a Multisoliton Background in a Vector Nonlinear Schrödinger Equation}},
journal = {SIAM Journal on Applied Mathematics},
volume = {75},
number = {1},
pages = {1-20},
year = {2015},
publisher={Springer}
}

@book{Deift1993_long,
author="Deift, P.A.
and Its, A.R.
and Zhou, X.",
title="Important Developments in Soliton Theory",
year="1993",
publisher="Springer Berlin Heidelberg"
}

@book{deift1994casestudy,
  title={{Long-time behavior of the focusing nonlinear Schrodinger equation--a case study}},
  author={Deift, P.A. and Zhou, X.},
  title={New series: lectures in mathematical sciences},
  year={1994},
  publisher={University of Tokyo}
}

@article{McLaughlin_2018,
  author = {Borghese, M. and Jenkins, R.and McLaughlin, K.D.T.-R.},
  title = {{Long time asymptotic behavior of the focusing nonlinear Schrödinger equation}},
  journal = {Annales de l'I.H.P. Analyse non linéaire},
  volume = {35},
  number = {4},
  pages = {887--920},
  year = {2018},
  publisher = {Elsevier},
  doi = {10.1016/j.anihpc.2017.08.006}
}

@article{Geng_Liu_2017,
  author = {Geng, X.G. and Liu, H.},
  title = {{The nonlinear steepest descent method to long-time asymptotics of the coupled nonlinear Schrödinger equation}},
  journal = {Journal of Nonlinear Science},
  volume = {27},
  number = {2},
  pages = {739--763},
  year = {2017},
  publisher = {Springer},
  doi = {10.1007/s00332-017-9426-x}
}

@article{Fanengui2024Ablowitz-Ladik,
title = {{Long-time asymptotics for the defocusing Ablowitz-Ladik system with initial data in lower regularity}},
journal = {Advances in Mathematics},
volume = {450},
pages = {109769},
year = {2024},
issn = {0001-8708},
doi = {https://doi.org/10.1016/j.aim.2024.109769},
url = {https://www.sciencedirect.com/science/article/pii/S0001870824002846},
author = {Chen, M.S. and He, J.S. and Fan, E.G.},
}

@article{Fan2024Hunter-Saxton,
title = {{Long-time asymptotics of the Hunter-Saxton equation on the line}},
journal = {Journal of Differential Equations},
volume = {390},
pages = {451-493},
year = {2024},
issn = {0022-0396},
doi = {https://doi.org/10.1016/j.jde.2024.02.012},
url = {https://www.sciencedirect.com/science/article/pii/S0022039624000834},
author = {Ju, L.M. and Xu, K. and Fan, E.G.},
}

@article{Fan2024Camassa-Holm,
title = {{The Cauchy problem of the Camassa-Holm equation in a weighted Sobolev space: Long-time and Painlevé asymptotics}},
journal = {Journal of Differential Equations},
volume = {380},
pages = {24-91},
year = {2024},
issn = {0022-0396},
doi = {https://doi.org/10.1016/j.jde.2023.10.019},
url = {https://www.sciencedirect.com/science/article/pii/S0022039623006629},
author = {Xu, K. and Yang, Y.L. and Fan, E.G.},
}

@article{Fan2023mKdV,
title = {{On the Cauchy problem of defocusing mKdV equation with finite density 
initial data: Long time asymptotics in soliton-less regions}},
journal = {Journal of Differential Equations},
volume = {372},
pages = {55-122},
year = {2023},
issn = {0022-0396},
doi = {https://doi.org/10.1016/j.jde.2023.06.038},
url = {https://www.sciencedirect.com/science/article/pii/S002203962300445X},
author = {Xu, T.Y. and Zhang, Z.C. and Fan, E.G.},
}

@article{Fan2023Novikov,
title = {{Soliton resolution and large time behavior of solutions to the Cauchy problem for 
the Novikov equation with a nonzero background}},
journal = {Advances in Mathematics},
volume = {426},
pages = {109088},
year = {2023},
issn = {0001-8708},
doi = {https://doi.org/10.1016/j.aim.2023.109088},
url = {https://www.sciencedirect.com/science/article/pii/S0001870823002311},
author = {Yang, Y.L. and  Fan, E.G.},
}

@article{Fan2023nonload-mKdV,
  title={{Long time asymptotic behavior for the nonlocal mKdV equation in solitonic space--time regions}},
  author={Zhou, X. and  Fan, E.G.},
  journal={Mathematical Physics, Analysis and Geometry},
  volume={26},
  number={1},
  pages={3},
  year={2023},
  publisher={Springer}
}

@article{Fan2022DNLS,
  title={{Long-time asymptotic behavior for the derivative Schr{\"o}dinger equation with finite density type initial data}},
  author={Yang, Y.L. and Fan, E.G.},
  journal={Chinese Annals of Mathematics, Series B},
  volume={43},
  number={6},
  pages={893--948},
  year={2022},
  publisher={Springer}
}

@article{Fan2022Sasa-Satsuma,
title = {{Long time and Painlevé-type asymptotics for the Sasa-Satsuma equation in solitonic space time regions}},
journal = {Journal of Differential Equations},
volume = {329},
pages = {89-130},
year = {2022},
issn = {0022-0396},
doi = {https://doi.org/10.1016/j.jde.2022.05.006},
url = {https://www.sciencedirect.com/science/article/pii/S0022039622003059},
author = {Xun, W.K. and Fan, E.G.},
}

@article{GengLiu2024DNLS,
url = {https://doi.org/10.1515/ans-2023-0145},
title = {{Long-time asymptotic behavior for the Hermitian symmetric space derivative nonlinear Schrödinger equation}},
author = {Chen, M.M. and Geng, X.G. and Liu, H.},
pages = {819--856},
volume = {24},
number = {4},
journal = {Advanced Nonlinear Studies},
doi = {doi:10.1515/ans-2023-0145},
year = {2024},
lastchecked = {2024-12-15}
}

@article{Geng2024short-pulse,
title = {Long-time asymptotics for the coupled modified complex short-pulse equation},
journal = {Communications on Pure and Applied Analysis},
volume = {23},
number = {4},
pages = {507-545},
year = {2024},
issn = {1534-0392},
doi = {10.3934/cpaa.2024023},
url = {https://www.aimsciences.org/article/id/660537f7d8a66b482f68303b},
author = {Liu, W.H. and Geng, X.G. and Liu, H.},
}

@article{Geng2024short-pulse2,
title = {Long-time asymptotics for the coupled complex short-pulse equation with decaying initial data},
journal = {Journal of Differential Equations},
volume = {386},
pages = {113-163},
year = {2024},
issn = {0022-0396},
doi = {https://doi.org/10.1016/j.jde.2023.12.019},
url = {https://www.sciencedirect.com/science/article/pii/S0022039623007957},
author = {Geng, X.G. and Liu, W.H. and Li, R.M.},
}

@article{Geng2022CNLS,
  title={{Spectral analysis and long-time asymptotics of a coupled nonlinear Schr{\"o}dinger System}},
  author={Wang, K.D. and Geng, X.G. and Chen, M.M. and Li, R.M.},
  journal={Bulletin of the Malaysian Mathematical Sciences Society},
  volume={45},
  number={5},
  pages={2071--2106},
  year={2022},
  publisher={Springer}
}

@article{Backlundtransformations20001,
	author = {Terng, C.L. and Uhlenbeck, K.},
	title = {Bäcklund transformations and loop group actions},
	year = {2000},
	journal = {Communications on Pure and Applied Mathematics},
	volume = {53},
	number = {1},
	pages = {1-75},
	publisher={Wiley Online Library}
}

@article{lingzhang2024,
title = {{Large and infinite-order solitons of the coupled nonlinear Schrödinger equation}},
journal = {Physica D},
volume = {457},
pages = {133981},
year = {2024},
author = {Ling, L.M. and Zhang, X.E.},
publisher={Elsevier B.V.}
}

@article{chen2021soliton,
  title={{Soliton resolution for the focusing modified KdV equation}},
  author={Chen, G. and Liu, J.Q.},
  journal={Annales de l'Institut Henri Poincar{\'e} C, Analyse non lin{\'e}aire},
  volume={38},
  number={6},
  pages={2005--2071},
  year={2021},
  publisher={Elsevier}
}

@article{charlier2024soliton,
  title={{The soliton resolution conjecture for the Boussinesq equation}},
  author={Charlier, C. and Lenells, J.},
  journal={Journal de Math{\'e}matiques Pures et Appliqu{\'e}es},
  volume={191},
  pages={103621},
  year={2024},
  publisher={Elsevier}
}

@article{charlier2024boussinesq,
  title={{Boussinesq's equation for water waves: Asymptotics in Sector I}},
  author={Charlier, C. and Lenells, J.},
  journal={Advances in Nonlinear Analysis},
  volume={13},
  number={1},
  pages={20240022},
  year={2024},
  publisher={De Gruyter}
}

@article{charlier2024boussinesq-1,
  title={{Boussinesq's equation for water waves: Asymptotics in Sector V}},
  author={Charlier, C. and Lenells, J.},
  journal={SIAM Journal on Mathematical Analysis},
  volume={56},
  number={3},
  pages={4104--4142},
  year={2024},
  publisher={SIAM}
}

@article{charlier2023good,
  title={{The ``good" Boussinesq equation: long-time asymptotics}},
  author={Charlier, C. and Lenells, J. and Wang, D.S.},
  journal={Analysis \& PDE},
  volume={16},
  number={6},
  pages={1351--1388},
  year={2023},
  publisher={Mathematical Sciences Publishers}
}

@article{Bilman-2020,
    AUTHOR = {Bilman, D. and Ling, L.M. and Miller, P.D.},
     TITLE = {{Extreme superposition: rogue waves of infinite order and the
              Painlev{\'e}-III hierarchy}},
   JOURNAL = {Duke Math. J.},
    VOLUME = {169},
      YEAR = {2020},
    NUMBER = {4},
     PAGES = {671--760},
}

@book{beals1988direct,
  title={Direct and inverse scattering on the line},
  author={Beals, R. and Deift, P.A. and Tomei, C.},
  year={1988},
  publisher={American Mathematical Soc.}
}
\end{document}